\documentclass[12pt]{amsart}

\usepackage{amsmath}
\usepackage{amsfonts}
\usepackage{amssymb}
\usepackage[all]{xy}           

\usepackage{bbding}
\usepackage{txfonts}
\usepackage{amscd}
\usepackage{dsfont}
\usepackage{mathrsfs}

\usepackage{yfonts}

\usepackage[shortlabels]{enumitem}
\usepackage{ifpdf}
\ifpdf
  \usepackage[colorlinks,final,backref=page,hyperindex]{hyperref}
\else
  \usepackage[colorlinks,final,backref=page,hyperindex,hypertex]{hyperref}
\fi
\usepackage{tikz}
\usepackage[active]{srcltx}

\usepackage{graphicx}
\makeatletter

\newtheorem{thm}{Theorem}[section]

\newtheorem{cor}[thm]{Corollary}
\newtheorem{pro}[thm]{Proposition}
\theoremstyle{definition}   
\newtheorem{ex}[thm]{Example}
\newtheorem{rmk}[thm]{Remark}
\newtheorem{defi}[thm]{Definition}

\newcommand {\emptycomment}[1]{}

\newcommand{\be }{\begin{equation}}
\newcommand{\ee }{\end{equation}}

\newcommand{\iii}{{\mathbf{1}}}
\newcommand{\III}{{\mathbf{I}}}

\newcommand{\g}{\mathfrak g}

\newcommand{\huaB}{\mathcal{B}}

\newcommand{\huaL}{\mathcal{L}}

\newcommand{\huaC}{{\mathcal{C}}}

\newcommand{\huaJ}{\mathcal{J}}

\newcommand{\huaO}{{\mathcal{O}}}

\newcommand{\huaN}{\mathcal{N}}

\newcommand{\frkJ}{\mathfrak J}

\newcommand{\frkL}{\mathfrak L}

\newcommand{\frkR}{\mathfrak R}

\newcommand{\Id}{{\rm{Id}}}

\newcommand{\br}[1]{   [ \cdot,    \cdot  ]   }

\newcommand{\End}{\mathrm{End}}

\newcommand{\K}{\mathbb{K}}

\newcommand{\R}{\mathbb{R}}

\newcommand{\swe}{\textswab{e}}
\newcommand{\swg}{\textswab{g}}
\newcommand{\rmG}{{\rm G}}

\begin{document}

\title[Leibniz $2$-algebras, linear $2$-racks and the Zamolodchikov Tetrahedron equation]
{Leibniz $2$-algebras, linear $2$-racks and the Zamolodchikov Tetrahedron equation}

\author{Nanyan Xu}
\address{Department of Mathematics, Jilin University, Changchun 130012, Jilin, China}
\email{xuny20@mails.jlu.edu.cn}
\author{Yunhe Sheng}
\address{Department of Mathematics, Jilin University, Changchun 130012, Jilin, China}
\email{shengyh@jlu.edu.cn}

\vspace{-2cm}


\begin{abstract}
In this paper, first we show that a central Leibniz 2-algebra naturally gives rise to a solution of the Zamolodchikov Tetrahedron equation. Then we introduce the notion of linear 2-racks and show that a linear 2-rack also gives rise to a solution of the Zamolodchikov Tetrahedron equation. We show that a central Leibniz 2-algebra gives rise to a linear 2-rack if the underlying 2-vector space is splittable. Finally we discuss the relation between linear 2-racks and 2-racks, and show that a linear 2-rack gives rise to a 2-rack structure on the group-like category. A concrete example of strict 2-racks is constructed from an action of a strict 2-group.
\end{abstract}

\renewcommand{\thefootnote}{}

\maketitle

\tableofcontents

\allowdisplaybreaks


\vspace{-1.5cm}
\section{Introduction}

The main purpose of this paper is to construct solutions of the Zamolodchikov Tetrahedron equation using central Leibniz 2-algebras and linear 2-racks. The relations between these algebraic structures are also investigated.

\subsection{The Yang-Baxter equation and related algebraic structures}

The Yang-Baxter equation is a fundamental equation in mathematical physics, originating from the study of integrable quantum systems and statistical mechanics in the 1960s-70s. Its name honors the foundational work of C. Yang \cite{Yang} and R. Baxter \cite{Baxter}. A solution of the Yang-Baxter equation on a vector space $V$ is an  invertible  linear map $\huaB:V\otimes V\rightarrow V\otimes V$ satisfying
\begin{eqnarray}\label{YBE}
(\huaB\otimes{\Id}_V)({\Id}_V\otimes\huaB)(\huaB\otimes{\Id}_V)
=({\Id}_V\otimes\huaB)(\huaB\otimes{\Id}_V)({\Id}_V\otimes\huaB).
\end{eqnarray}

Leibniz algebras, first discovered by Bloh under the name of D-algebras \cite{Blo}, and  rediscovered by Loday in his study of algebraic K-theory \cite{Loday1}. Specially, a (right) {\bf Leibniz algebra} is a vector space $\g$ together with a bilinear operation $[\cdot,\cdot]_\g$ such that the Leibniz identity holds:
\begin{eqnarray*}
\label{Leibniz}[[x,y]_\g,z]_\g=[[x,z]_\g,y]_\g+[x,[y,z]_\g]_\g,\,\,\,\forall x,y,z\in\g.
\end{eqnarray*}
In \cite{Lebed1}, Lebed showed that a central Leibniz algebra $(\g,[\cdot,\cdot]_\g,e)$, where $e$ is a central element, gives rise to a solution of the Yang-Baxter equation via the following formula:
\begin{eqnarray}\label{cen-Leibniz-sol}
&&\huaB(x\otimes y)=y\otimes x+e\otimes[x,y]_\g,\,\,\,\forall x,y\in\g.
\end{eqnarray}

A (right) rack is a set $X$ equipped with a map $\lhd:X\times X\to X$ such that the map~$\bullet\lhd x:X\rightarrow X,~y\mapsto y\lhd x$ is bijective for all $x\in X$, and
$$(x\lhd y)\lhd z=(x\lhd z)\lhd(y\lhd z), \quad\forall x,y,z\in X. $$
Racks play very important roles in the study of set-theoretical solutions of the Yang-Baxter equation \cite{AG,DRS,LV}.
Linear racks are rack objects in the category of coalgebras, see \cite{Carter,Krahmer,Lebed3} for more details. In \cite{Lebed3},  Lebed showed that a linear rack $(V,\Delta,\varepsilon,\lhd)$, where $(V,\Delta,\varepsilon)$ is a cocommutative coalgebra,  gives rise to a solution $\huaB:V\otimes V\rightarrow V\otimes V$  of the Yang-Baxter equation:
\begin{eqnarray}\label{quantum rack to solution formula}
\huaB\underline{}(u\otimes v)=v_{(1)}\otimes (u\lhd v_{(2)}),\,\,\,\forall u\otimes v\in V\otimes V.
\end{eqnarray}
Abramov and Zappala  generalized this construction to trilinear racks (reversible TSD) \cite{Abramov}.

Lebed showed in \cite{Lebed3} that a Leibniz algebra $\g$ naturally gives rise to a linear rack structure on $\g\oplus\K$, such that the solutions of the Yang-Baxter equation given by the linear rack $\g\oplus\K$ and the central Leibniz algebra $\g\oplus\K$ are the same.

\subsection{2-vector spaces and decategorification}

A {\bf $2$-vector space} is an internal category in the category of vector spaces \cite{BC}. That is, a $2$-vector space  $V=(V_0,V_1,s,t,\iii)$ consists of
  a vector space of objects $V_0$, a vector space of morphisms $V_1$, linear maps (source and target maps) $s,t: V_1 \to V_0$,
and a linear map (the identity map) $\iii: V_0 \to V_1$ satisfying some compatibility conditions.
One can put Lie/Leibniz algebra structures on 2-vector spaces, and obtain Lie/Leibniz 2-algebras (equivalently, 2-term $L_\infty$-algebras/Leibniz$_\infty$-algebras)  \cite{BC,SL}.

2-vector spaces are considered as the categorification of vector spaces. In other words, one can obtain vector spaces by decategorification of $2$-vector spaces.
One of the most natural ways to decategorify the information in a category is to take its Grothendieck group, which can be defined for an abelian category, an additive category, or an essentially small category with a bifunctor.
Given a bifunctor $\star:\huaC\times\huaC\to\huaC$ on an essentially small category $\huaC$, one can obtain the Grothendieck group of $\huaC$ with respect to the bifunctor $\star$ as the quotient group $$K^\star_0(\huaC)=\mathcal{F}(\huaO_\huaC)/\huaN_\star(\huaC),$$
where $\huaO_\huaC={\rm ob}(\huaC)/\cong$ is the set of all isomorphism classes of objects in $\huaC$, $\mathcal{F}(\huaO_\huaC)$ is the free abelian group generated by the set $\huaO_\huaC$, and $\huaN_\star(\huaC)$ is the normal subgroup generated by elements $\{\overline{a\star b}-\overline{a}-\overline{b}~|~a,b\in{\rm ob}(\huaC)\}$.

There is an alternative viewpoint regarding the Grothendieck group of $\huaC$ with respect to an associative
and commutative bifunctor $\star$.
Given an associative and commutative bifunctor $\star$,
$(\huaO_\huaC,\star)$ is an abelian semigroup,
where $\overline{a}{\star}\overline{b}=\overline{a\star b}$ for any $\overline{a},\overline{b}\in\huaO_\huaC$.
Then the Grothendieck group $K^\star_0(\huaC)$ with respect to the  bifunctor $\star$ is actually the group completion of the abelian semigroup $(\huaO_\huaC,\star)$. More details about Grothendieck groups are available in \cite{LM}.

For a $2$-vector space $V=(V_0,V_1,s,t,\iii)$,  the addition bifunctor $+$ is associative and commutative.
Moreover, the set of all isomorphism classes of objects $\overline{V_0}=V_0/\cong$ endowed with $+$ is not merely an abelian semigroup, but indeed an abelian group, which has $\overline{0}$ as the unit and $\overline{-x}$ as the inverse for each $\overline{x}$.
Therefore, the group completion of the abelian semigroup $(\overline{V_0},+)$ is
the abelian group $(\overline{V_0},+)$ itself, which implies that the Grothendieck group of a $2$-vector space $V$ is the abelian group $(\overline{V_0},+)$.
Furthermore, by inheriting the scalar multiplication structure from $V_0$, $\overline{V_0}$ becomes a vector space.
In this case, the vector space $\overline{V_0}$ is the decategorification of the $2$-vector space $V=(V_0,V_1,s,t,\iii)$.

\subsection{The Zamolodchikov Tetrahedron equation}

As an extension of the Yang-Baxter equation, the Zamolodchikov Tetrahedron equation is
\begin{eqnarray}\label{ZT-matrix-eq}
S_{123}\circ S_{145}\circ S_{246}\circ S_{356}&=&S_{356}\circ S_{246}\circ S_{145}\circ S_{123},
\end{eqnarray}
where $V$ is a vector space, $S\in\End(V^{\otimes 3})$ and $S_{ijk}\in\End(V^{\otimes 6})$ for $1\leqslant i<j<k\leqslant 6$ are maps that act as $S$ on the $ij~k$-th tensor components of the vector space $V^{\otimes 6}$, while acting as the identity on the remaining components. This equation was first proposed by Zamolodchikov \cite{Zam1,Zam2}, who generalized the ``triangle" structure of the Yang-Baxter equation to a ``tetrahedron" structure by introducing the braiding relationship in three-dimensional spaces and proposed the algebraic form of the equation together with a conjectured solution. 
 Later, Baxter rigorously proved this solution in \cite{Bax} and derived the partition function of the Zamolodchikov model for an $n\times\infty\times\infty$ lattice in \cite{Bax2}. 
The Zamolodchikov Tetrahedron equation is closely related to the integrable lattice models of statistical mechanics and quantum field theory, promoting significant progress in these research fields \cite{Kon,KOS,KS,MS}.
Recently, the relationship between the Zamolodchikov Tetrahedron equation and many algebraic structures has been explored \cite{BV,GSZ,IKT}, and the theory of set-theoretic solutions has also been greatly developed \cite{CK,IK1,IK2}.

In \cite{KV}, Kapranov and Voevodsky  
 interpreted the Zamolodchikov Tetrahedron equation as the commutativity of a three-dimensional diagram with the shape of a permutohedron in a monoidal $2$-category.
They showed that the usual vertex formulation \eqref{ZT-matrix-eq} of the Zamolodchikov Tetrahedron equation can be obtained by choosing a special example of monoidal $2$-categories \cite[Example 6.8]{KV}.
Moreover, it is obvious that the $2$-category {\bf$2$Vect} with $2$-vector spaces as objects, linear functors as $1$-morphisms, and linear natural transformations as $2$-morphisms,
  is a monoidal $2$-category, and the associated Zamolodchikov Tetrahedron equation was described in \cite{BC}. More precisely,  the Zamolodchikov Tetrahedron equation on a $2$-vector space $V$ for a linear invertible functor $B:V\otimes V\to V\otimes V$ and a linear natural isomorphism $Y:(B\otimes\Id)(\Id\otimes B)(B\otimes\Id)\Rightarrow(\Id\otimes B)(B\otimes\Id)(\Id\otimes B)$  reads as follows:
 \begin{eqnarray*}
&\big[(\Id\otimes Y)\ast\Id_{(B\otimes\Id\otimes\Id)(\Id\otimes B\otimes\Id)(\Id\otimes\Id\otimes B)}\big]
\big[\Id_{(\Id\otimes B\otimes\Id)(\Id\otimes\Id\otimes B)}\ast(Y\otimes\Id)\ast\Id_{\Id\otimes\Id\otimes B}\big]&\\
&\big[\Id_{(\Id\otimes B\otimes\Id)(B\otimes\Id\otimes\Id)}\ast(\Id\otimes Y)\ast\Id_{B\otimes\Id\otimes\Id}\big]
\big[(Y\otimes\Id)\ast\Id_{(\Id\otimes\Id\otimes B)(\Id\otimes B\otimes\Id)(B\otimes\Id\otimes\Id)}\big]&\\
&=&\\
&\big[\Id_{(\Id\otimes\Id\otimes B)(\Id\otimes B\otimes\Id)(B\otimes\Id\otimes\Id)}\ast(\Id\otimes Y)\big]
\big[\Id_{\Id\otimes\Id\otimes B}\ast(Y\otimes\Id)\ast\Id_{(\Id\otimes\Id\otimes B)(\Id\otimes B\otimes\Id)}\big]&\\
&\big[\Id_{B\otimes\Id\otimes\Id}\ast(\Id\otimes Y)\ast\Id_{(B\otimes\Id\otimes\Id)(\Id\otimes B\otimes\Id)}\big]
\big[\Id_{(B\otimes\Id\otimes\Id)(\Id\otimes B\otimes\Id)(\Id\otimes\Id\otimes B)}\ast(Y\otimes\Id)\big],&
\end{eqnarray*}
where the both sides of the equality are understood as  the vertical composition of the linear natural transformations, and
 $\ast$ represents the horizontal composition of natural transformations, $\Id_{\Id\otimes\Id\otimes B}$ is the identity natural isomorphism of the functor $\Id\otimes\Id\otimes B$ and the same for the others. See the appendix for a diagram description of this equation.

Let the vector space $\overline{V_0}$ be the decategorification of the $2$-vector space $V=(V_0,V_1,s,t,\iii)$, and $(B,Y)$ a solution of the Zamolodchikov Tetrahedron equation on $V$.
Define a linear map $\overline{B}:\overline{V_0}\otimes\overline{V_0}\to \overline{V_0}\otimes\overline{V_0}$ as follows:
\begin{eqnarray}\label{deficat-of-ZTE}
\overline{B}(\overline{x}\otimes\overline{y})=\overline{B(x\otimes y)},
\quad\forall\overline{x}\otimes\overline{y}\in\overline{V_0}\otimes\overline{V_0}.
\end{eqnarray}
Here we identify the vector space $\overline{V_0}\otimes\overline{V_0}$ with $\overline{V_0\otimes V_0}$.
Then $\overline{B}$ satisfies the Yang-Baxter equation:
$$(\overline{B}\otimes\Id)(\Id\otimes\overline{B})(\overline{B}\otimes\Id)
=(\Id\otimes\overline{B})(\overline{B}\otimes\Id)(\Id\otimes\overline{B}),$$
since $Y$ is a linear natural isomorphism. Thus, the Yang-Baxter equation serves as the decategorification of the Zamolodchikov Tetrahedron equation.

\subsection{Main results and outline of the paper}

In this paper,  we show that a Leibniz 2-algebra (which is the categorification of a Leibniz algebra, and equivalent to a 2-term Leibniz$_\infty$-algebra \cite{ammardefiLeibnizalgebra}) with a central object naturally gives rise to a solution of the Zamolodchikov Tetrahedron equation. On the one hand, this generalizes Baez and Crans' construction in \cite{BC}; on the other hand, this can be viewed as the categorification of Lebed's construction of solutions of the Yang-Baxter equation using central Leibniz algebras given in \cite{Lebed1}.  We introduce the notion of linear 2-racks and show that a linear 2-rack also gives rise to a solution of the Zamolodchikov Tetrahedron equation, which can be viewed as the categorification of Lebed's construction of solutions of the Yang-Baxter equation using linear racks given in \cite{Lebed3}. We  show that a central Leibniz 2-algebra gives rise to a linear 2-rack if the underlying 2-vector space is splittable. Finally we  show that a linear 2-rack gives rise to a 2-rack structure on the group-like category, and construct a concrete example of strict 2-racks from an action of a strict 2-group. The results in this paper can be summarized by the following diagram, where the dotted arrows are what we obtained:

{\footnotesize\begin{equation}\label{diagram:main}
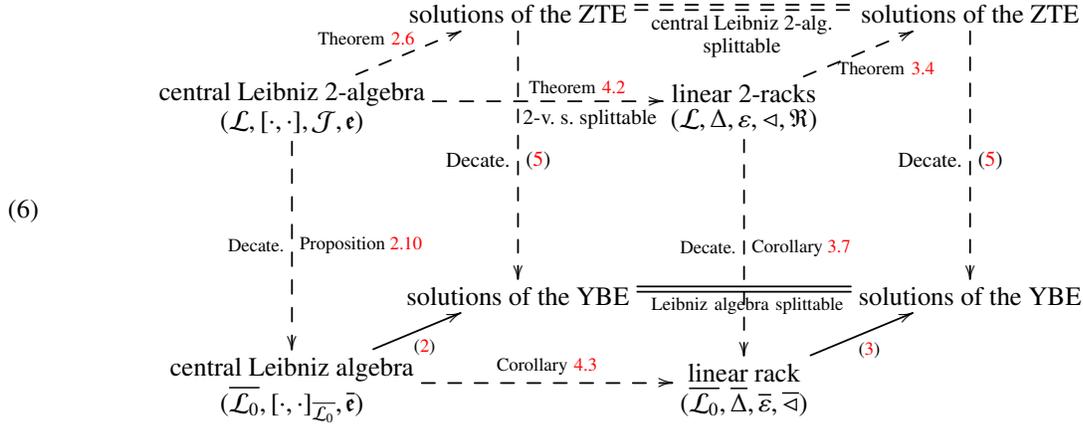
\begin{array}{l}
\xymatrix@!0@C=9.5ex@R=7.9ex{
&&
\txt{\rm solutions of the ZTE}
\ar@{-->}[ddd]_-{\txt{\\\tiny~{\rm Decate.}}}^-{\txt{\\\tiny\eqref{deficat-of-ZTE}}}
\ar@{==}[rrr]_-{\tiny\txt{\rm central Leibniz $2$-alg. \\splittable}}
&&&
&\txt{\rm solutions of the ZTE}
\ar@{-->}[ddd]_-{\txt{\\\tiny~{\rm Decate.}}}^-{\txt{\\\tiny\eqref{deficat-of-ZTE}}}\\
\txt{\rm central Leibniz $2$-algebra\\$(\huaL,[\cdot,\cdot],\huaJ,\swe)$}
\ar@{-->}[urr]^-{\rm Theorem~\ref{cen-Lei-to-sol}\quad}
\ar@{-->}[rrrr]^-{\qquad\,\,\rm Theorem~\ref{cen-lei-2-alg-to-lin-2-rack}}_-{\tiny\txt{\qquad\qquad 2-v.~s.~ splittable}}
\ar@{-->}[ddd]_-{{\rm Decate.}}^-{\text{{\rm Proposition~\ref{decate-Leibniz-2-alg}}}}
&&&
&\txt{\rm linear $2$-racks\\$(\huaL,\Delta,\varepsilon,\lhd,\frkR)$}
\ar@{-->}[urr]_-{~\rm Theorem~\ref{lin-2-rack-to-sol}}
\ar@{-->}[ddd]_-{{\rm Decate.}}^-{\text{{\rm Corollary~\ref{decate-lin-2-rack}}}}
&&\\
&&&&&&&\\
&&
\txt{\rm solutions of the YBE}
\ar@{=}[rrrr]_{\quad\rm ~Leibniz~algebra~splittable}
&&&
&\txt{\rm solutions of the YBE}\\
\txt{\rm central Leibniz algebra\\ $(\overline{\huaL_0},[\cdot,\cdot]_{\overline{\huaL_0}},\overline{\swe})$}
\ar[urr]_-{\eqref{cen-Leibniz-sol}}
\ar@{-->}[rrrr]^-{\rm Corollary~\ref{cor:cl-lr}}
&&&
&\txt{\rm linear rack\\ $(\overline{\huaL_0},\overline{\Delta},\overline{\varepsilon},{\overline{\lhd}})$}
\ar[urr]_-{\eqref{quantum rack to solution formula}}
&&
}
\end{array}
\end{equation}}

The paper is organized as follows. In Section \ref{sec:cL}, we construct solutions of the Zamolodchikov Tetrahedron equation using central Leibniz 2-algebras. In Section \ref{sec:l2rack}, we construct solutions of the Zamolodchikov Tetrahedron equation using linear 2-racks. In Section \ref{sec:passage}, we show that certain central Leibniz 2-algebras give rise to linear 2-racks. In Section \ref{sec:r}, we show that on the group-like category of a linear 2-rack, there is naturally a 2-rack structure.

\section{Central Leibniz $2$-algebras and the Zamolodchikov Tetrahedron equation}\label{sec:cL}

In this section, we show that central Leibniz 2-algebras give solutions of the Zamolodchikov Tetrahedron equation. 
 The notion of (left) Leibniz $2$-algebras was introduced in \cite{SL}, which is equivalent to 2-term Leibniz$_\infty$-algebras \cite{ammardefiLeibnizalgebra}. In this paper, we use right Leibniz $2$-algebras.

\begin{defi}
A (right) {\bf Leibniz $2$-algebra} consists of
\begin{itemize}
  \item a 2-vector space $\huaL=(\huaL_0,\huaL_1,s,t,\iii)$, i.e. an internal category in the category of vector spaces;
  \item a linear functor (bracket) $[\cdot,\cdot]:\huaL\otimes\huaL\to\huaL$;
  \item a linear natural isomorphism $\huaJ_{x\otimes y\otimes z}:[[x,y],z]\rightarrow[[x,z],y]+[x,[y,z]]$, satisfies the following Jacobiator identity:
  \begin{eqnarray}\label{Jacob}
  &&(\huaJ_{[x,w]\otimes y\otimes z}+\huaJ_{x\otimes[y,w]\otimes z}+\huaJ_{x\otimes y\otimes[z,w]})\circ([\huaJ_{x\otimes y\otimes w},\iii_z]+\iii_{[[x,y],[z,w]]})\circ\huaJ_{[x,y]\otimes z\otimes w}\\
  &=&([\huaJ_{x\otimes z\otimes w},\iii_y]+\iii_{[[x,z],[y,w]]+[[x,w],[y,z]]}+[\iii_x, \huaJ_{y\otimes z\otimes w}])\circ(\huaJ_{[x,z]\otimes y\otimes w}+\huaJ_{x\otimes[y,z]\otimes w})\circ[\huaJ_{x\otimes y\otimes z},\iii_w],\nonumber
  \end{eqnarray}
\end{itemize}
where $x,y,z,w\in\huaL_0$ and the above identity can be described as the following commutative diagram:
{\footnotesize\begin{equation*}\label{Leib-2-identity}
\xymatrix@R=1.5pc@C=0.3pc{
 &
 \text{$[[[x,y],z],w]$}
 \ar[dl]_-{\text{$\huaJ_{[x,y]\otimes z\otimes w}$}}
 \ar[dr]^-{\text{$[\huaJ_{x\otimes y\otimes z},\iii_w]$}}
 &\\
 \text{$[[[x,y],w],z]+[[x,y],[z,w]]$}
 \ar[d]_-{\text{$[\huaJ_{x\otimes y\otimes w},\iii_z]+\iii_{[[x,y],[z,w]]}$}}
 &
 &\text{$[[[x,z],y],w]+[[x,[y,z]],w]$}
 \ar[d]^-{\text{$\huaJ_{[x,z]\otimes y\otimes w}+\huaJ_{x\otimes [y,z]\otimes w}$}}\\
 \txt{$[[[x,w],y],z]+[[x,[y,w]],z]$\\$+[[x,y],[z,w]]$}
 \ar[dr]_-{\text{$\huaJ_{[x,w]\otimes y\otimes z}+\huaJ_{x\otimes [y,w]\otimes z}+\huaJ_{x\otimes y\otimes[z,w]}\quad\qquad\qquad$}}
 &
 &\txt{$[[[x,z],w],y]+[[x,z],[y,w]]$\\$+[[x,w],[y,z]]+[x,[[y,z],w]]$}
 \ar[dl]^-{\text{$\qquad\qquad\qquad\qquad\qquad[\huaJ_{x\otimes z\otimes w},
 \iii_y]+\iii_{[[x,z],[y,w]]+[[x,w],[y,z]]}+[\iii_x, \huaJ_{y\otimes z\otimes w}]$}}\\
 &
 \txt{$[[[x,w],z],y]+[[x,w],[y,z]]$\\$+[[x,z],[y,w]]+[x,[[y,w],z]]$\\$+[[x,[z,w]],y]+[x,[y,[z,w]]]$}
 &
 }
 \end{equation*}
 }
\end{defi}

\begin{defi}\cite{BC}
Let $(\huaL,[\cdot,\cdot],\huaJ)$ and $(\huaL',[\cdot,\cdot]',\huaJ')$ be two Leibniz $2$-algebras. A {\bf homomorphism} $(F_0,F_1,F_2):\huaL\to\huaL'$ between two Leibniz $2$-algebras consists of:
\begin{enumerate}
  \item[(i)] a linear functor $F=(F_0,F_1):\huaL\to\huaL'$   between the underlying $2$-vector spaces;
  \item[(ii)] a linear natural isomorphism $F_2:[\cdot,\cdot]'\Rightarrow(F\otimes F)\to F\circ [\cdot,\cdot]$ such that for any $x,y,z\in\huaL_0$ the following diagram with respect to the composition of morphisms is commutative:
     {\footnotesize \begin{displaymath}
\xymatrix@C=3.7ex@R=3.5ex{
  \txt{$[[F(x),F(y)]',F(z)]'$}
  \ar[rr]^-{\huaJ'_{F(x)\otimes F(y)\otimes F(z)}}
  \ar[d]_-{[F_2(x,y),\iii_{F(z)}]'}
  &
  &\txt{$[[F(x),F(z)]',F(y)]'+[F(x),[F(y),F(z)]']'$}
  \ar[d]^-{[F_2(x,z),\iii_{F(y)}]'+[\iii_{F(x)},F_2(y,z)]'}\\
\txt{$[F[x,y],F(z)]'$}
\ar[d]_-{F_2([x,y],z)}
  &\rotatebox{165}{{\txt{\Huge $\circlearrowright$}}}
  &\txt{$[F[x,z],F(y)]'+[F(x),F[y,z]]'$}
  \ar[d]^-{F_2([x,z],y)+F_2(x,[y,z])}\\
  \txt{$F[[x,y],z]$}
  \ar[rr]_-{F(\huaJ_{x\otimes y\otimes z})}
  &
  &\txt{$F([[x,z],y]+[x,[y,z]]).$}
}
\end{displaymath}}
\end{enumerate}
\end{defi}

\begin{defi}
Let $(\huaL,[\cdot,\cdot],\huaJ)$ be a Leibniz $2$-algebra. If an object $\swe\in\huaL_0$ satisfies
\begin{eqnarray}
&&{[}\iii_\swe,f]=0=[f,\iii_\swe],\qquad\forall f\in\huaL_1,\label{central-element}
\end{eqnarray}
then $\swe$ is called a {\bf central object} and $(\huaL,[\cdot,\cdot],\huaJ,\swe)$ is called a ${\bf central~Leibniz~2}$-${\bf algebra}$.
\end{defi}

In fact, since the linear functor $[\cdot,\cdot]$ preserves the source map $s$, for any $x\in\huaL_0$, we have
\begin{eqnarray}
&&[\swe,x]=[s(\iii_\swe),s(\iii_x)]=s[\iii_\swe,\iii_x]=0=s[\iii_x,\iii_\swe]=[s(\iii_x),s(\iii_\swe)]=[x,\swe],\label{central-s}
\end{eqnarray}

\begin{ex}\label{cen-Lei-to-cen-Lei-2-alg}
Let $(\g,[\cdot,\cdot]_\g,e)$ be a right central Leibniz algebra, i.e., $[x,e]_\g=0=[e,x]_\g$ holds for any $x\in\g$.
If a skew-symmetric bilinear form $\omega\in\wedge^2\g^*$ satisfies the following invariant condition:
\begin{eqnarray*}
\omega([x,y]_\g,z)=\omega(y,[x,z]_\g+[z,x]_\g),\quad\forall x,y,z\in\g,
\end{eqnarray*}
then $\swg=(\g,\g\oplus\R,s,t,\iii)$ is a $2$-vector space and $(\swg,[\cdot,\cdot],\huaJ)$ is a Leibniz $2$-algebra showed in \cite{TS}, where for any $x,y,z\in\g$ and $(x,a),(y,b)\in\g\oplus\R$,
\begin{eqnarray*}
&&s(x,a)=x,\qquad\qquad\qquad\,\,\,\,
t(x,a)=x,\qquad\qquad\qquad\quad\,
\iii_x=(x,0),\\
&&\,\,\,{[}x,y]=[x,y]_\g,\qquad
{[}(x,a),(y,b)]=([x,y]_\g,0),\qquad
\huaJ_{x\otimes y\otimes z}=([[x,y]_\g,z]_\g, \omega([x,y]_\g,z)).
\end{eqnarray*}
It is easy to see that $e\in\g$ is a central object of this Leibniz $2$-algebra.
\end{ex}

\begin{ex}\label{cen-ext-Lei-2-alg}
Let $(\huaL,[\cdot,\cdot],\huaJ)$ be a Leibniz $2$-algebra.
Consider the $2$-vector space $\K\oplus\huaL$, where $\K=(\K,\K,\Id_\K,\Id_\K,\Id_\K)$ is a $2$-vector space.
Define $[\cdot,\cdot]_\oplus:(\K\oplus\huaL)\otimes(\K\oplus\huaL)\to(\K\oplus\huaL)$ as follows:
\begin{eqnarray*}
&&[(a,x),(b,y)]_\oplus=(0,[x,y]),\qquad
{[}(a,f),(b,g){]}_\oplus=(0,[f,g]),
\end{eqnarray*}
where $(a,x),(b,y)\in\K\oplus\huaL_0$ and $(a,f),(b,g)\in\K\oplus\huaL_1$.
Define a linear natural isomorphism $$\frkJ_{(a,x)\otimes(b,y)\otimes(c,z)}:[[(a,x),(b,y)]_\oplus,(c,z)]_\oplus\to[[(a,x),(c,z)]_\oplus,(b,y)]_\oplus+[(a,x),[(b,y),(c,z)]_\oplus]_\oplus$$
by $\frkJ_{(a,x)\otimes(b,y)\otimes(c,z)}=(0,\huaJ_{x\otimes y\otimes z}),$
where $(a,x),(b,y),(c,z)\in\K\oplus\huaL_0$.
Then $(\K\oplus\huaL,[\cdot,\cdot]_\oplus,\frkJ)$ is a Leibniz $2$-algebra
and $(1,0)$ is a central object of $\K\oplus\huaL$.
\end{ex}

The following theorem shows that
using central Leibniz $2$-algebras, one can obtain solutions of the Zamolodchikov Tetrahedron equation.
\begin{thm}\label{cen-Lei-to-sol}
Let $(\huaL,[\cdot,\cdot],\huaJ,\swe)$ be a central Leibniz $2$-algebra.
Define $B:\huaL\otimes\huaL\to\huaL\otimes\huaL$ by
\begin{eqnarray*}
B(x\otimes y)&=&y\otimes x+\swe\otimes[x,y],\qquad\quad\,\,\,\forall x,y\in\huaL_0,\\
B(f\otimes g)&=&g\otimes f+\iii_\swe\otimes[f,g],\qquad\,\,\,\,\,\forall f,g\in\huaL_1,
\end{eqnarray*}
and define $Y:(B\otimes\Id)(\Id\otimes B)(B\otimes\Id)\Rightarrow(\Id\otimes B)(B\otimes\Id)(\Id\otimes B)$ by
$$Y_{x\otimes y\otimes z}=\iii_{z\otimes y\otimes x+\swe\otimes[y,z]\otimes x+\swe\otimes y\otimes[x,z]+z\otimes \swe\otimes[x,y]}+\iii_\swe\otimes\iii_\swe\otimes\huaJ_{x\otimes y\otimes z},\quad\forall x,y,z\in\huaL_0.$$
Then $(B,Y)$ is a solution of the Zamolodchikov Tetrahedron equation.
\end{thm}

\begin{proof}
Firstly, we show that $B$ is a linear invertible functor.
For any $f\otimes g\in\huaL_1\otimes\huaL_1$, since
the linear functor $[\cdot,\cdot]$ preserves the source map $s$ and the target map $t$, we have
\begin{eqnarray*}
\big(B\circ(s\otimes s)\big)(f\otimes g)
&=&B(s(f)\otimes s(g))
=s(g)\otimes s(f)+\swe\otimes[s(f),s(g)]\\
&=&s(g)\otimes s(f)+s(\iii_\swe)\otimes s[f,g]
=(s\otimes s)(g\otimes f+\iii_\swe\otimes[f,g])\\
&=&\big((s\otimes s)\circ B\big)(f\otimes g),\\
\big(B\circ(t\otimes t)\big)(f\otimes g)
&=&B(t(f)\otimes t(g))
=t(g)\otimes t(f)+\swe\otimes[t(f),t(g)]\\
&=&t(g)\otimes t(f)+t(\iii_\swe)\otimes t[f,g]
=(t\otimes t)(g\otimes f+\iii_\swe\otimes[f,g])\\
&=&\big((t\otimes t)\circ B\big)(f\otimes g),
\end{eqnarray*}
which implies that $B$ preserves the source and target maps.
For any $x\otimes y\in\huaL_0\otimes\huaL_0$, since
the linear functor $[\cdot,\cdot]$ preserves the identity map, we have
\begin{eqnarray*}
\big(B\circ(\iii\otimes\iii)\big)(x\otimes y)
&=&B(\iii_x\otimes \iii_y)
=\iii_y\otimes \iii_x+\iii_\swe\otimes[\iii_x,\iii_y]\\
&=&\iii_y\otimes \iii_x+\iii_\swe\otimes \iii_{[x,y]}
=(\iii\otimes\iii)(y\otimes x+\swe\otimes[x,y])\\
&=&\big((\iii\otimes\iii)\circ B\big)(x\otimes y),
\end{eqnarray*}
which implies that $B$ preserves the identity map.
Given $f\otimes g,f'\otimes g'\in\huaL_1\otimes\huaL_1$ satisfying $(t\otimes t)(f\otimes g)=(s\otimes s)(f'\otimes g')$, since
the linear functor $[\cdot,\cdot]$ preserves the composition, we have
\begin{eqnarray*}
B(f'\otimes g')\circ B(f\otimes g)
&=&(g'\otimes f'+\iii_\swe\otimes[f',g'])\circ(g\otimes f+\iii_\swe\otimes[f,g])\\
&=&(g'\otimes f')\circ(g\otimes f)+(\iii_\swe\otimes[f',g'])\circ(\iii_\swe\otimes[f,g])\\
&=&g'g\otimes f'f+\iii_\swe\otimes([f',g']\circ[f,g])\\
&=&g'g\otimes f'f+\iii_\swe\otimes[f'f,g'g]\\
&=&B(f'f\otimes g'g)
=B\big((f'\otimes g')\circ(f\otimes g)\big),
\end{eqnarray*}
which implies that $B$ preserves the composition. Then we obtain that $B$ is a linear functor.
Define $\widetilde{B}:\huaL\otimes\huaL\to\huaL\otimes\huaL$ as follows:
\begin{eqnarray*}
\widetilde{B}(x\otimes y)&=&y\otimes x-[y,x]\otimes\swe,\,\qquad\,\,\,\,\forall x,y\in\huaL_0,\\
\widetilde{B}(f\otimes g)&=&g\otimes f-[g,f]\otimes\iii_\swe,\qquad\forall f,g\in\huaL_1.
\end{eqnarray*}
Similarly, $\widetilde{B}$ is also a linear functor.
For any $x\otimes y\in\huaL_0$ and $f\otimes g\in\huaL_1$, we have
\begin{eqnarray*}
(\widetilde{B}\circ B)(x\otimes y)
&=&\widetilde{B}(y\otimes x+\swe\otimes[x,y])\\
&=&x\otimes y-[x,y]\otimes\swe+[x,y]\otimes\swe-[[x,y],\swe]\otimes\swe\\
&\overset{\eqref{central-s}}=&x\otimes y,\\
(\widetilde{B}\circ B)(f\otimes g)
&=&\widetilde{B}(g\otimes f+\iii_\swe\otimes[f,g])\\
&=&f\otimes g-[f,g]\otimes\iii_\swe+[f,g]\otimes\iii_\swe-[[f,g],\iii_\swe]\otimes\iii_\swe\\
&\overset{\eqref{central-element}}=&f\otimes g,\\
(B\circ\widetilde{B})(x\otimes y)
&=&B(y\otimes x-[y,x]\otimes\swe)\\
&=&x\otimes y+\swe\otimes[y,x]-\swe\otimes[y,x]-\swe\otimes[[y,x],\swe]\\
&\overset{\eqref{central-s}}=&x\otimes y,\\
(B\circ\widetilde{B})(f\otimes g)
&=&B(g\otimes f-[g,f]\otimes\iii_\swe)\\
&=&f\otimes g+\iii_\swe\otimes[g,f]-\iii_\swe\otimes[g,f]-\iii_\swe\otimes[[g,f],\iii_\swe]\\
&\overset{\eqref{central-element}}=&f\otimes g,
\end{eqnarray*}
which means that the linear functor $B$ is invertible.

Secondly, we show that $Y$ is a linear natural isomorphism.
For any $x,y,z\in\huaL_0$, we have
\begin{eqnarray*}
(B\otimes \Id)(\Id\otimes B)(B\otimes \Id)(x\otimes y\otimes z)
&=&z\otimes y\otimes x+\swe\otimes[y,z]\otimes x+\swe\otimes y\otimes[x,z]+z\otimes \swe\otimes[x,y]\\
&&+\swe\otimes\swe\otimes[[x,y],z],\\
(\Id\otimes B)(B\otimes \Id)(\Id\otimes B)(x\otimes y\otimes z)
&=&z\otimes y\otimes x+\swe\otimes[y,z]\otimes x+\swe\otimes y\otimes[x,z]+z\otimes \swe\otimes[x,y]\\
&&+\swe\otimes\swe\otimes[[x,z],y]+\swe\otimes\swe\otimes[x,[y,z]],
\end{eqnarray*}
which implies that
\begin{eqnarray*}
s(Y_{x\otimes y\otimes z})&=&(B\otimes \Id)(\Id\otimes B)(B\otimes \Id)(x\otimes y\otimes z),\\
t(Y_{x\otimes y\otimes z})&=&(\Id\otimes B)(B\otimes \Id)(B\otimes \Id)(x\otimes y\otimes z),
\end{eqnarray*}
that is, $Y$ is compatible with the source and target maps.
By direct calculation, for any $f\otimes g\otimes h:x\otimes y\otimes z\to x'\otimes y'\otimes z'$, we have the following commutative diagram:
\begin{displaymath}
\xymatrix@C=1.2ex@R=0.5ex{
  \txt{$(B\otimes \Id)(\Id\otimes B)(B\otimes \Id)
  (x\otimes y\otimes z)$}
  \ar[rr]^-{Y_{x\otimes y\otimes z}}
  \ar[dd]_-{(B\otimes \Id)(\Id\otimes B)(B\otimes \Id)
  (f\otimes g\otimes h)}
  &
  &\txt{$(\Id\otimes B)(B\otimes \Id)(\Id\otimes B)
  (x\otimes y\otimes z)$}
  \ar[dd]^-{(\Id\otimes B)(B\otimes \Id)(\Id\otimes B)
  (f\otimes g\otimes h)}\\
  &\rotatebox{165}{{\txt{\Huge $\circlearrowright$}}}&\\
  \txt{$(B\otimes \Id)(\Id\otimes B)(B\otimes \Id)
  (x'\otimes y'\otimes z')$}
  \ar[rr]_-{Y_{x'\otimes y'\otimes z'}}
  &
  &\txt{$(\Id\otimes B)(B\otimes \Id)(\Id\otimes B)
  (x'\otimes y'\otimes z')$}
}
\end{displaymath}
which implies that $Y$ is a linear natural isomorphism.

Finally, for any $x,y,z,w\in\huaL_0$, we have
\begin{eqnarray*}
&&(B\otimes\Id\otimes\Id)(\Id\otimes B\otimes\Id)(B\otimes\Id\otimes\Id)
(\Id\otimes\Id\otimes B)(\Id\otimes B\otimes\Id)(B\otimes\Id\otimes\Id)
(x\otimes y\otimes z\otimes w)\\
&=&w\otimes z\otimes y\otimes x+\swe\otimes[z,w]\otimes y\otimes x
+\swe\otimes z\otimes[y,w]\otimes x+w\otimes\swe\otimes[y,z]\otimes x
+\swe\otimes\swe\otimes[[y,z],w]\otimes x\\
&&+\swe\otimes z\otimes y\otimes[x,w]+\swe\otimes\swe\otimes[y,z]\otimes[x,w]
+w\otimes\swe\otimes y\otimes[x,z]+\swe\otimes\swe\otimes[y,w]\otimes[x,z]\\
&&+\swe\otimes\swe\otimes y\otimes[[x,z],w]+w\otimes z\otimes\swe\otimes[x,y]
+\swe\otimes[z,w]\otimes\swe\otimes[x,y]+\swe\otimes z\otimes\swe\otimes[[x,y],w]\\
&&+w\otimes\swe\otimes\swe\otimes[[x,y],z]+\swe\otimes\swe\otimes\swe\otimes[[[x,y],z],w],
\end{eqnarray*}
which can be illustrated as the  diagram:  
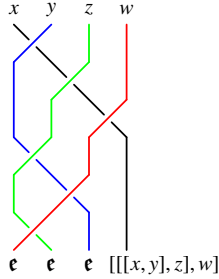
\[
\xy
   (0,0)*{}="00";
   (5,0)*{}="10";
   (10,0)*{}="20";
   (15,0)*{}="30";
   (0,5)*{}="01";
   (5,5)*{}="11";
   (10,5)*{}="21";
   (15,5)*{}="31";
   (0,10)*{}="02";
   (5,10)*{}="12";
   (10,10)*{}="22";
   (15,10)*{}="32";
   (0,15)*{}="03";
   (5,15)*{}="13";
   (10,15)*{}="23";
   (15,15)*{}="33";
   (0,-5)*{}="0-1";
   (5,-5)*{}="1-1";
   (10,-5)*{}="2-1";
   (15,-5)*{}="3-1";
   (0,-10)*{}="0-2";
   (5,-10)*{}="1-2";
   (10,-10)*{}="2-2";
   (15,-10)*{}="3-2";
   (0,-15)*{}="0-3";
   (5,-15)*{}="1-3";
   (10,-15)*{}="2-3";
   (15,-15)*{}="3-3";
   (2,13)*{}="213";
   (3,12)*{}="312";
   (7,8)*{}="78";
   (8,7)*{}="87";
   (12,3)*{}="123";
   (13,2)*{}="132";
   (2,-2)*{}="c2-2";
   (3,-3)*{}="c3-3";
   (7,-7)*{}="7-7";
   (8,-8)*{}="8-8";
   (2,-12)*{}="2-12";
   (3,-13)*{}="3-13";
   "0-3";"2-1" **[red]@{-};
   "2-1";"20" **[red]@{-};
   "20";"31" **[red]@{-};
   "31";"33" **[red]@{-};
   "1-3";"3-13" **[green]@{-};
   "2-12";"0-2" **[green]@{-};
   "0-2";"0-1" **[green]@{-};
   "0-1";"10" **[green]@{-};
   "10";"11" **[green]@{-};
   "11";"22" **[green]@{-};
   "22";"23" **[green]@{-};
   "2-3";"2-2" **[blue]@{-};
   "2-2";"8-8" **[blue]@{-};
   "7-7";"c3-3" **[blue]@{-};
   "c2-2";"00" **[blue]@{-};
   "00";"02" **[blue]@{-};
   "02";"13" **[blue]@{-};
   "3-3";"30" **@{-};
   "30";"132" **@{-};
   "123";"87" **@{-};
   "78";"312" **@{-};
   "213";"03" **@{-};
   (0,17)*{\txt{\tiny $x$}}="x";
   (5,17)*{\txt{\tiny $y$}}="y";
   (10,17)*{\txt{\tiny $z$}}="z";
   (15,17)*{\txt{\tiny $w$}}="w";
   (0,-17)*{\txt{\tiny $\swe$}}="d1";
   (5,-17)*{\txt{\tiny $\swe$}}="d2";
   (10,-17)*{\txt{\tiny $\swe$}}="d3";
   (20,-17)*{\txt{\tiny $[[[x,y],z],w]$}}="d4";
\endxy
\]
Here in the diagram, only terms of the form ``$\swe\otimes\swe\otimes\swe\otimes-$'' are shown, since the omitted terms were found to have no effect in subsequent calculations.
Similarly, both sides of the Zamolodchikov Tetrahedron equation can be shown as following diagrams:
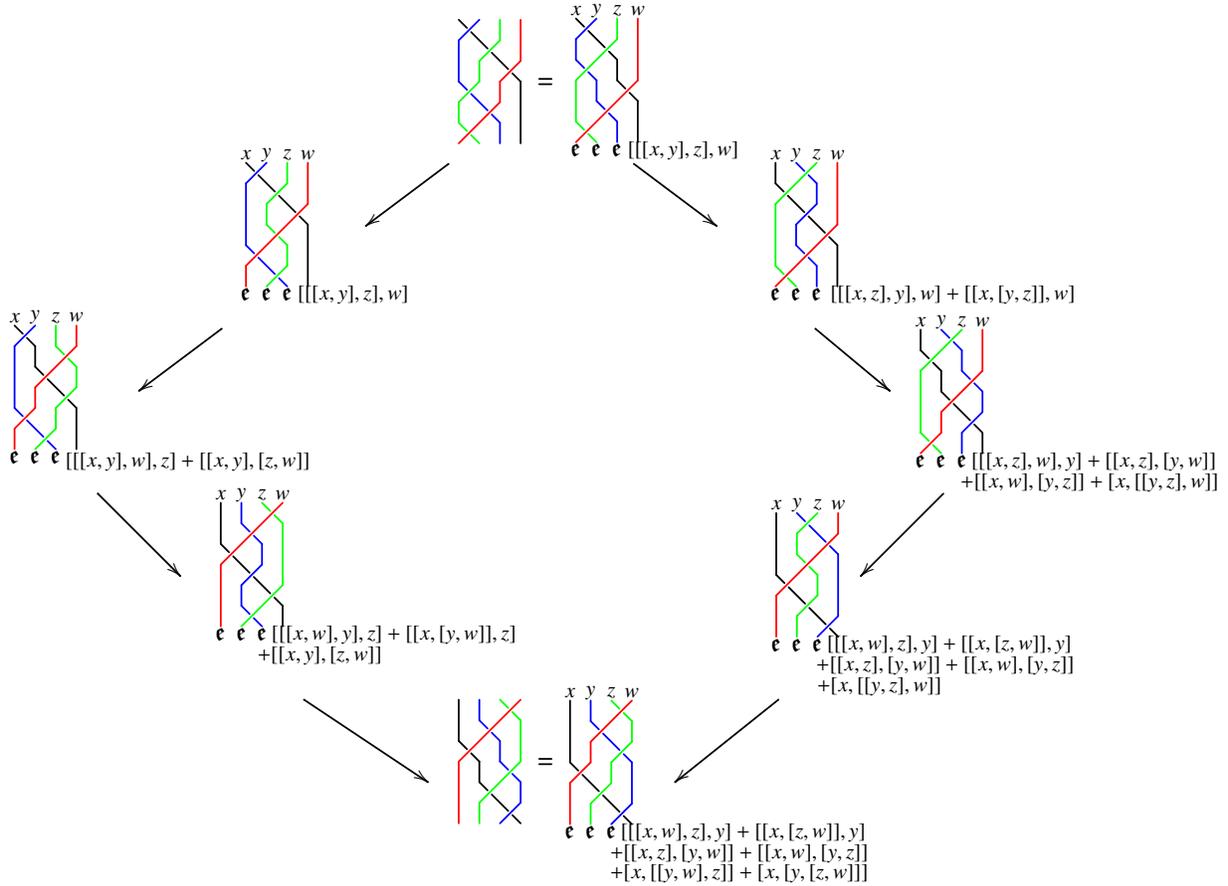
\begin{equation*}
  \xy 0;/r.13pc/:
    (-5,50)*+{
  \xy
   (0,0)*{}="00";
   (5,0)*{}="10";
   (10,0)*{}="20";
   (15,0)*{}="30";
   (0,5)*{}="01";
   (5,5)*{}="11";
   (10,5)*{}="21";
   (15,5)*{}="31";
   (0,10)*{}="02";
   (5,10)*{}="12";
   (10,10)*{}="22";
   (15,10)*{}="32";
   (0,15)*{}="03";
   (5,15)*{}="13";
   (10,15)*{}="23";
   (15,15)*{}="33";
   (0,-5)*{}="0-1";
   (5,-5)*{}="1-1";
   (10,-5)*{}="2-1";
   (15,-5)*{}="3-1";
   (0,-10)*{}="0-2";
   (5,-10)*{}="1-2";
   (10,-10)*{}="2-2";
   (15,-10)*{}="3-2";
   (0,-15)*{}="0-3";
   (5,-15)*{}="1-3";
   (10,-15)*{}="2-3";
   (15,-15)*{}="3-3";
   (2,13)*{}="213";
   (3,12)*{}="312";
   (7,8)*{}="78";
   (8,7)*{}="87";
   (12,3)*{}="123";
   (13,2)*{}="132";
   (2,-2)*{}="c2-2";
   (3,-3)*{}="c3-3";
   (7,-7)*{}="7-7";
   (8,-8)*{}="8-8";
   (2,-12)*{}="2-12";
   (3,-13)*{}="3-13";
   "0-3";"2-1" **[red]@{-};
   "2-1";"20" **[red]@{-};
   "20";"31" **[red]@{-};
   "31";"33" **[red]@{-};
   "1-3";"3-13" **[green]@{-};
   "2-12";"0-2" **[green]@{-};
   "0-2";"0-1" **[green]@{-};
   "0-1";"10" **[green]@{-};
   "10";"11" **[green]@{-};
   "11";"22" **[green]@{-};
   "22";"23" **[green]@{-};
   "2-3";"2-2" **[blue]@{-};
   "2-2";"8-8" **[blue]@{-};
   "7-7";"c3-3" **[blue]@{-};
   "c2-2";"00" **[blue]@{-};
   "00";"02" **[blue]@{-};
   "02";"13" **[blue]@{-};
   "3-3";"30" **@{-};
   "30";"132" **@{-};
   "123";"87" **@{-};
   "78";"312" **@{-};
   "213";"03" **@{-};
\endxy
    }="Z1";
    (-45,15)*+{
 \xy
   (0,0)*{}="00";
   (5,0)*{}="10";
   (10,0)*{}="20";
   (15,0)*{}="30";
   (0,5)*{}="01";
   (5,5)*{}="11";
   (10,5)*{}="21";
   (15,5)*{}="31";
   (0,10)*{}="02";
   (5,10)*{}="12";
   (10,10)*{}="22";
   (15,10)*{}="32";
   (0,15)*{}="03";
   (5,15)*{}="13";
   (10,15)*{}="23";
   (15,15)*{}="33";
   (0,-5)*{}="0-1";
   (5,-5)*{}="1-1";
   (10,-5)*{}="2-1";
   (15,-5)*{}="3-1";
   (0,-10)*{}="0-2";
   (5,-10)*{}="1-2";
   (10,-10)*{}="2-2";
   (15,-10)*{}="3-2";
   (0,-15)*{}="0-3";
   (5,-15)*{}="1-3";
   (10,-15)*{}="2-3";
   (15,-15)*{}="3-3";
   (2,13)*{}="213";
   (3,12)*{}="312";
   (7,8)*{}="78";
   (8,7)*{}="87";
   (12,3)*{}="123";
   (13,2)*{}="132";
   (7,-2)*{}="7-2";
   (8,-3)*{}="8-3";
   (2,-7)*{}="2-7";
   (3,-8)*{}="3-8";
   (7,-12)*{}="7-12";
   (8,-13)*{}="8-13";
   "0-3";"0-2" **[red]@{-};
   "0-2";"31" **[red]@{-};
   "31";"33" **[red]@{-};
   "1-3";"2-2" **[green]@{-};
   "2-2";"2-1" **[green]@{-};
   "2-1";"8-3" **[green]@{-};
   "7-2";"10" **[green]@{-};
   "10";"11" **[green]@{-};
   "11";"22" **[green]@{-};
   "22";"23" **[green]@{-};
   "2-3";"8-13" **[blue]@{-};
   "7-12";"3-8" **[blue]@{-};
   "2-7";"0-1" **[blue]@{-};
   "0-1";"02" **[blue]@{-};
   "02";"13" **[blue]@{-};
   "3-3";"30" **@{-};
   "30";"132" **@{-};
   "123";"87" **@{-};
   "78";"312" **@{-};
   "213";"03" **@{-};
   (0,17)*{\txt{\tiny $x$}}="x";
   (5,17)*{\txt{\tiny $y$}}="y";
   (10,17)*{\txt{\tiny $z$}}="z";
   (15,17)*{\txt{\tiny $w$}}="w";
   (0,-17)*{\txt{\tiny $\swe$}}="d1";
   (5,-17)*{\txt{\tiny $\swe$}}="d2";
   (10,-17)*{\txt{\tiny $\swe$}}="d3";
   (26,-17)*{\txt{\tiny $[[[x,y],z],w]$}}="d4";
\endxy
    }="Z2";
    (-85,-25)*+{
 \xy
   (0,0)*{}="00";
   (5,0)*{}="10";
   (10,0)*{}="20";
   (15,0)*{}="30";
   (0,5)*{}="01";
   (5,5)*{}="11";
   (10,5)*{}="21";
   (15,5)*{}="31";
   (0,10)*{}="02";
   (5,10)*{}="12";
   (10,10)*{}="22";
   (15,10)*{}="32";
   (0,15)*{}="03";
   (5,15)*{}="13";
   (10,15)*{}="23";
   (15,15)*{}="33";
   (0,-5)*{}="0-1";
   (5,-5)*{}="1-1";
   (10,-5)*{}="2-1";
   (15,-5)*{}="3-1";
   (0,-10)*{}="0-2";
   (5,-10)*{}="1-2";
   (10,-10)*{}="2-2";
   (15,-10)*{}="3-2";
   (0,-15)*{}="0-3";
   (5,-15)*{}="1-3";
   (10,-15)*{}="2-3";
   (15,-15)*{}="3-3";
   (2,13)*{}="213";
   (3,12)*{}="312";
   (12,8)*{}="128";
   (13,7)*{}="137";
   (7,3)*{}="73";
   (8,2)*{}="82";
   (12,-2)*{}="12-2";
   (13,-3)*{}="13-3";
   (2,-7)*{}="2-7";
   (3,-8)*{}="3-8";
   (7,-12)*{}="7-12";
   (8,-13)*{}="8-13";
   "0-3";"0-2" **[red]@{-};
   "0-2";"1-1" **[red]@{-};
   "1-1";"10" **[red]@{-};
   "10";"32" **[red]@{-};
   "32";"33" **[red]@{-};
   "1-3";"2-2" **[green]@{-};
   "2-2";"2-1" **[green]@{-};
   "2-1";"30" **[green]@{-};
   "30";"31" **[green]@{-};
   "31";"137" **[green]@{-};
   "128";"22" **[green]@{-};
   "22";"23" **[green]@{-};
   "2-3";"8-13" **[blue]@{-};
   "7-12";"3-8" **[blue]@{-};
   "2-7";"0-1" **[blue]@{-};
   "0-1";"02" **[blue]@{-};
   "02";"13" **[blue]@{-};
   "3-3";"3-1" **@{-};
   "3-1";"13-3" **@{-};
   "12-2";"82" **@{-};
   "73";"11" **@{-};
   "11";"12" **@{-};
   "12";"312" **@{-};
   "213";"03" **@{-};
   (0,17)*{\txt{\tiny $x$}}="x";
   (5,17)*{\txt{\tiny $y$}}="y";
   (10,17)*{\txt{\tiny $z$}}="z";
   (15,17)*{\txt{\tiny $w$}}="w";
   (0,-17)*{\txt{\tiny $\swe$}}="d1";
   (5,-17)*{\txt{\tiny $\swe$}}="d2";
   (10,-17)*{\txt{\tiny $\swe$}}="d3";
   (42,-18)*{\txt{\tiny $[[[x,y],w],z]+[[x,y],[z,w]]$}}="d4";
\endxy
    }="Z3";
    (-35,-70)*+{
 \xy
   (0,0)*{}="00";
   (5,0)*{}="10";
   (10,0)*{}="20";
   (15,0)*{}="30";
   (0,5)*{}="01";
   (5,5)*{}="11";
   (10,5)*{}="21";
   (15,5)*{}="31";
   (0,10)*{}="02";
   (5,10)*{}="12";
   (10,10)*{}="22";
   (15,10)*{}="32";
   (0,15)*{}="03";
   (5,15)*{}="13";
   (10,15)*{}="23";
   (15,15)*{}="33";
   (0,-5)*{}="0-1";
   (5,-5)*{}="1-1";
   (10,-5)*{}="2-1";
   (15,-5)*{}="3-1";
   (0,-10)*{}="0-2";
   (5,-10)*{}="1-2";
   (10,-10)*{}="2-2";
   (15,-10)*{}="3-2";
   (0,-15)*{}="0-3";
   (5,-15)*{}="1-3";
   (10,-15)*{}="2-3";
   (15,-15)*{}="3-3";
   (12,13)*{}="1213";
   (13,12)*{}="1312";
   (7,8)*{}="78";
   (8,7)*{}="87";
   (2,3)*{}="c23";
   (3,2)*{}="c32";
   (7,-2)*{}="7-2";
   (8,-3)*{}="8-3";
   (12,-7)*{}="12-7";
   (13,-8)*{}="13-8";
   (7,-12)*{}="7-12";
   (8,-13)*{}="8-13";
   "0-3";"00" **[red]@{-};
   "00";"33" **[red]@{-};
   "1-3";"3-1" **[green]@{-};
   "3-1";"32" **[green]@{-};
   "32";"1312" **[green]@{-};
   "1213";"23" **[green]@{-};
   "2-3";"8-13" **[blue]@{-};
   "7-12";"1-2" **[blue]@{-};
   "1-2";"1-1" **[blue]@{-};
   "1-1";"20" **[blue]@{-};
   "20";"21" **[blue]@{-};
   "21";"87" **[blue]@{-};
   "78";"12" **[blue]@{-};
   "12";"13" **[blue]@{-};
   "3-3";"3-2" **@{-};
   "3-2";"13-8" **@{-};
   "12-7";"8-3" **@{-};
   "7-2";"c32" **@{-};
   "c23";"01" **@{-};
   "01";"03" **@{-};
   (0,17)*{\txt{\tiny $x$}}="x";
   (5,17)*{\txt{\tiny $y$}}="y";
   (10,17)*{\txt{\tiny $z$}}="z";
   (15,17)*{\txt{\tiny $w$}}="w";
   (0,-17)*{\txt{\tiny $\swe$}}="d1";
   (5,-17)*{\txt{\tiny $\swe$}}="d2";
   (10,-17)*{\txt{\tiny $\swe$}}="d3";
   (42,-17)*{\txt{\tiny $[[[x,w],y],z]+[[x,[y,w]],z]$}}="d4";
   (24,-22)*{\txt{\tiny $+[[x,y],[z,w]]$}}="d4";
\endxy
    }="Z4";
(-5,-115)*+{
 \xy
   (0,0)*{}="00";
   (5,0)*{}="10";
   (10,0)*{}="20";
   (15,0)*{}="30";
   (0,5)*{}="01";
   (5,5)*{}="11";
   (10,5)*{}="21";
   (15,5)*{}="31";
   (0,10)*{}="02";
   (5,10)*{}="12";
   (10,10)*{}="22";
   (15,10)*{}="32";
   (0,15)*{}="03";
   (5,15)*{}="13";
   (10,15)*{}="23";
   (15,15)*{}="33";
   (0,-5)*{}="0-1";
   (5,-5)*{}="1-1";
   (10,-5)*{}="2-1";
   (15,-5)*{}="3-1";
   (0,-10)*{}="0-2";
   (5,-10)*{}="1-2";
   (10,-10)*{}="2-2";
   (15,-10)*{}="3-2";
   (0,-15)*{}="0-3";
   (5,-15)*{}="1-3";
   (10,-15)*{}="2-3";
   (15,-15)*{}="3-3";
   (12,13)*{}="1213";
   (13,12)*{}="1312";
   (7,8)*{}="78";
   (8,7)*{}="87";
   (2,3)*{}="c23";
   (3,2)*{}="c32";
   (12,-2)*{}="12-2";
   (13,-3)*{}="13-3";
   (7,-7)*{}="7-7";
   (8,-8)*{}="8-8";
   (12,-12)*{}="12-12";
   (13,-13)*{}="13-13";
   "0-3";"00" **[red]@{-};
   "00";"33" **[red]@{-};
   "1-3";"1-2" **[green]@{-};
   "1-2";"30" **[green]@{-};
   "30";"32" **[green]@{-};
   "32";"1312" **[green]@{-};
   "1213";"23" **[green]@{-};
   "2-3";"3-2" **[blue]@{-};
   "3-2";"3-1" **[blue]@{-};
   "3-1";"13-3" **[blue]@{-};
   "12-2";"20" **[blue]@{-};
   "20";"21" **[blue]@{-};
   "21";"87" **[blue]@{-};
   "78";"12" **[blue]@{-};
   "12";"13" **[blue]@{-};
   "3-3";"13-13" **@{-};
   "12-12";"8-8" **@{-};
   "7-7";"1-1" **@{-};
   "1-1";"10" **@{-};
   "10";"c32" **@{-};
   "c23";"01" **@{-};
   "01";"03" **@{-};
\endxy
    }="Z5";
   (35,50)*+{
\xy
   (0,0)*{}="00";
   (5,0)*{}="10";
   (10,0)*{}="20";
   (15,0)*{}="30";
   (0,5)*{}="01";
   (5,5)*{}="11";
   (10,5)*{}="21";
   (15,5)*{}="31";
   (0,10)*{}="02";
   (5,10)*{}="12";
   (10,10)*{}="22";
   (15,10)*{}="32";
   (0,15)*{}="03";
   (5,15)*{}="13";
   (10,15)*{}="23";
   (15,15)*{}="33";
   (0,-5)*{}="0-1";
   (5,-5)*{}="1-1";
   (10,-5)*{}="2-1";
   (15,-5)*{}="3-1";
   (0,-10)*{}="0-2";
   (5,-10)*{}="1-2";
   (10,-10)*{}="2-2";
   (15,-10)*{}="3-2";
   (0,-15)*{}="0-3";
   (5,-15)*{}="1-3";
   (10,-15)*{}="2-3";
   (15,-15)*{}="3-3";
   (2,13)*{}="213";
   (3,12)*{}="312";
   (7,8)*{}="78";
   (8,7)*{}="87";
   (2,3)*{}="c23";
   (3,2)*{}="c32";
   (12,-2)*{}="12-2";
   (13,-3)*{}="13-3";
   (7,-7)*{}="7-7";
   (8,-8)*{}="8-8";
   (2,-12)*{}="2-12";
   (3,-13)*{}="3-13";
   "0-3";"30" **[red]@{-};
   "30";"33" **[red]@{-};
   "1-3";"3-13" **[green]@{-};
   "2-12";"0-2" **[green]@{-};
   "0-2";"00" **[green]@{-};
   "00";"22" **[green]@{-};
   "22";"23" **[green]@{-};
   "2-3";"2-2" **[blue]@{-};
   "2-2";"8-8" **[blue]@{-};
   "7-7";"1-1" **[blue]@{-};
   "1-1";"10" **[blue]@{-};
   "10";"c32" **[blue]@{-};
   "c23";"01" **[blue]@{-};
   "01";"02" **[blue]@{-};
   "02";"13" **[blue]@{-};
   "3-3";"3-1" **@{-};
   "3-1";"13-3" **@{-};
   "12-2";"20" **@{-};
   "20";"21" **@{-};
   "21";"87" **@{-};
   "78";"312" **@{-};
   "213";"03" **@{-};
   (0,17)*{\txt{\tiny $x$}}="x";
   (5,17)*{\txt{\tiny $y$}}="y";
   (10,17)*{\txt{\tiny $z$}}="z";
   (15,17)*{\txt{\tiny $w$}}="w";
   (0,-17)*{\txt{\tiny $\swe$}}="d1";
   (5,-17)*{\txt{\tiny $\swe$}}="d2";
   (10,-17)*{\txt{\tiny $\swe$}}="d3";
   (26,-17)*{\txt{\tiny $[[[x,y],z],w]$}}="d4";
\endxy
   }="Y1";
   (100,15)*+{
 \xy
   (0,0)*{}="00";
   (5,0)*{}="10";
   (10,0)*{}="20";
   (15,0)*{}="30";
   (0,5)*{}="01";
   (5,5)*{}="11";
   (10,5)*{}="21";
   (15,5)*{}="31";
   (0,10)*{}="02";
   (5,10)*{}="12";
   (10,10)*{}="22";
   (15,10)*{}="32";
   (0,15)*{}="03";
   (5,15)*{}="13";
   (10,15)*{}="23";
   (15,15)*{}="33";
   (0,-5)*{}="0-1";
   (5,-5)*{}="1-1";
   (10,-5)*{}="2-1";
   (15,-5)*{}="3-1";
   (0,-10)*{}="0-2";
   (5,-10)*{}="1-2";
   (10,-10)*{}="2-2";
   (15,-10)*{}="3-2";
   (0,-15)*{}="0-3";
   (5,-15)*{}="1-3";
   (10,-15)*{}="2-3";
   (15,-15)*{}="3-3";
   (7,13)*{}="713";
   (8,12)*{}="812";
   (2,8)*{}="28";
   (3,7)*{}="37";
   (7,3)*{}="73";
   (8,2)*{}="82";
   (12,-2)*{}="12-2";
   (13,-3)*{}="13-3";
   (7,-7)*{}="7-7";
   (8,-8)*{}="8-8";
   (2,-12)*{}="2-12";
   (3,-13)*{}="3-13";
   "0-3";"30" **[red]@{-};
   "30";"33" **[red]@{-};
   "1-3";"3-13" **[green]@{-};
   "2-12";"0-2" **[green]@{-};
   "0-2";"01" **[green]@{-};
   "01";"23" **[green]@{-};
   "2-3";"2-2" **[blue]@{-};
   "2-2";"8-8" **[blue]@{-};
   "7-7";"1-1" **[blue]@{-};
   "1-1";"10" **[blue]@{-};
   "10";"21" **[blue]@{-};
   "21";"22" **[blue]@{-};
   "22";"812" **[blue]@{-};
   "713";"13" **[blue]@{-};
   "3-3";"3-1" **@{-};
   "3-1";"13-3" **@{-};
   "12-2";"82" **@{-};
   "73";"37" **@{-};
   "28";"02" **@{-};
   "02";"03" **@{-};
   (0,17)*{\txt{\tiny $x$}}="x";
   (5,17)*{\txt{\tiny $y$}}="y";
   (10,17)*{\txt{\tiny $z$}}="z";
   (15,17)*{\txt{\tiny $w$}}="w";
   (0,-17)*{\txt{\tiny $\swe$}}="d1";
   (5,-17)*{\txt{\tiny $\swe$}}="d2";
   (10,-17)*{\txt{\tiny $\swe$}}="d3";
   (43,-17)*{\txt{\tiny $[[[x,z],y],w]+[[x,[y,z]],w]$}}="d4";
\endxy
    }="Y2";
   (135,-28)*+{
 \xy
   (0,0)*{}="00";
   (5,0)*{}="10";
   (10,0)*{}="20";
   (15,0)*{}="30";
   (0,5)*{}="01";
   (5,5)*{}="11";
   (10,5)*{}="21";
   (15,5)*{}="31";
   (0,10)*{}="02";
   (5,10)*{}="12";
   (10,10)*{}="22";
   (15,10)*{}="32";
   (0,15)*{}="03";
   (5,15)*{}="13";
   (10,15)*{}="23";
   (15,15)*{}="33";
   (0,-5)*{}="0-1";
   (5,-5)*{}="1-1";
   (10,-5)*{}="2-1";
   (15,-5)*{}="3-1";
   (0,-10)*{}="0-2";
   (5,-10)*{}="1-2";
   (10,-10)*{}="2-2";
   (15,-10)*{}="3-2";
   (0,-15)*{}="0-3";
   (5,-15)*{}="1-3";
   (10,-15)*{}="2-3";
   (15,-15)*{}="3-3";
   (7,13)*{}="713";
   (8,12)*{}="812";
   (2,8)*{}="28";
   (3,7)*{}="37";
   (12,3)*{}="123";
   (13,2)*{}="132";
   (7,-2)*{}="7-2";
   (8,-3)*{}="8-3";
   (13,-8)*{}="13-8";
   (12,-7)*{}="12-7";
   (2,-12)*{}="2-12";
   (3,-13)*{}="3-13";
   "0-3";"1-2" **[red]@{-};
   "1-2";"1-1" **[red]@{-};
   "1-1";"31" **[red]@{-};
   "31";"33" **[red]@{-};
   "1-3";"3-13" **[green]@{-};
   "2-12";"0-2" **[green]@{-};
   "0-2";"01" **[green]@{-};
   "01";"23" **[green]@{-};
   "2-3";"2-2" **[blue]@{-};
   "2-2";"3-1" **[blue]@{-};
   "3-1";"30" **[blue]@{-};
   "30";"132" **[blue]@{-};
   "123";"21" **[blue]@{-};
   "21";"22" **[blue]@{-};
   "22";"812" **[blue]@{-};
   "713";"13" **[blue]@{-};
   "3-3";"3-2" **@{-};
   "3-2";"13-8" **@{-};
   "12-7";"8-3" **@{-};
   "7-2";"10" **@{-};
   "10";"11" **@{-};
   "11";"37" **@{-};
   "28";"02" **@{-};
   "02";"03" **@{-};
   (0,17)*{\txt{\tiny $x$}}="x";
   (5,17)*{\txt{\tiny $y$}}="y";
   (10,17)*{\txt{\tiny $z$}}="z";
   (15,17)*{\txt{\tiny $w$}}="w";
   (0,-17)*{\txt{\tiny $\swe$}}="d1";
   (5,-17)*{\txt{\tiny $\swe$}}="d2";
   (10,-17)*{\txt{\tiny $\swe$}}="d3";
   (42,-17)*{\txt{\tiny $[[[x,z],w],y]+[[x,z],[y,w]]$}}="d4";
   (41,-22)*{\txt{\tiny $+[[x,w],[y,z]]+[x,[[y,z],w]]$}}="d4";
\endxy
    }="Y3";
   (100,-75)*+{
 \xy
   (0,0)*{}="00";
   (5,0)*{}="10";
   (10,0)*{}="20";
   (15,0)*{}="30";
   (0,5)*{}="01";
   (5,5)*{}="11";
   (10,5)*{}="21";
   (15,5)*{}="31";
   (0,10)*{}="02";
   (5,10)*{}="12";
   (10,10)*{}="22";
   (15,10)*{}="32";
   (0,15)*{}="03";
   (5,15)*{}="13";
   (10,15)*{}="23";
   (15,15)*{}="33";
   (0,-5)*{}="0-1";
   (5,-5)*{}="1-1";
   (10,-5)*{}="2-1";
   (15,-5)*{}="3-1";
   (0,-10)*{}="0-2";
   (5,-10)*{}="1-2";
   (10,-10)*{}="2-2";
   (15,-10)*{}="3-2";
   (0,-15)*{}="0-3";
   (5,-15)*{}="1-3";
   (10,-15)*{}="2-3";
   (15,-15)*{}="3-3";
   (7,13)*{}="713";
   (8,12)*{}="812";
   (12,8)*{}="128";
   (13,7)*{}="137";
   (7,3)*{}="73";
   (8,2)*{}="82";
   (2,-2)*{}="c2-2";
   (3,-3)*{}="c3-3";
   (7,-7)*{}="7-7";
   (8,-8)*{}="8-8";
   (12,-12)*{}="12-12";
   (13,-13)*{}="13-13";
   "0-3";"0-1" **[red]@{-};
   "0-1";"32" **[red]@{-};
   "32";"33" **[red]@{-};
   "1-3";"1-2" **[green]@{-};
   "1-2";"2-1" **[green]@{-};
   "2-1";"20" **[green]@{-};
   "20";"82" **[green]@{-};
   "73";"11" **[green]@{-};
   "11";"12" **[green]@{-};
   "12";"23" **[green]@{-};
   "2-3";"3-2" **[blue]@{-};
   "3-2";"31" **[blue]@{-};
   "31";"137" **[blue]@{-};
   "128";"812" **[blue]@{-};
   "713";"13" **[blue]@{-};
   "3-3";"13-13" **@{-};
   "12-12";"8-8" **@{-};
   "7-7";"c3-3" **@{-};
   "c2-2";"00" **@{-};
   "00";"03" **@{-};
   (0,17)*{\txt{\tiny $x$}}="x";
   (5,17)*{\txt{\tiny $y$}}="y";
   (10,17)*{\txt{\tiny $z$}}="z";
   (15,17)*{\txt{\tiny $w$}}="w";
   (0,-17)*{\txt{\tiny $\swe$}}="d1";
   (5,-17)*{\txt{\tiny $\swe$}}="d2";
   (10,-17)*{\txt{\tiny $\swe$}}="d3";
   (42,-17)*{\txt{\tiny $[[[x,w],z],y]+[[x,[z,w]],y]$}}="d4";
   (41,-22)*{\txt{\tiny $+[[x,z],[y,w]]+[[x,w],[y,z]]$}}="d4";
   (25,-27)*{\txt{\tiny $+[x,[[y,z],w]]$}}="d4";
\endxy
    }="Y4";
   (50,-120.5)*+{
 \xy
   (0,0)*{}="00";
   (5,0)*{}="10";
   (10,0)*{}="20";
   (15,0)*{}="30";
   (0,5)*{}="01";
   (5,5)*{}="11";
   (10,5)*{}="21";
   (15,5)*{}="31";
   (0,10)*{}="02";
   (5,10)*{}="12";
   (10,10)*{}="22";
   (15,10)*{}="32";
   (0,15)*{}="03";
   (5,15)*{}="13";
   (10,15)*{}="23";
   (15,15)*{}="33";
   (0,-5)*{}="0-1";
   (5,-5)*{}="1-1";
   (10,-5)*{}="2-1";
   (15,-5)*{}="3-1";
   (0,-10)*{}="0-2";
   (5,-10)*{}="1-2";
   (10,-10)*{}="2-2";
   (15,-10)*{}="3-2";
   (0,-15)*{}="0-3";
   (5,-15)*{}="1-3";
   (10,-15)*{}="2-3";
   (15,-15)*{}="3-3";
   (12,13)*{}="1213";
   (13,12)*{}="1312";
   (7,8)*{}="78";
   (8,7)*{}="87";
   (12,3)*{}="123";
   (13,2)*{}="132";
   (2,-2)*{}="c2-2";
   (3,-3)*{}="c3-3";
   (7,-7)*{}="7-7";
   (8,-8)*{}="8-8";
   (12,-12)*{}="12-12";
   (13,-13)*{}="13-13";
   "0-3";"0-1" **[red]@{-};
   "0-1";"10" **[red]@{-};
   "10";"11" **[red]@{-};
   "11";"33" **[red]@{-};
   "1-3";"1-2" **[green]@{-};
   "1-2";"2-1" **[green]@{-};
   "2-1";"20" **[green]@{-};
   "20";"31" **[green]@{-};
   "31";"32" **[green]@{-};
   "32";"1312" **[green]@{-};
   "1213";"23" **[green]@{-};
   "2-3";"3-2" **[blue]@{-};
   "3-2";"30" **[blue]@{-};
   "30";"132" **[blue]@{-};
   "123";"87" **[blue]@{-};
   "78";"12" **[blue]@{-};
   "12";"13" **[blue]@{-};
   "3-3";"13-13" **@{-};
   "12-12";"8-8" **@{-};
   "7-7";"c3-3" **@{-};
   "c2-2";"00" **@{-};
   "00";"03" **@{-};
   (0,17)*{\txt{\tiny $x$}}="x";
   (5,17)*{\txt{\tiny $y$}}="y";
   (10,17)*{\txt{\tiny $z$}}="z";
   (15,17)*{\txt{\tiny $w$}}="w";
   (0,-17)*{\txt{\tiny $\swe$}}="d1";
   (5,-17)*{\txt{\tiny $\swe$}}="d2";
   (10,-17)*{\txt{\tiny $\swe$}}="d3";
   (42,-17)*{\txt{\tiny $[[[x,w],z],y]+[[x,[z,w]],y]$}}="d4";
   (41,-22)*{\txt{\tiny $+[[x,z],[y,w]]+[[x,w],[y,z]]$}}="d4";
   (41,-27)*{\txt{\tiny $+[x,[[y,w],z]]+[x,[y,[z,w]]]$}}="d4";
\endxy
    }="Y5";
            (7,50)*{}="X1";(10,50)*{}="X2";{\ar@{=} "X1";"X2"};
            (7,-115)*{}="X1";(10,-115)*{}="X2";{\ar@{=} "X1";"X2"};
            (-15,30)*{}="X1";(-35,15)*{}="X2";{\ar@{->} "X1";"X2"};
            (30,30)*{}="X1";(50,15)*{}="X2";{\ar@{->} "X1";"X2"};
            (-70,-10)*{}="X1";(-90,-25)*{}="X2";{\ar@{->} "X1";"X2"};
            (74,-10)*{}="X1";(92,-25)*{}="X2";{\ar@{->} "X1";"X2"};
            (-100,-50)*{}="X1";(-80,-70)*{}="X2";{\ar@{->} "X1";"X2"};
            (105,-50)*{}="X1";(85,-70)*{}="X2";{\ar@{->} "X1";"X2"};
            (-50,-100)*{}="X1";(-20,-120)*{}="X2";{\ar@{->} "X1";"X2"};
            (65,-100)*{}="X1";(40,-120)*{}="X2";{\ar@{->} "X1";"X2"};
  \endxy
\end{equation*}
which implies that $(B,Y)$ is a solution of the Zamolodchikov Tetrahedron equation as a consequence of the Jacobiator identity \eqref{Jacob}.
\end{proof}

\begin{ex}\label{cen-Lei-2-alg-ex1}
Consider the central Leibniz $2$-algebra $(\swg,[\cdot,\cdot],\huaJ,e)$ introduced in Example \ref{cen-Lei-to-cen-Lei-2-alg}, where $\swg=(\g,\g\oplus\R,s,t,\iii)$ is a Leibniz $2$-algebra and $(\g,[\cdot,\cdot]_\g,e)$ is a central Leibniz algebra.
By Theorem \ref{cen-Lei-to-sol}, we obtain a solution $(B,Y)$ of the Zamolodchikov Tetrahedron equation on $\swg$, where the linear invertible functor
$$B:\swg\otimes\swg\to\swg\otimes\swg$$
and the linear natural isomorphism
$$Y:(B\otimes \Id)(\Id\otimes B)(B\otimes \Id)\Rightarrow(\Id\otimes B)(B\otimes \Id)(\Id\otimes B)$$
are defined respectively as follows:
\begin{eqnarray}
B(x\otimes y)&=&y\otimes x+e\otimes[x,y]_\g,\label{cen-lei-to-cen-lei-2-alg-sol-B0}\\
B((x,a)\otimes(y,b))&=&(y,b)\otimes(x,a)+(e,0)\otimes([x,y]_\g,0),\label{cen-lei-to-cen-lei-2-alg-sol-B1}\\
Y_{x\otimes y\otimes z}&=&\iii_{z\otimes y\otimes x+e\otimes[y,z]_\g\otimes x+e\otimes y\otimes[x,z]_\g+z\otimes e\otimes[x,y]_\g}\\
&&+(e,0)\otimes(e,0)\otimes([x,y]_\g,z]_\g, \omega([x,y]_\g,z)),\label{cen-lei-to-cen-lei-2-alg-sol-Y}\nonumber
\end{eqnarray}
where $x,y,z\in\g$ and $(x,a),(y,b)\in\g\oplus\R$.
\end{ex}

\begin{ex}\label{cen-Lei-2-alg-ex2}
Consider the central Leibniz $2$-algebra $(\K\oplus\huaL,[\cdot,\cdot]_\oplus,\frkJ,(1,0))$ introduced in Example \ref{cen-ext-Lei-2-alg}.
Then by Theorem \ref{cen-Lei-to-sol}, we obtain a solution $(B,Y)$ of the Zamolodchikov Tetrahedron equation on $\K\oplus\huaL$, where the linear invertible functor $$B:(\K\oplus\huaL)\otimes(\K\oplus\huaL)\to(\K\oplus\huaL)\otimes(\K\oplus\huaL)$$
and the linear natural isomorphism
$$Y:(B\otimes \Id)(\Id\otimes B)(B\otimes \Id)\Rightarrow(\Id\otimes B)(B\otimes \Id)(\Id\otimes B)$$
are defined respectively as follows:
\begin{eqnarray}
B\big((a,x)\otimes(b,y)
&=&(b,y)\otimes(a,x)+(1,0)\otimes(0,[x,y]),\label{B-on-K+L_0}\\
B\big((a,f)\otimes(b,g)\big)
&=&(b,g)\otimes(a,f)+(1,0)\otimes(0,[f,g]),\label{B-on-K+L_1}\\
Y_{(a,x)\otimes(b,y)\otimes(c,z)}
&=&\iii_{(c,z)\otimes(b,y)\otimes(a,x)+(1,0)\otimes(0,[y,z])\otimes(a,x)+(1,0)\otimes(b,y)\otimes(0,[x,z])+(c,z)\otimes (1,0)\otimes(0,[x,y])}\label{Y-on-K+L_0}\\
&&+(1,0)\otimes(1,0)\otimes(0,\huaJ_{x\otimes y\otimes z}),\nonumber
\end{eqnarray}
where $(a,x),(b,y),(c,z)\in\K\oplus\huaL_0$ and $(a,f),(b,g)\in\K\oplus\huaL_1$.
\end{ex}

\begin{rmk}
In \cite{BC}, Baez and Crans used a Lie $2$-algebra to give a solution of the Zamolodchikov Tetrahedron equation of the same formula as in the above example. Theorem \ref{cen-Lei-to-sol} demonstrates that, in the process of constructing the solution, the antisymmetry of $[\cdot,\cdot]$ does not play a role, and this construction can be enhanced to central Leibniz $2$-algebras.
\end{rmk}

At the end of this section, we show that one can obtain a central Leibniz algebra through decategorification of a central Leibniz 2-algebra, and the corresponding solution of the Yang-Baxter equation is exactly the decategorification of the solution of the Zamolodchikov Tetrahedron equation given in Theorem \ref{cen-Lei-to-sol}.

\begin{pro}\label{decate-Leibniz-2-alg}
Let $\huaL=(\huaL_0,\huaL_1,s,t,\iii)$ be a $2$-vector space,  and $(\huaL,[\cdot,\cdot],\huaJ)$ a Leibniz $2$-algebra.
Define a linear operator $[\cdot,\cdot]_{\overline{\huaL_0}}:\overline{\huaL_0}\otimes\overline{\huaL_0}\to\overline{\huaL_0}$ as follows:
$$[\overline{x},\overline{y}]_{\overline{\huaL_0}}=\overline{[x,y]},\quad\forall \overline{x},\overline{y}\in\overline{\huaL_0},$$
where the vector space $\overline{\huaL_0}$ is the decategorification of the $2$-vector space $\huaL$.
Then $(\overline{\huaL_0},[\cdot,\cdot]_{\overline{\huaL_0}})$ is a Leibniz algebra, which is the decategorification of the Leibniz $2$-algebra $(\huaL,[\cdot,\cdot],\huaJ)$.
In particular, if $\swe$ is a central object of $(\huaL,[\cdot,\cdot],\huaJ)$, then $\overline{\swe}\in\overline{\huaL_0}$ is a central element of $(\overline{\huaL_0},[\cdot,\cdot]_{\overline{\huaL_0}})$.

Moreover, we have the following commutative diagram:
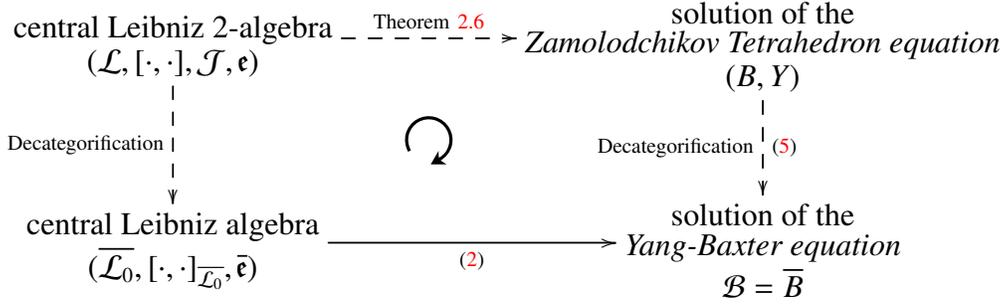
\begin{center}
\begin{displaymath}
\xymatrix@C=3ex@R=0.5ex{
  \txt{\rm central Leibniz $2$-algebra\\ $(\huaL,[\cdot,\cdot],\huaJ,\swe)$}
  \ar@{-->}[rr]^-{{\rm Theorem}~\ref{cen-Lei-to-sol}}
  \ar@{-->}[dd]_-{{\rm Decategorification}}
  &
  &\txt{\rm solution of the\\Zamolodchikov Tetrahedron equation\\$(B,Y)$}
  \ar@{-->}[dd]^-{\eqref{deficat-of-ZTE}}_-{{\rm Decategorification}}\\
  &\rotatebox{165}{{\txt{\Huge $\circlearrowright$}}}&\\
  \txt{\rm central Leibniz algebra\\ $(\overline{\huaL_0},[\cdot,\cdot]_{\overline{\huaL_0}},\overline{\swe})$}
  \ar[rr]_-{\eqref{cen-Leibniz-sol}}
  &
  &\txt{\rm solution of the\\Yang-Baxter equation\\$\huaB=\overline{B}$}
}
\end{displaymath}
\end{center}
 \end{pro}
\begin{proof}
First, we show that the bracket $[\cdot,\cdot]_{\overline{\huaL_0}}$ is well-defined.
It is well known that, in the $2$-vector space $\huaL$, every morphism $f\in\huaL_1$ is an isomorphism, whose inverse is $\iii_{t(f)-s(f)}-f$.
Therefore, $\overline{x}=\overline{x'}$ in $\overline{\huaL_0}$ means that there is a morphism $f:x\to x'$. Then we obtain a morphism $[f,\iii_y]:[x,y]\to[x',y]$, which implies that $\overline{[x,y]}=\overline{[x',y]}$. Therefore, we have
$$[\overline{x},\overline{y}]_{\overline{\huaL_0}}
=\overline{[x,y]}
=\overline{[x',y]}
=[\overline{x'},\overline{y}]_{\overline{\huaL_0}},$$
that is, the bracket $[\cdot,\cdot]_{\overline{\huaL_0}}$ is well-defined.
The Leibniz identity holds since the Jacobiator $\huaJ$ is a linear natural isomorphism. Thus $(\overline{\huaL_0},[\cdot,\cdot]_{\overline{\huaL_0}})$ is a Leibniz algebra.

If $\swe$ is a central object of the Leibniz $2$-algebra $(\huaL,[\cdot,\cdot],\huaJ)$, then by \eqref{central-s} we have
$$[\overline{\swe},\overline{x}]_{\overline{\huaL_0}}=\overline{[\swe,x]}
=\overline{0}
=\overline{[x,\swe]}=[\overline{x},\overline{\swe}]_{\overline{\huaL_0}},
\quad\forall\overline{x}\in\overline{\huaL_0},$$
which implies that $\overline{\swe}$ is a central object of the Leibniz algebra $(\overline{\huaL_0},[\cdot,\cdot]_{\overline{\huaL_0}})$.

Moreover, by identifying the vector space $\overline{\huaL_0}\otimes\overline{\huaL_0}$ with $\overline{\huaL_0\otimes \huaL_0}$, we obtain that the solution $\huaB$ of the Yang-Baxter equation on the vector space $\overline{\huaL_0}$ induced by \eqref{cen-Leibniz-sol} is the same as the decategorification of the solution ${B}$ of the Zamolodchikov Tetrahedron equation induced by \eqref{deficat-of-ZTE}:
  \begin{eqnarray*}
    \huaB(\bar{x}\otimes \bar{y})&\overset{\eqref{cen-Leibniz-sol}}=&\bar{y}\otimes \bar{x}+\bar{\swe}\otimes [\bar{x}, \bar{y}]_{\bar{\huaL_0}}\\
    &=&\overline{y\otimes x+\swe\otimes[x,y]}\\
    &=&\overline{B(x\otimes y)}\\
    &\overset{\eqref{deficat-of-ZTE}}=&\overline{B}(\bar{x}\otimes \bar{y}),
  \end{eqnarray*}
where $\bar{x},\bar{y}\in\overline{\huaL_0}$.
\end{proof}

\section{Linear 2-racks and the Zamolodchikov Tetrahedron equation}\label{sec:l2rack}

In this section, we introduce the notion of linear 2-racks, and show that a linear 2-rack naturally gives rise to a solution of the Zamolodchikov Tetrahedron equation.

\begin{defi}
A (right) {\bf linear $2$-rack} structure on a $2$-vector space $V=(V_0,V_1,s,t,\iii)$ consists of
\begin{itemize}
    \item[\rm(i)] a linear functor $\Delta:V\to V\otimes V$ satisfies:
  \begin{eqnarray}
  (\Delta\otimes\Id)\circ\Delta&=&(\Id\otimes\Delta)\circ\Delta,\label{coasso}\\
  \tau\circ\Delta&=&\Delta,\label{cocomm}
  \end{eqnarray}
  where $\tau:V\otimes V\to V\otimes V$ is a linear functor defined by $\tau(x\otimes y)=y\otimes x$ for both objects and morphisms.
  Denote by $\Delta(x)=x_{(1)}\otimes x_{(2)}$ for either objects or morphisms.
  \item[\rm(ii)] a linear functor $\varepsilon:V\to \K$ satisfies:
  \begin{eqnarray}\label{counit}
  (\varepsilon\otimes\Id)\circ\Delta~=~\Id~=~(\Id\otimes\varepsilon)\circ\Delta;
  \end{eqnarray}
  \item[\rm(iii)] a linear functor $\lhd:V\otimes V\to V$ satisfies:
  \begin{eqnarray}
  \Delta\circ\lhd
  &=&(\lhd\otimes\lhd)\circ(\Id\otimes\tau\otimes\Id)\circ(\Delta\otimes\Delta),\label{rack-coproduct}\\
  \varepsilon\circ\lhd&=&\varepsilon\otimes\varepsilon,\label{rack-counit}
  \end{eqnarray}
  and there exists a linear functor $\widetilde{\lhd}:V\otimes V\to V$ such that
  \begin{eqnarray}
  \widetilde{\lhd}\circ(\lhd\otimes\Id)\circ(\Id\otimes\Delta)~=
  &\Id\otimes\varepsilon&=~\lhd\circ(\widetilde{\lhd}\otimes\Id)\circ(\Id\otimes\Delta),\label{rack-inv}
  \end{eqnarray}
  \item[\rm(iv)] a linear natural isomorphism $\frkR_{x\otimes y\otimes z}:(x\lhd y)\lhd z\to(x\lhd z_{(1)})\lhd(y\lhd z_{(z)})$ satisfies the following identity:
  \begin{eqnarray}\label{lin-distri}
  &&(\frkR_{x\otimes z_{(1)}\otimes w_{(1)}}\lhd\frkR_{x\otimes z_{(2)}\otimes w_{(2)}})\circ\frkR_{(x\lhd{z_{(1)}})\otimes(y\lhd z_{(2)})\otimes w}\circ(\frkR_{x\otimes y\otimes z}\lhd\iii_w)\\
  &=&\frkR_{(x\lhd w_{(1)(1)})\otimes(y\lhd w_{(1)(2)})\otimes(z\lhd w_{(2)})}\circ\big({\frkR_{x\otimes y\otimes w_{(1)}}\lhd(\iii_z\lhd\iii_{w_{(2)}})}\big)\circ{\frkR_{(x\lhd y)\otimes z\otimes w}},\nonumber
  \end{eqnarray}
\end{itemize}
where $x,y,z,w\in V_0$ and the above identity can be showed as the following commutative diagram:
\begin{equation*}
\xymatrix@R=1.5pc@C=0.05pc{
 &\txt{$\bigg((x\lhd y)\lhd z\bigg)\lhd w$}
 \ar[dl]_-{\frkR_{x\otimes y\otimes z}\lhd\iii_w}
 \ar[dr]^-{\frkR_{(x\lhd y)\otimes z\otimes w}}
 &\\
 \txt{$\bigg((x\lhd z_{(1)})\lhd(y\lhd z_{(2)})\bigg)\lhd w$}
 \ar[d]_-{\frkR_{(x\lhd{z_{(1)}})\otimes(y\lhd z_{(2)})\otimes w}}
 &
 &\txt{$\bigg((x\lhd y)\lhd w_{(1)}\bigg)\lhd(z\lhd w_{(2)})$}
 \ar[d]^-{\frkR_{x\otimes y\otimes w_{(1)}}\lhd(\iii_z\lhd\iii_{w_{(2)}})}\\
 \txt{$\bigg((x\lhd z_{(1)})\lhd w_{(1)}\bigg)\lhd\bigg((y\lhd z_{(2)})\lhd w_{(2)}\bigg)$}
 \ar[d]_-{\frkR_{x\otimes z_{(1)}\otimes w_{(1)}}\lhd\frkR_{x\otimes z_{(2)}\otimes w_{(2)}}}
 &
 &\txt{$\bigg((x\lhd w_{(1)(1)})\lhd(y\lhd w_{(1)(2)})\bigg)\lhd(z\lhd w_{(2)})$}
 \ar[d]^-{\frkR_{(x\lhd w_{(1)(1)})\otimes(y\lhd w_{(1)(2)})\otimes(z\lhd w_{(2)})}}\\
 \txt{$\bigg((x\lhd w_{(1)(1)})\lhd(z_{(1)}\lhd w_{(1)(2)})\bigg)$\\
 $\lhd\bigg((y\lhd w_{(2)(1)})\lhd(z_{(2)}\lhd w_{(2)(2)})\bigg)$}
 \ar@{=}[rr]_{\eqref{cocomm}~{\rm and}~\eqref{rack-coproduct}}
 &
 &\txt{$\bigg((x\lhd w_{(1)(1)})\lhd{(z\lhd w_{(2)})}_{(1)}\bigg)$\\
 $\lhd\bigg((y\lhd w_{(1)(2)})\lhd{(z\lhd w_{(2)})}_{(2)}\bigg)$}
 }
 \end{equation*}
\end{defi}

We will denote a linear $2$-rack by $(V,\Delta,\varepsilon,\lhd,\frkR)$.

\begin{rmk}
  (i) and (ii) mean that $(V,\Delta,\varepsilon)$ is the strict categorification of a cocommutative coassociative coalgebra.
\end{rmk}

\begin{ex}\label{ex-cen-lei-K+L-to lin-2-rack}
Consider the central Leibniz $2$-algebra $(\K\oplus\huaL,[\cdot,\cdot]_\oplus,\frkJ,(1,0))$ introduced in Example \ref{cen-ext-Lei-2-alg},
where $\huaL=(\huaL_0,\huaL_1,s,t,\iii)$ is a $2$-vector space and $(\huaL,[\cdot,\cdot],\huaJ)$ is a Leibniz $2$-algebra.
Define two linear functors $\Delta:\K\oplus\huaL\to(\K\oplus\huaL)\otimes(\K\oplus\huaL)$ and $\varepsilon:\K\oplus\huaL\to\K$ as follows:
\begin{eqnarray*}
\Delta\Big((a,x)\Big)&=&(a,x)\otimes(1,0)+(1,0)\otimes(0,x),\\
\Delta\Big((a,f)\Big)&=&(a,f)\otimes(1,0)+(1,0)\otimes(0,f),\\
\varepsilon\Big((a,x)\Big)&=&a,\\
\varepsilon\Big((a,f)\Big)&=&a.
\end{eqnarray*}
where $(a,x)\in\K\oplus\huaL_0$ and $(a,f)\in\K\oplus\huaL_1$. Define a linear functor $\lhd:(\K\oplus\huaL)\otimes(\K\oplus\huaL)\to\K\oplus\huaL$ as follows:
\begin{eqnarray*}
(a,x)\lhd(b,y)&=&(ab,bx+[x,y]),\quad\forall(a,x),(b,y)\in\K\oplus\huaL_0,\\
(a,f)\lhd(b,g)&=&(ab,bf+[f,g]),\quad\forall(a,f),(b,g)\in\K\oplus\huaL_1,
\end{eqnarray*}
and define a linear functor $\widetilde{\lhd}:(\K\oplus\huaL)\otimes(\K\oplus\huaL)\to\K\oplus\huaL$ by:
\begin{eqnarray*}
(a,x)~\widetilde{\lhd}~(b,y)&=&(ab,bx-[x,y]),\quad\forall(a,x),(b,y)\in\K\oplus\huaL_0,\\
(a,f)~\widetilde{\lhd}~(b,g)&=&(ab,bf-[f,g]),\quad\forall(a,f),(b,g)\in\K\oplus\huaL_1.
\end{eqnarray*}
Finally, define a linear natural isomorphism by
$$\frkR_{(a,x)\otimes(b,y)\otimes(c,z)}=(\iii_{abc},\iii_{bcx+b[x,z]+c[x,y]}+\huaJ_{x\otimes y\otimes z}).$$
Then   $(\K\oplus\huaL,\Delta,\varepsilon,\lhd,\frkR)$ is a linear $2$-rack.

\end{ex}

\begin{thm}\label{lin-2-rack-to-sol}
\emptycomment{
Let $V$ be a $2$-vector space,
$\Delta:V\to V\otimes V$ and $\varepsilon:V\to\K$ two linear functors satisfying $\eqref{coasso}$, $\eqref{cocomm}$ and $\eqref{counit}$,
$\lhd:V\otimes V\to V$ a linear functor satisfies $\eqref{rack-coproduct}$ and $\eqref{rack-counit}$,
$\frkR_{x\otimes y\otimes z}:(x\lhd y)\lhd z\to(x\lhd z_{(1)})\lhd(y\lhd z_{(z)})$ a linear natural isomorphism.
Define a linear functor $B:V\otimes V\to V\otimes V$ by $B=(\Id\otimes\lhd)\circ(\tau\otimes\Id)\circ(\Id\otimes\Delta)$,
that is, $B(x\otimes y)=y_{(1)}\otimes(x\lhd y_{(2)})$ for either objects or morphisms.
Define a linear natural isomorphism $Y:(B\otimes\Id)(\Id\otimes B)(B\otimes\Id)
\Rightarrow(\Id\otimes B)(B\otimes\Id)(\Id\otimes B)$ by $Y_{x\otimes y\otimes z}=\Id_{z_{(1)}}\otimes\Id_{y_{(1)}\lhd z_{(2)}}\otimes\frkR_{x\otimes y_{(2)}\otimes z_{(3)}}$.
Then $(B,Y)$ is a solution of Zamolodchikov Tetrahedron Equation
if and only if $(V,\Delta,\varepsilon,\lhd,\frkR)$ is a linear $2$-rack.}
Let $(V,\Delta,\varepsilon,\lhd,\frkR)$ be a linear $2$-rack.
Define a linear functor $B:V\otimes V\to V\otimes V$ and a linear natural isomorphism $Y:(B\otimes\Id)(\Id\otimes B)(B\otimes\Id)
\Rightarrow(\Id\otimes B)(B\otimes\Id)(\Id\otimes B)$ respectively as follows:
\begin{eqnarray*}
B&=&(\Id\otimes\lhd)\circ(\tau\otimes\Id)\circ(\Id\otimes\Delta),\\
Y_{x\otimes y\otimes z}&=&\iii_{z_{(1)}}\otimes\iii_{y_{(1)}\lhd z_{(2)}}\otimes\frkR_{x\otimes y_{(2)}\otimes z_{(3)}},\qquad\forall x,y,z\in V_0.
\end{eqnarray*}
Then $(B,Y)$ is a solution of the Zamolodchikov Tetrahedron equation.
\end{thm}

\begin{proof}
It is obvious that ${B}(x\otimes y)=y_{(1)}\otimes(x {\lhd}~y_{(2)})$, where $x$ is either an object or a morphism.
Let $\widetilde{\lhd}:V\otimes V\to V$ be the linear functor such that \eqref{rack-inv} holds.
Define a linear functor $\widetilde{B}:V\otimes V\to V\otimes V$ by $$\widetilde{B}=(\widetilde{\lhd}\otimes\Id)\circ(\tau\otimes\Id)\circ(\Id\otimes\tau)\circ(\Delta\otimes\Id),$$ that is, $\widetilde{B}(x\otimes y)=(y~\widetilde{\lhd}~x_{(1)})\otimes x_{(2)}$ for either objects or morphisms.
Then for any $x\otimes y\in V_0\otimes V_0$, we have
\begin{eqnarray*}
(\widetilde{B}\circ B)(x\otimes y)
&=&\widetilde{B}\Big(y_{(1)}\otimes(x\lhd y_{(2)})\Big)
=\Big((x\lhd y_{(2)})~\widetilde{\lhd}~y_{(1)(1)}\Big)\otimes y_{(1)(2)}\\
&\overset{\eqref{coasso},\eqref{cocomm}}=&\Big((x\lhd y_{(1)(1)})~\widetilde{\lhd}~y_{(1)(2)}\Big)\otimes y_{(2)}
\overset{\eqref{rack-inv}}= x\otimes\varepsilon(y_{(1)})y_{(2)}\\
&\overset{\eqref{counit}}=&x\otimes y,\\
({B}\circ\widetilde{B})(x\otimes y)
&=&B\Big((y~\widetilde{\lhd}~x_{(1)})\otimes x_{(2)}\Big)
=x_{(2)(1)}\otimes\Big((y~\widetilde{\lhd}~x_{(1)})\lhd x_{(2)(2)}\Big)\\
&\overset{\eqref{coasso},\eqref{cocomm}}=&
x_{(1)}\otimes\Big((y~\widetilde{\lhd}~x_{(2)(1)})\lhd x_{(2)(2)}\Big)
\overset{\eqref{rack-inv}}= x_{(1)}\varepsilon(x_{(2)})\otimes y\\
&\overset{\eqref{counit}}=&x\otimes y.
\end{eqnarray*}
Similarly, $\widetilde{B}\circ B=\Id$ and ${B}\circ\widetilde{B}=\Id$ hold for any $f\otimes g\in V_1\otimes V_1$. Therefore, the linear functor $B$ is invertible.

For any $x,y,z\in V_0$, using \eqref{coasso}, \eqref{cocomm} and \eqref{rack-coproduct}, we have
\begin{eqnarray*}
(B\otimes \Id)(\Id\otimes B)(B\otimes \Id)(x\otimes y\otimes z)
&=&z_{(1)}\otimes(y_{(1)}\lhd z_{(2)})\otimes\big((x\lhd y_{(2)})\lhd z_{(3)}\big),\\
(\Id\otimes B)(B\otimes \Id)(\Id\otimes B)(x\otimes y\otimes z)
&=&z_{(1)}\otimes(y_{(1)}\lhd z_{(2)})\otimes\big((x\lhd z_{(3)(1)})\lhd(y_{(2)}\lhd z_{(3)(2)})\big),
\end{eqnarray*}
which implies that
\begin{eqnarray*}
s(Y_{x\otimes y\otimes z})&=&(B\otimes \Id)(\Id\otimes B)(B\otimes \Id)(x\otimes y\otimes z),\\
t(Y_{x\otimes y\otimes z})&=&(\Id\otimes B)(B\otimes \Id)(B\otimes \Id)(x\otimes y\otimes z),
\end{eqnarray*}
that is, $Y$ is compatible with the source and target maps.
By direct calculation, for any $f\otimes g\otimes h:x\otimes y\otimes z\to x'\otimes y'\otimes z'$, we have
\begin{center}
\begin{displaymath}
\xymatrix@C=3ex@R=0.3ex{
  \txt{$(B\otimes \Id)(\Id\otimes B)(B\otimes \Id)
  (x\otimes y\otimes z)$}
  \ar[rr]^-{Y_{x\otimes y\otimes z}}
  \ar[dd]_-{(B\otimes \Id)(\Id\otimes B)(B\otimes \Id)
  (f\otimes g\otimes h)}
  &
  &\txt{$(\Id\otimes B)(B\otimes \Id)(\Id\otimes B)
  (x\otimes y\otimes z)$}
  \ar[dd]^-{(\Id\otimes B)(B\otimes \Id)(\Id\otimes B)
  (f\otimes g\otimes h)}\\
  &\rotatebox{165}{{\txt{\Huge $\circlearrowright$}}}&\\
  \txt{$(B\otimes \Id)(\Id\otimes B)(B\otimes \Id)
  (x'\otimes y'\otimes z')$}
  \ar[rr]_-{Y_{x'\otimes y'\otimes z'}}
  &
  &\txt{$(\Id\otimes B)(B\otimes \Id)(\Id\otimes B)
  (x'\otimes y'\otimes z')$}
}
\end{displaymath}
\end{center}
which implies that $Y$ is a linear natural isomorphism.

Similar to the proof of Theorem \ref{cen-Lei-to-sol}, we obtain that $(B,Y)$ satisfies the Zamolodchikov Tetrahedron equation because $\eqref{lin-distri}$ holds, and we omit details.
\end{proof}

\begin{ex}
Consider the linear $2$-rack $(\K\oplus\huaL,\Delta,\varepsilon,\lhd,\frkR)$ obtained by a central Leibniz $2$-algebra $(\K\oplus\huaL,[\cdot,\cdot]_\oplus,\frkJ,(1,0))$, which is showed in Example \ref{ex-cen-lei-K+L-to lin-2-rack}.
By Theorem \ref{lin-2-rack-to-sol}, we obtain a solution $(B,Y)$ of the Zamolodchikov Tetrahedron equation on $\K\oplus\huaL$, where the lnear invertible functor $$B:(\K\oplus\huaL)\otimes(\K\oplus\huaL)\to(\K\oplus\huaL)\otimes(\K\oplus\huaL)$$
and the linear natural isomorphism
$$Y:(B\otimes \Id)(\Id\otimes B)(B\otimes \Id)\Rightarrow(\Id\otimes B)(B\otimes \Id)(\Id\otimes B)$$
are defined respectively by \eqref{B-on-K+L_0}-\eqref{Y-on-K+L_0},
i.e. the solution induced by the linear $2$-rack $(\K\oplus\huaL,\Delta,\varepsilon,\lhd,\frkR)$ is the same as the one induced by the central Leibniz $2$-algebra $(\K\oplus\huaL,[\cdot,\cdot]_\oplus,\frkJ,(1,0))$.
\end{ex}

\begin{rmk}
  Motivated by this example, one can expect certain relations between central Leibniz $2$-algebras and linear $2$-racks. In next section, we will show that one can obtain a linear $2$-rack from a central Leibniz $2$-algebra under certain conditions.
\end{rmk}

Similar to Proposition \ref{decate-Leibniz-2-alg}, one can obtain a linear rack via the decategorification of a linear $2$-rack, such that the solution  of the Yang-Baxter equation given by the linear rack and the solution obtained by the decategorification of the Zamolodchikov Tetrahedron equation are the same.
\begin{cor}\label{decate-lin-2-rack}
Let $V=(V_0,V_1,s,t,\iii)$ be a $2$-vector space,  and $(V,\Delta,\varepsilon,\lhd,\frkR)$ a linear $2$-rack.
Define linear maps
$\overline{\Delta}:\overline{V_0}\to\overline{V_0}\otimes\overline{V_0}$,
$\overline{\varepsilon}:\overline{V_0}\to\K$, and
$\overline{\lhd}:\overline{V_0}\otimes\overline{V_0}\to\overline{V_0}$ as follows:
\begin{eqnarray*}
&&\overline{\Delta}(\overline{x})=\overline{\Delta(x)},\qquad
\overline{\varepsilon}(\overline{x})=\varepsilon(x),\qquad
\overline{x}~\overline{\lhd}~\overline{y}=\overline{x\lhd y},
\end{eqnarray*}
where $\overline{x},\overline{y}\in\overline{V_0}$ and the vector space $\overline{V_0}$ is the decategorification of $V$.
Then $(\overline{V_0},\overline{\Delta},\overline{\varepsilon},\overline{\lhd})$ is a linear rack, which is  the decategorification of the linear $2$-rack $(V,\Delta,\varepsilon,\lhd,\frkR)$.

Moreover, we have the following commutative diagram:
\begin{center}
\begin{displaymath}
\xymatrix@C=3ex@R=0.5ex{
  \txt{\rm linear $2$-rack\\ $(V,\Delta,\varepsilon,\lhd,\frkR)$}
  \ar@{-->}[rr]^-{{\rm Theorem}~\ref{lin-2-rack-to-sol}}
  \ar@{-->}[dd]_-{{\rm Decategorification}}
  &
  &\txt{\rm the solution of the\\Zamolodchikov Tetrahedron equation\\$(B,Y)$}
  \ar@{-->}[dd]^-{\eqref{deficat-of-ZTE}}_-{{\rm Decategorification}}\\
  &\rotatebox{165}{{\txt{\Huge $\circlearrowright$}}}&\\
  \txt{\rm linear rack\\ $(\overline{V_0},\overline{\Delta},\overline{\varepsilon},\overline{\lhd})$}
  \ar[rr]_-{\eqref{quantum rack to solution formula}}
  &
  &\txt{\rm the solution of the\\Yang-Baxter equation\\$\huaB=\overline{B}$}
}
\end{displaymath}
\end{center}
 \end{cor}

\section{On the passage from central Leibniz $2$-algebras to linear $2$-racks}\label{sec:passage}

In \cite{Lebed3}, Lebed showed that a Leibniz algebra $\g$ gives rise to a linear rack structure on  $\g\oplus \K$. Moreover,  the solutions constructed respectively by the central Leibniz algebra $(\g\oplus \K, (0,1)) $ and the corresponding linear rack $\g\oplus \K$ are the same.
We understand $(\g\oplus \K, (0,1)) $ as the trivial central extension of the Leibniz algebra $\g$.

\emptycomment{\begin{equation}\label{cen-lei-lin-rack-sol}
 \begin{array}{l}
\xymatrix@!0@C=9.5ex@R=7ex{
&&\txt{\rm solutions of the YBE~$\huaB$}
&&\\
&& {\color{blue}{\txt{\Huge \rotatebox{165}{$\circlearrowright$}}}}
&&\\
\txt{\rm central Leibniz algebras\\$(\g\oplus \K, (0,1))$}
\ar[uurr]^-{\eqref{cen-Leibniz-sol}}
\ar[rrrr]_-{\text{\cite[Proposition~4.7]{Lebed3}}}
&&&
&\txt{\rm linear racks\\$\g\oplus \K$.}
\ar[uull]_-{\eqref{quantum rack to solution formula}}
&
}
\end{array}
\end{equation}
\begin{equation}\label{cen-lei-lin-rack-sol}
 \begin{array}{l}
\xymatrix@!0@C=9.5ex@R=7ex{
&&\txt{\rm solutions of the YBE~$\huaB$}
&&\\
&&{\color{blue}{\txt{\Huge \rotatebox{165}{$\circlearrowright$}\\\rm($\huaL$~is~splittable)}}}
&&\\
\txt{\rm central Leibniz algebras\\$(\huaL,[\cdot,\cdot],\swe)$}
\ar[uurr]^-{\eqref{cen-Leibniz-sol}}
\ar[rrrr]_-{\text{\cite[Proposition~4.7]{Lebed3}}}
&&&
&\txt{\rm linear racks\\$(\huaL,\Delta,\varepsilon,\lhd)$}
\ar[uull]_-{\eqref{quantum rack to solution formula}}
&
}
\end{array}
\end{equation}
}

Since we have shown that central Leibniz algebras and linear racks serve as decategorifications of central Leibniz $2$-algebras and linear $2$-racks respectively, there are natural questions:
\begin{itemize}
\item[\bf Q1:] Whether any central Leibniz algebra give rise to a linear rack.
  \item[{\bf Q2:}] Whether any central Leibniz 2-algebra give rise to a linear 2-rack. If there is certain connection, what is the relation between the corresponding solutions of the Zamolodchikov Tetrahedron equation constructed from previous sections.
\end{itemize}
In this section, we will show that if the underlying 2-vector space of a central Leibniz 2-algebra is splittable, then it gives rise to a linear 2-rack. Moreover, if the central Leibniz 2-algebra is splittable, then the corresponding solutions of the Zamolodchikov Tetrahedron equation constructed from the central Leibniz 2-algebra and the associated linear 2-rack are the same. Since any vector space is splittable, so we obtained that any central Leibniz algebra gives rise to a linear rack as a byproduct, and this result is an enhancement of the construction  by Lebed given in \cite[Proposition~4.7]{Lebed3}.

Let $V=(V_0,V_1,s,t,\iii)$ be a $2$-vector space.
For an arbitrary object $e\in V_0$, it is obvious that $\widetilde{e}=(\langle e\rangle,\langle\iii_e\rangle,s|_{\langle\iii_e\rangle},
t|_{\langle\iii_e\rangle},\iii|_{\langle e\rangle})$
is a $2$-vector space, where $\langle e\rangle$ and $\langle\iii_e\rangle$ are the subvector spaces generated by $e$ and $\iii_e$ respectively, $s|_{\langle\iii_e\rangle}$ is the restriction of the linear map $s$ to ${\langle\iii_e\rangle}$, similarly for
$t|_{\langle\iii_e\rangle}$ and $\iii|_{\langle e\rangle}$.
Moreover, $\overline{V}=(V_0/{\langle e\rangle},V_1/{\langle\iii_e\rangle},
\overline{s},\overline{t},\overline{\iii})$ is also a $2$-vector space, where $\overline{s},\overline{t},\overline{\iii}$ are defined as follows:
$$\overline{s}(\overline{f})=\overline{s(f)},\quad
\overline{t}(\overline{f})=\overline{t(f)},\quad
\overline{\iii}(\overline{x})=\overline{\iii(x)},\quad
\forall \overline{x}\in V_0/{\langle e\rangle}
,~\overline{f}\in V_1/{\langle\iii_e\rangle}.$$
Let $\pi_0:V_0\to V_0/{\langle e\rangle}$ and $\pi_1:V_1\to V_1/{\langle\Id_e\rangle}$ be the canonical projections.
Then the definitions of $\overline{s},\overline{t},\overline{\iii}$ show that
\begin{eqnarray}\label{overline-sti}
\overline{s}\circ\pi_1=\pi_0\circ s,\quad
\overline{t}\circ\pi_1=\pi_0\circ t,\quad
\overline{\iii}\circ\pi_0=\pi_1\circ\iii,
\end{eqnarray}
which implies that $\pi=(\pi_0,\pi_1)$ is a $2$-vector space homomorphism from ${V}$ to $\overline{V}$.

Let $(\huaL,[\cdot,\cdot],\huaJ,\swe)$ be a central Leibniz $2$-algebra.  Then $(\overline{\huaL},[\cdot,\cdot]_{\overline{\huaL}},\overline{\huaJ})$ is a Leibniz $2$-algebra, where $\overline{\huaL}=(\huaL_0/{\langle \swe\rangle},\huaL_1/{\langle\iii_\swe\rangle},
\overline{s},\overline{t},\overline{\iii})$ is the quotient 2-vector space, the linear functor $[\cdot,\cdot]_{\overline{\huaL}}$ and the linear natural isomorphism $\overline{\huaJ}$ are defined as follows:
\begin{eqnarray}
&&[\overline{x},\overline{y}]_{\overline{\huaL}}=[\pi_0(x),\pi_0(y)]_{\overline{\huaL}}=\pi_0([x,y])=\overline{[x,y]},\label{pi0}\\
&&[\overline{f},\overline{g}]_{\overline{\huaL}}=[\pi_1(f),\pi_1(g)]_{\overline{\huaL}}=\pi_1([f,g])=\overline{[f,g]},\label{pi1}\\
&&\overline{\huaJ}_{\overline{x},\overline{y},\overline{z}}=\overline{\huaJ}_{\pi_0(x),\pi_0(y),\pi_0(z)}=\pi_1(\huaJ_{x,y,z})=\overline{\huaJ_{x,y,z}},\label{huaJ}
\end{eqnarray}
where $\overline{x}=\pi_0(x), \overline{y}=\pi_0(y), \overline{z}=\pi_0(z)\in\overline{\huaL_0}$ and $\overline{f}=\pi_1(f), \overline{g}=\pi_1(g)\in\overline{\huaL_1}$.

Before we give the relation between central Leibniz 2-algebras and linear 2-racks, we introduce the following splittable condition.

\begin{defi}
\begin{itemize}
  \item[{\rm(i)}]Let $V=(V_0,V_1,s,t,\iii)$ be a $2$-vector space and $e\in V_0$.
If there exist a section $\sigma_0:V_0/{\langle e\rangle}\to V_0$ of $\pi_0$ and a section $\sigma_1:V_1/{\langle\iii_e\rangle}\to V_1$ of $\pi_1$ such that $\sigma=(\sigma_0,\sigma_1)$ is a $2$-vector space homomorphism from $\overline{V}$ to ${V}$,
then the 2-vector space $V$ is called {\bf splittable} with respect to the object $e$.

    \item[{\rm(ii)}]Let $(\huaL,[\cdot,\cdot],\huaJ,\swe)$ be a central Leibniz $2$-algebra. If there exist a section $\sigma_0:\huaL_0/{\langle \swe\rangle}\to \huaL_0$ of $\pi_0$ and a section $\sigma_1:\huaL_1/{\langle\iii_\swe\rangle}\to \huaL_1$ of $\pi_1$ such that $\sigma=(\sigma_0,\sigma_1)$ is a Leibniz $2$-algebra  homomorphism from $\overline{\huaL}$ to ${\huaL}$,
then the central Leibniz 2-algebra $(\huaL,[\cdot,\cdot],\huaJ,\swe)$ is called {\bf splittable} with respect to the central object $\swe$.
\end{itemize}
\end{defi}

Obviously, if a central Leibniz 2-algebra is splittable, then the underlying 2-vector space is automatically splittable, and the converse is not true. Now we are ready to give the main result in this section.

\begin{thm}\label{cen-lei-2-alg-to-lin-2-rack}
Let $(\huaL,[\cdot,\cdot],\huaJ,\swe)$ be a central Leibniz $2$-algebra. 
\begin{itemize}
  \item[{\rm(i)}]If the $2$-vector space $\huaL$ is splittable with respect to the object $\swe$, then there is a linear $2$-rack structure on the $2$-vector space $\huaL$.
    \item[{\rm(ii)}]If the central Leibniz $2$-algebra $(\huaL,[\cdot,\cdot],\huaJ,\swe)$ is splittable with respect to the central object $\swe$, then the solutions of the Zamolodchikov Tetrahedron equation constructed by these two structures are the same, which can be showed as following commutative diagram:
\begin{equation}\label{sol-cen-Lei=sol-lin-2-rack}
\begin{array}{l}
  \xymatrix@R=2.5pc@C=8.5pc{
  \txt{central Leibniz $2$-algebra\\$(\huaL,[\cdot,\cdot],\huaJ,\swe)$}
  \ar@{-->}[r]^-{\rm 2-v.~ s.~ splittable}
  \ar@{-->}[d]^-{{\rm Theorem}~\ref{cen-Lei-to-sol}}
  &\txt{linear $2$-rack\\$(\huaL,\Delta,\varepsilon,\lhd,\frkR)$}
  \ar@{-->}[d]^-{{\rm Theorem}~\ref{lin-2-rack-to-sol}}\\
  \txt{solution of the {\rm ZTE}\\$(B^{\rm Lei},Y^{\rm Lei})$}
  \ar@{==}[r]_-{\tiny\txt{\rm central Leibniz $2$-algebra splittable}}
  &\txt{solution of the {\rm ZTE}\\$(B^\lhd,Y^\lhd)$}
  }
  \end{array}
  \end{equation}
\end{itemize}
\end{thm}

\begin{proof}
First, for any $x\in\huaL_0$ and $f\in\huaL_1$, we define $\Delta:\huaL\to\huaL\otimes\huaL$ as follows:
\begin{eqnarray}
\Delta(x)&=&x\otimes\swe+\swe\otimes\sigma_0\pi_0(x)=x\otimes\swe+\swe\otimes\sigma_0(\overline{x}),\label{cen-lei-2-alg-coprod0}\\
\Delta(f)&=&f\otimes\iii_\swe+\iii_\swe\otimes\sigma_1\pi_1(f)=f\otimes\iii_\swe+\iii_\swe\otimes\sigma_1(\overline{f})\label{cen-lei-2-alg-coprod1}.
\end{eqnarray}
By direct calculation and the fact that $\sigma$ and $\pi$ are linear functors, for any $f,g\in\huaL_1$, we obtain that
\begin{eqnarray*}
&&\,\,~(s\otimes s)\circ\Delta=\Delta\circ s,\qquad\,\,\,\,
(t\otimes t)\circ\Delta=\Delta\circ t,\\
&&\Delta(g)\circ\Delta(f)=\Delta(g\circ f),\quad
(\iii\otimes\iii)\circ\Delta=\Delta\circ\iii,
\end{eqnarray*}
which implies that $\Delta$ is a linear functor.

It is obvious that $2$-vector spaces $\widetilde{e}=(\langle e\rangle,\langle\iii_e\rangle,s|_{\langle\iii_e\rangle},
t|_{\langle\iii_e\rangle},\iii|_{\langle e\rangle})$ and $\K=(\K,\K,\Id_\K,\Id_\K,\Id_\K)$ are isomorphic by a linear invertible functor $\Phi:\widetilde{\swe}\to\K$, which is defined by
\begin{eqnarray*}
\Phi(a\swe)=a,\quad\Phi({a\iii_\swe})=a,\quad\forall a\in\K,\,a\swe\in\langle\swe\rangle,\,a\iii_\swe\in\langle\iii_\swe\rangle.
\end{eqnarray*}
Then for any $x\in\huaL_0$ and $f\in\huaL_1$, we define $\varepsilon:\huaL\to\K$ as follows:
\begin{eqnarray}
\varepsilon(x)&=&\Phi(x-\sigma_0\pi_0(x))=\Phi(x-\sigma_0(\overline{x})),\label{cen-lei-2-alg-coun0}\\
\varepsilon(f)&=&\Phi(f-\sigma_1\pi_1(f))=\Phi(f-\sigma_1(\overline{f})).\label{cen-lei-2-alg-coun1}
\end{eqnarray}
Similarly, by direct calculation, $\varepsilon$ is a linear functor since $\Phi$, $\sigma$ and $\pi$ are linear functors.
Moreover, \eqref{coasso}-\eqref{counit} hold by the fact that $\pi\circ\sigma=\Id_{\overline{\huaL}}$.

Second, for any $x,y\in\huaL_0$ and $f,g\in\huaL_1$, we define $\lhd:\huaL\otimes\huaL\to\huaL$ as follows:
\begin{eqnarray}
x\lhd y&=&\Phi(y-\sigma_0\pi_0(y))x+\sigma_0[\pi_0(x),\pi_0(y)]_{\overline{\huaL}}
=\Phi(y-\sigma_0(\overline{y}))x+\sigma_0[\overline{x},\overline{y}]_{\overline{\huaL}},\label{cen-lei-2-alg-lhd0}\\
f\lhd g&=&\Phi(g-\sigma_1\pi_1(g))f+\sigma_1[\pi_1(f),\pi_1(g)]_{\overline{\huaL}}
=\Phi(g-\sigma_1(\overline{g}))f+\sigma_1[\overline{f},\overline{g}]_{\overline{\huaL}}.\label{cen-lei-2-alg-lhd1}
\end{eqnarray}
Since $\Phi$ and $[\cdot,\cdot]$ are linear functors, by direct calculation, we will obtain that $\lhd$ is a linear functor and conditions \eqref{rack-coproduct}, \eqref{rack-counit} hold.

Now we define a linear functor $\widetilde{\lhd}:\huaL\otimes\huaL\to\huaL$ as follows:
\begin{eqnarray*}
x~\widetilde{\lhd}~y&=&\Phi(y-\sigma_0\pi_0(y))x-\sigma_0[\pi_0(x),\pi_0(y)]_{\overline{\huaL}}
=\Phi(y-\sigma_0(\overline{y}))x-[\overline{x},\overline{y}]_{\overline{\huaL}},\\
f~\widetilde{\lhd}~g&=&\Phi(g-\sigma_1\pi_1(g))f-\sigma_1[\pi_1(f),\pi_1(g)]_{\overline{\huaL}}
=\Phi(g-\sigma_1(\overline{g}))f-[\overline{f},\overline{g}]_{\overline{\huaL}},
\end{eqnarray*}
which will make \eqref{rack-inv} hold.

Third, for any $x,y,z\in\huaL_0$, we have
\begin{eqnarray*}
(x\lhd y)\lhd z&=&(\Phi(y-\sigma_0(\overline{y}))x+\sigma_0[\overline{x},\overline{y}]_{\overline{\huaL}})\lhd z\\
&=&\Phi(z-\sigma_0(\overline{z}))\Phi(y-\sigma_0(\overline{y}))x+\Phi(z-\sigma_0(\overline{z}))\sigma_0[\overline{x},\overline{y}]_{\overline{\huaL}}\\
&&+\Phi(y-\sigma_0(\overline{y}))\sigma_0[\overline{x},\overline{z}]_{\overline{\huaL}}+\sigma_0[[\overline{x},\overline{y}]_{\overline{\huaL}},\overline{z}]_{\overline{\huaL}},\\
(x\lhd z_{(1)})\lhd (y\lhd z_{(2)})&=&(x\lhd z)\lhd (y\lhd\swe)+(x\lhd\swe)\lhd (y\lhd\sigma_0(\overline{z}))\\
&=&(\Phi(z-\sigma_0(\overline{z}))x+\sigma_0[\overline{x},\overline{z}]_{\overline{\huaL}})\lhd y
+x\lhd\sigma_0[\overline{y},\overline{z}]_{\overline{\huaL}}\\
&=&\Phi(z-\sigma_0(\overline{z}))\Phi(y-\sigma_0(\overline{y}))x+\Phi(y-\sigma_0(\overline{y}))\sigma_0[\overline{x},\overline{z}]_{\overline{\huaL}}\\
&&+\Phi(z-\sigma_0(\overline{z}))\sigma_0[\overline{x},\overline{y}]_{\overline{\huaL}}+\sigma_0[[\overline{x},\overline{z}]_{\overline{\huaL}},\overline{y}]_{\overline{\huaL}}\\
&&+\Phi(\sigma_0[\overline{y},\overline{z}]_{\overline{\huaL}}-\sigma_0\pi_0\sigma_0[\overline{y},\overline{z}]_{\overline{\huaL}})x
+\sigma_0[\overline{x},[\overline{y},\overline{z}]_{\overline{\huaL}}]_{\overline{\huaL}}\\
&=&\Phi(z-\sigma_0(\overline{z}))\Phi(y-\sigma_0(\overline{y}))x+\Phi(y-\sigma_0(\overline{y}))\sigma_0[\overline{x},\overline{z}]_{\overline{\huaL}}\\
&&+\Phi(z-\sigma_0(\overline{z}))\sigma_0[\overline{x},\overline{y}]_{\overline{\huaL}}
+\sigma_0\big([[\overline{x},\overline{z}]_{\overline{\huaL}},\overline{y}]_{\overline{\huaL}}+[\overline{x},[\overline{y},\overline{z}]_{\overline{\huaL}}]_{\overline{\huaL}}\big).
\end{eqnarray*}
Thus we define $\frkR_{x\otimes y\otimes z}:(x\lhd y)\lhd z\to(x\lhd z_{(1)})\lhd (y\lhd z_{(2)})$ as follows:
\begin{eqnarray}\label{cen-lei-2-alg-frkR}
\frkR_{x\otimes y\otimes z}&=&\iii_{\Phi(z-\sigma_0(\overline{z}))\Phi(y-\sigma_0(\overline{y}))x+\Phi(z-\sigma_0(\overline{z}))\sigma_0[\overline{x},\overline{y}]_{\overline{\huaL}}
+\Phi(y-\sigma_0(\overline{y}))\sigma_0[\overline{x},\overline{z}]_{\overline{\huaL}}}+\sigma_1\big(\overline{\huaJ}_{\overline{x}\otimes\overline{y}\otimes\overline{z}}\big).
\end{eqnarray}
It is straightforward to see that for any linear maps $f:x\to x'$, $g:y\to y'$ and $h:z\to z'$, we have
$$((f\lhd h_{(1)})\lhd(g\lhd h_{(2)}))\circ\frkR_{x\otimes y\otimes z}=\frkR_{x'\otimes y'\otimes z'}\circ((f\lhd g)\lhd h),$$
which implies that $\frkR$ defined as above is a linear natural isomorphism.
Since the linear natural isomorphism $\overline{\huaJ}$ satisfies \eqref{Jacob} and $\sigma_1$ preserves the composition, we will obtain that $\frkR$ satisfies the condition \eqref{lin-distri}.
Therefore, $(\huaL,\Delta,\varepsilon,\lhd,\frkR)$ is a linear $2$-rack.

Finally, for any $x,y,z\in\huaL_0$ and $f,g\in\huaL_1$, on the one hand, by Theorem \ref{lin-2-rack-to-sol}, a solution $(B^{\lhd},Y^{\lhd})$ of the Zamolochikov Tetrahedron equation, which is induced by the linear $2$-rack $(\huaL,\Delta,\varepsilon,\lhd,\frkR)$, is as follows:
\begin{eqnarray*}
B^{\lhd}(x\otimes y)&=&y_{(1)}\otimes(x {\lhd}~y_{(2)})
\overset{\eqref{cen-lei-2-alg-coprod0}}=y\otimes(x\lhd\swe)+\swe\otimes(x\lhd\sigma_0(\overline{y}))
\overset{\eqref{cen-lei-2-alg-lhd0}}=y\otimes x+\swe\otimes\sigma_0[\overline{x},\overline{y}]_{\overline{\huaL}},\\
B^{\lhd}(f\otimes g)&=&g_{(1)}\otimes(f {\lhd}~g_{(2)})
\overset{\eqref{cen-lei-2-alg-coprod1}}=g\otimes(f\lhd\iii_\swe)+\swe\otimes(f\lhd\sigma_1(\overline{g}))
\overset{\eqref{cen-lei-2-alg-lhd1}}=g\otimes f+\iii_\swe\otimes\sigma_1[\overline{f},\overline{g}]_{\overline{\huaL}},\\
Y^{\lhd}_{x\otimes y\otimes z}&=&\iii_{z_{(1)}}\otimes\iii_{y_{(1)}\lhd z_{(2)}}\otimes\frkR_{x\otimes y_{(2)}\otimes z_{(3)}}\\
&\overset{\eqref{cen-lei-2-alg-coprod0}}=&\iii_{z}\otimes\iii_{y\lhd\swe}\otimes\frkR_{x\otimes\swe\otimes\swe}
+\iii_{z}\otimes\iii_{\swe\lhd\swe}\otimes\frkR_{x\otimes\sigma_0(\overline{y})\otimes\swe}
+\iii_{\swe}\otimes\iii_{y\lhd\sigma_0(\overline{z})}\otimes\frkR_{x\otimes\swe\otimes\swe}\\
&&+\iii_{\swe}\otimes\iii_{\swe\lhd\sigma_0(\overline{z})}\otimes\frkR_{x\otimes\sigma_0(\overline{y})\otimes\swe}
+\iii_{\swe}\otimes\iii_{y\lhd\swe}\otimes\frkR_{x\otimes\swe\otimes\sigma_0(\overline{z})}
+\iii_{\swe}\otimes\iii_{\swe\lhd\swe}\otimes\frkR_{x\otimes\sigma_0(\overline{y})\otimes\sigma_0(\overline{z})}\\
&\overset{\eqref{cen-lei-2-alg-lhd0},\eqref{cen-lei-2-alg-frkR}}=&
\iii_{z}\otimes\iii_{y}\otimes\iii_{x}
+\iii_{z}\otimes\iii_{\swe}\otimes\iii_{\sigma_0[\overline{x},\overline{y}]_{\overline{\huaL}}}
+\iii_{\swe}\otimes\iii_{\sigma_0[\overline{y},\overline{z}]}\otimes\iii_{x}\\
&&+\iii_{\swe}\otimes\iii_{y}\otimes\iii_{\sigma_0[\overline{x},\overline{z}]_{\overline{\huaL}}}
+\iii_\swe\otimes\iii_\swe\otimes\sigma_1(\overline{\huaJ}_{\overline{x}\otimes\overline{y}\otimes\overline{z}})\\
&=&\iii_{z\otimes y\otimes x+z\otimes\swe\otimes\sigma_0[\overline{x},\overline{y}]_{\overline{\huaL}}+\swe\otimes\sigma_0[\overline{y},\overline{z}]_{\overline{\huaL}}\otimes{x}
+\swe\otimes{y}\otimes\sigma_0[\overline{x},\overline{z}]_{\overline{\huaL}}}
+\iii_\swe\otimes\iii_\swe\otimes\sigma_1(\overline{\huaJ}_{\overline{x}\otimes\overline{y}\otimes\overline{z}}).
\end{eqnarray*}
On the other hand,
by Theorem \ref{cen-Lei-to-sol}, a solution $(B^{\rm Lei},Y^{\rm Lei})$, which is induced by the central Leibniz $2$-algebra $(\huaL,[\cdot,\cdot],\huaJ)$, is as follows:
\begin{eqnarray*}
B^{\rm Lei}(x\otimes y)&=&y\otimes x+\swe\otimes[x,y],\\
B^{\rm Lei}(f\otimes g)&=&g\otimes f+\iii_\swe\otimes[f,g],\\
Y^{\rm Lei}_{x\otimes y\otimes z}&=&\iii_{z\otimes y\otimes x+\swe\otimes[y,z]\otimes x+\swe\otimes y\otimes[x,z]+z\otimes \swe\otimes[x,y]}+\iii_\swe\otimes\iii_\swe\otimes\huaJ_{x\otimes y\otimes z}.
\end{eqnarray*}
Since $(\sigma_0,\sigma_1)$ is a Leibniz $2$-algebra homomorphism from $\overline{\huaL}$ to $\huaL$, we have
\begin{eqnarray}
[\sigma_0(\overline{x}),\sigma_0(\overline{y})]&=&\sigma_0[\overline{x},\overline{y}]_{\overline{\huaL}},
\qquad\forall \overline{x},\overline{y}\in\overline{\huaL_0}.\label{sigma0}\\
{[}\sigma_1(\overline{f}),\sigma_1(\overline{g})]&=&{\sigma_1[\overline{f},\overline{g}]_{\overline{\huaL}}},
\qquad\forall \overline{f},\overline{g}\in\overline{\huaL_1}.\label{sigma1}
\end{eqnarray}
It is obvious that $x-\sigma_0(x),~y-\sigma_0(y)\in\langle\swe\rangle$ and $f-\sigma_1(f),~g-\sigma_1(g)\in\langle\iii_\swe\rangle$, we assume that
\begin{eqnarray*}
x-\sigma_0(\overline{x})&=&a\swe,\quad \,\,\, y-\sigma_0(\overline{y})\,\,\,\,=\,\,\,\,b\swe,\\
f-\sigma_1(\overline{f})&=&a\iii_\swe,\quad g-\sigma_1(\overline{g})\,\,\,\,=\,\,\,\,b\iii_\swe,
\end{eqnarray*}
where $a,b\in\K$. Then we have
\begin{eqnarray}
[x,y]&=&[\sigma_0(\overline{x})+a\swe,\sigma_0(\overline{y})+b\swe]\label{two-bracket0}\\
&=&[\sigma_0(\overline{x}),\sigma_0(\overline{y})]+[a\swe,\sigma_0(\overline{y})]
+[\sigma_0(\overline{x}),b\swe]+[a\swe,b\swe]\nonumber\\
&\overset{\eqref{central-s}}=&[\sigma_0(\overline{x}),\sigma_0(\overline{y})]
\overset{\eqref{sigma0}}=\sigma_0[\overline{x},\overline{y}]_{\overline{\huaL}},\nonumber\\
{[}f,g]&=&[\sigma_1(\overline{f})+a\iii_\swe,\sigma_1(\overline{g})+b\iii_\swe]\label{two-bracket1}\\
&=&[\sigma_1(\overline{f}),\sigma_1(\overline{g})]+[a\iii_\swe,\sigma_1(\overline{g})]
+[\sigma_1(\overline{f}),b\iii_\swe]+[a\iii_\swe,b\iii_\swe]\nonumber\\
&\overset{\eqref{central-element}}=&[\sigma_1(\overline{f}),\sigma_1(\overline{g})]
\overset{\eqref{sigma1}}=\sigma_1[\overline{f},\overline{g}]_{\overline{\huaL}}.\nonumber
\end{eqnarray}
Moreover, we have
\begin{equation*}
\xymatrix@R=1.7pc@C=2pc{
 \txt{$[[x,y],z]$}
 \ar[rr]^-{\huaJ_{x\otimes y\otimes z}}
 \ar@{=}[d]_-{\eqref{two-bracket0}}
 &
 &\txt{$[[x,z],y]+[x,[y,z]]$}
 \ar@{=}[d]^-{\eqref{two-bracket0}}\\
 \txt{$\sigma_0([[\overline{x},\overline{y}]_{\overline{\huaL}},\overline{z}]_{\overline{\huaL}})$}
 \ar[rr]^-{\sigma_1({\overline{\huaJ}}_{\overline{x}\otimes\overline{y}\otimes\overline{z}})}
 & &\txt{$\sigma_0([[\overline{x},\overline{z}]_{\overline{\huaL}},\overline{y}]_{\overline{\huaL}}+[\overline{x},[\overline{y},\overline{z}]_{\overline{\huaL}}]_{\overline{\huaL}})$}
 }
 \end{equation*}
which implies that ${\huaJ_{x\otimes y\otimes z}}={\sigma_1({\overline{\huaJ}}_{\overline{x}\otimes\overline{y}\otimes\overline{z}})}$.
Therefore, we conclude that $(B^{\lhd},Y^{\lhd})$ and $(B^{\rm Lei},Y^{\rm Lei})$ are the same, i.e. diagram \eqref{sol-cen-Lei=sol-lin-2-rack} commutes.
\end{proof}

By Proposition~\ref{decate-Leibniz-2-alg} and Corollary~\ref{decate-lin-2-rack}, we obtain the following result, which is an enhancement of the construction given in  \cite[Proposition~4.7]{Lebed3}, where the author showed that there is a linear rack structure on $\huaL\oplus\mathbb K$ associated to any Leibniz algebra $\huaL$.

\begin{cor}\label{cor:cl-lr}
Any central Leibniz algebra gives rise to a linear rack structure on the underlying vector space. Moreover, we have the following commutative diagram:
\begin{center}
\begin{displaymath}
\xymatrix@C=3ex@R=0.5ex{
  \txt{\rm central Leibniz $2$-algebras\\$(\huaL,[\cdot,\cdot],\huaJ,\swe)$}
  \ar@{-->}[rr]^-{\rm Theorem~\ref{cen-lei-2-alg-to-lin-2-rack}}_-{\rm 2-v.~ s.~ splittable}
  \ar@{-->}[dd]_-{{\rm Decate.}}^-{\text{{\rm Proposition~\ref{decate-Leibniz-2-alg}}}}
  &
  &\txt{\rm linear $2$-racks\\$(\huaL,\Delta,\varepsilon,\lhd,\frkR)$}
 \ar@{-->}[dd]_-{{\rm Decate.}}^-{\text{{\rm Corollary~\ref{decate-lin-2-rack}}}}\\
  &\rotatebox{165}{{\txt{\Huge $\circlearrowright$}}}&\\
  \txt{\rm central Leibniz algebra\\ $(\overline{\huaL_0},[\cdot,\cdot]_{\overline{\huaL_0}},\overline{\swe})$}
 \ar@{-->}[rr]
  &
  &\txt{\rm linear rack\\ $(\overline{\huaL_0},\overline{\Delta},\overline{\varepsilon},{\overline{\lhd}})$}
}
\end{displaymath}
\end{center}
\end{cor}
\begin{proof}
Let $(\g,[\cdot,\cdot]_\g,e)$ be a central Leibniz algebra, $\overline{\g}=\g/\langle e\rangle$ the quotient vector space and $\pi:\g\to\overline{\g}$ the canonical projection.
Then $\overline{\g}=\g/\langle e\rangle$ with the bracket $[\cdot,\cdot]_{\overline{\g}}$ defined as follows is also a Leibniz algebra:
\begin{eqnarray*}
[\overline{x},\overline{y}]_{\overline{\g}}=\overline{[x,y]},
\qquad\forall\overline{x},\overline{y}\in\overline{\g}.
\end{eqnarray*}
Let $\sigma:\overline{\g}\to\g$ be a section of the projection $\pi$, i.e. $\pi\circ\sigma=\Id$.
Define linear maps $\Delta_\g:\g\to\g\otimes\g$, $\varepsilon_\g:\g\to\K$, $\lhd_\g:\g\otimes\g\to\g$ and $\widetilde{\lhd}_\g:\g\otimes\g\to\g$ as follows:
\begin{eqnarray}
\Delta_\g(x)&=&x\otimes\swe+\swe\otimes\sigma(\overline{x}),\label{cen-lei-alg-to-lin-rack-copro}\\
\varepsilon_\g(x)&=&\Phi\big(x-\sigma(\overline{x})\big),\label{cen-lei-alg-to-lin-rack-coun}\\
x\lhd_\g y&=&\Phi\big(y-\sigma(\overline{y})\big)x+\sigma[\overline{x},\overline{y}]_{\overline{\g}},\label{cen-lei-alg-to-lin-rack-lhd}\\
x~\widetilde{\lhd}_\g~ y&=&\Phi\big(y-\sigma(\overline{y})\big)x-\sigma[\overline{x},\overline{y}]_{\overline{\g}},\label{cen-lei-alg-to-lin-rack-lhd-inv}
\end{eqnarray}
where $x,y\in\g$ and the linear invertible map $\Phi:\langle\swe\rangle\to\K$ define by $\Phi(a\swe)=a$ for any $a\in\K$.
Then we will obtain that $(\g,\Delta_\g,\varepsilon_\g,\lhd_\g)$ is a linear rack by the facts that $\pi\circ\sigma=\Id$ and $[\cdot,\cdot]_{\overline{\g}}$ satisfies the Leibniz identity.
It is obvious that the above diagram is commutative.
\end{proof}

 By Proposition~\ref{decate-Leibniz-2-alg},
 Corollary~\ref{decate-lin-2-rack},
 Theorem~\ref{cen-lei-2-alg-to-lin-2-rack} and Corollary~\ref{cor:cl-lr}, we obtain the diagram \eqref{diagram:main}.

\begin{ex}
Consider the central Leibniz $2$-algebra $(\swg,[\cdot,\cdot],\huaJ,e)$ shown in Example \ref{cen-Lei-to-cen-Lei-2-alg}, where
$\swg=(\g,\g\oplus\R,s,t,\iii)$ is a $2$-vector space and $(\g,[\cdot,\cdot]_\g,e)$ is a central Leibniz algebra.
Let $\pi_0:\g\to\g/\langle e\rangle$ be the canonical projection. Choose a section $\sigma_0:\g/\langle e\rangle\to\g$ of $\pi_0$, i.e. $\pi_0\circ\sigma_0=\Id_{\g/\langle e\rangle}$. Then by Corollary \ref{cor:cl-lr}, there is a linear rack $(\g,\Delta_\g,\varepsilon_\g,\lhd_\g)$, which is defined by \eqref{cen-lei-alg-to-lin-rack-copro}-\eqref{cen-lei-alg-to-lin-rack-lhd}.

As we discussed at the beginning of this section, using \eqref{overline-sti}-\eqref{huaJ}, we obtain a $2$-vector space $\overline{\swg}=(\g/\langle e\rangle,\g\oplus\R/\langle(e,0)\rangle,\overline{s},\overline{t},\overline{\iii})$ and a Leibniz $2$-algebra $(\overline{\swg},[\cdot,\cdot]_{\overline{\swg}},\overline{\huaJ})$.
By identifying the vector space $\g\oplus\R/\langle(e,0)\rangle$ with $\g/\langle e\rangle\oplus\R$, we define two linear maps $\pi_1:\g\oplus\R\to\g\oplus\R/\langle(e,0)\rangle$ and $\sigma_1:\g\oplus\R/\langle(e,0)\rangle\to\g\oplus\R$ as follows:
\begin{eqnarray*}
\pi_1(x,a)&=&(\pi_0(x),a)=(\overline{x},a),\quad\forall(x,a)\in\g\oplus\R,\\
\sigma_1(\overline{x},a)&=&(\sigma_0(\overline{x}),a),\qquad\qquad\,\forall(\overline{x},a)\in\g\oplus\R/\langle(e,0)\rangle.
\end{eqnarray*}
It is obvious that $\pi_1\circ\sigma_1=\Id_{\g\oplus\R/\langle(e,0)\rangle}$, $\pi=(\pi_0,\pi_1):\swg\to\overline{\swg}$ and $\sigma=(\sigma_0,\sigma_1):\overline{\swg}\to{\swg}$ are $2$-vector space homomorphisms, which implies that the $2$-vector space $\swg$ is splittable with respect to the object $e$.
Then by Theorem \ref{cen-lei-2-alg-to-lin-2-rack}, there is a linear $2$-rack $(\swg,\Delta,\varepsilon,\lhd,\frkR)$, which is defined by \eqref{cen-lei-2-alg-coprod0}-\eqref{cen-lei-2-alg-frkR}.

By Proposition \ref{decate-Leibniz-2-alg}, it is easy to see the decategorification of the central Leibniz $2$-algebra $(\swg,[\cdot,\cdot],\huaJ,e)$ is the central Leibniz algebra $(\g,[\cdot,\cdot]_\g,e)$ since the isomorphic class of an arbitrary object $x\in\g$ only contains itself. Moreover, the decategorification of the induced linear $2$-rack $(\swg,\Delta,\varepsilon,\lhd,\frkR)$ is the induced linear rack $(\g,\Delta_\g,\varepsilon_\g,\lhd_\g)$.

Furthermore, there is also a Leibniz algebra structure on the vector space ${\g/\langle e\rangle}$. If the section $\sigma_0:{\g/\langle e\rangle}\to\g$ is a Leibniz algebra homomorphism, that is, $\sigma_0[\overline{x},\overline{y}]_{\g/\langle e\rangle}=[\sigma_0(\overline{x}),\sigma_0(\overline{g})]_\g$,
there is no doubt that $\sigma=(\sigma_0,\sigma_1):\overline{\swg}\to\swg$ is a Leibniz $2$-algebra homomorphism, which implies that the central Leibniz $2$-algebra $\swg$ is splittable with respect to the central object $e$.
At this moment, on the one hand, the solutions of the Yang-Baxter equation $\huaB^{\rm Lei}$ and  $\huaB^{\lhd_\g}$ induced respectively by the central Leibniz algebra $(\g,[\cdot,\cdot]_\g,e)$ and the induced linear rack $(\g,\Delta_\g,\varepsilon_\g,\lhd_\g)$ are the same, which are actually \eqref{cen-Leibniz-sol}. On the other hand, the solutions of the Zamolodchikov Tetrahedron equation $(B^{\rm Lei},Y^{\rm Lei})$ and  $(B^{\rm \lhd},Y^{\rm \lhd})$ induced respectively by the central Leibniz $2$-algebra $(\swg,[\cdot,\cdot],\huaJ,e)$ and the induced linear $2$-rack $(\swg,\Delta,\varepsilon,\lhd,\frkR)$ are the same, which are actually \eqref{cen-lei-to-cen-lei-2-alg-sol-B0}-\eqref{cen-lei-to-cen-lei-2-alg-sol-Y}.
In addition, the decategorification of the solution $(B^{\rm Lei},Y^{\rm Lei})=(B^{\rm \lhd},Y^{\rm \lhd})$ is actually the solution $\huaB^{\rm Lei}=\huaB^{\lhd_\g}$.
Using this example, one can get a more concrete sense of the commutativity of the diagram \eqref{diagram:main}.
\end{ex}

\section{Linear $2$-racks and $2$-racks}\label{sec:r}

In \cite{Cr}, the author gave rise to the notion of 2-racks as a small category $X$ equipped with
two functors $\rhd:X\times X\to X$ and $\lhd:X\times X\to X$,
four natural isomorphisms $\frkL_{x,y,z}:x\rhd(y\rhd z)\to (x\rhd y)\rhd(x\rhd z)$,
$\frkR_{x,y,z}:(x\lhd y)\lhd z\to (x\lhd z)\lhd(y\lhd z)$,
$L_{x,y}:x\rhd(y\lhd x)\to y$ and $R_{x,y}:(x\rhd y)\lhd x\to y$,
where $x,y,z\in X_0$, such that the following conditions hold:
\begin{itemize}
  \item (left) distributor identity:
  \begin{eqnarray*}
  \Big(\frkL_{w,z,y}\rhd\frkL_{w,z,x}\Big)\circ\frkL_{w,z\rhd y,z\rhd x}\circ\Big(\iii_w\rhd\frkL_{z,y,x}\Big)
  &=&\frkL_{w\rhd z,w\rhd y,w\rhd x}\circ\Big((\iii_w\rhd \iii_z)\rhd\frkL_{w,y,x}\Big)\rhd\frkL_{w,z,y\rhd x},
  \end{eqnarray*}
  \item (right) distributor identity:
  \begin{eqnarray*}
  \Big(\frkR_{x,z,w}\lhd\frkR_{y,z,w}\Big)\circ\frkR_{x\lhd z,y\lhd z,w}\circ\Big(\frkR_{x,y,z}\lhd\iii_w\Big)
  &=&\frkR_{x\lhd w,y\lhd w,z\lhd w}\circ\Big(\frkR_{x,y,w}\lhd(\iii_z\lhd \iii_w)\Big)\circ\frkR_{x\lhd y,z,w},
  \end{eqnarray*}
  \item additional coherence laws,
\end{itemize}
where the whole description of the above coherence laws is complicated, as functors $\rhd$ and $\lhd$ induce
$2^3 = 8$ possible ways of performing the third Reidemeister move.

In this section, we introduce the notion of a semistrict $2$-rack, which is a more strict version of the above one and used to describe group-like category of a linear $2$-rack.

\begin{defi}\label{def-2-rack}
A (right) {\bf semistrict $2$-rack} consists of
\begin{itemize}
  \item a small category $X=(X_0,X_1,s,t,\iii,\circ)$, i.e. an internal category in {\sf Set};
  \item a functor $\lhd:X\times X\to X$ such that,
  for any $x\in X_0$, $\bullet\lhd x:X\to X$ is an invertible functor, where for any $y\in X_0$ and $f\in X_1$, $\bullet\lhd x$ is defined as follows:
  $$(\bullet\lhd x)(y)=y\lhd x,\quad(\bullet\lhd x)(f)=f\lhd\iii_x.$$
  \item a natural isomorphism $\frkR_{x,y,z}:(x\lhd y)\lhd z\to (x\lhd z)\lhd(y\lhd z)$ satisfies the following (right) distributor identity for any $x,y,z,w\in X_0$:
  \begin{eqnarray}\label{2-rack-identity}
 \quad\quad\Big( \frkR_{x,z,w}\lhd\frkR_{y,z,w}\Big)\circ\frkR_{x\lhd z,y\lhd z,w}\circ\Big(\frkR_{x,y,z}\lhd \iii_w\Big)
  &=&\frkR_{x\lhd w,y\lhd w,z\lhd w}\circ\Big(\frkR_{x,y,w}\lhd(\iii_z\lhd \iii_w)\Big)\circ\frkR_{x\lhd y,z,w}.
  \end{eqnarray}
\end{itemize}
The above identity can be described as the following commutative diagram:
{\small\begin{equation*}%
\xymatrix@R=1.5pc@C=0.1pc{
 &
 \text{$\Big((x\lhd y)\lhd z\Big)\lhd w$}
 \ar[dl]_-{\text{$\frkR_{x,y,z}\lhd \iii_w$}}
 \ar[dr]^-{\text{$\frkR_{x\lhd y,z,w}$}}
 &\\
 \text{$\Big((x\lhd z)\lhd(y\lhd z)\Big)\lhd w$}
 \ar[d]_-{\text{$\frkR_{x\lhd z,y\lhd z,w}$}}
 &
 &\text{$\Big((x\lhd y)\lhd w\Big)\lhd(z\lhd w)$}
 \ar[d]^-{\text{$\frkR_{x,y,w}\lhd(\iii_z\lhd\iii_w)$}}\\
 \text{$\Big((x\lhd z)\lhd w\Big)\lhd\Big((y\lhd z)\lhd w\Big)$}
 \ar[dr]_-{\text{$\frkR_{x,z,w}\lhd\frkR_{y,z,w}$}}
 &
 &\text{$\Big((x\lhd w)\lhd(y\lhd w)\Big)\lhd(z\lhd w)$}
 \ar[dl]^-{\text{\quad$\frkR_{x\lhd w,y\lhd w,z\lhd w}$}}\\
 &
 \text{$\Big((x\lhd w)\lhd(z\lhd w)\Big)\lhd\Big((y\lhd w)\lhd(z\lhd w)\Big)$}
 &
 }
 \end{equation*}}

Moreover, if the natural isomorphism $\frkR$ is the identity natural isomorphism, then $(X,\lhd)$ is called a {\bf strict $2$-rack }.
\end{defi}

\begin{pro}
Let $(X,\lhd)$ be a strict $2$-rack, where $X=(X_0,X_1,s,t,\iii,\circ)$ is a small category. Then $(X_0,\lhd)$ and $(X_1,\lhd)$ are racks.
\end{pro}

\begin{proof}
Since $\bullet\lhd x:X\to X$ is an invertible functor for any $x\in X_0$,
there is a functor $\bullet~\widetilde{\lhd}~x:X\to X$ such that
$(\bullet~\widetilde{\lhd}~x)\lhd x=\Id_X$ and $(\bullet\lhd x)~\widetilde{\lhd}~x=\Id_X$, that is,
\begin{eqnarray*}
&&\quad(y\lhd x)~\widetilde{\lhd}~x=y,\qquad\quad(y~\widetilde{\lhd}~x)\lhd x=y,\qquad\,\forall y\in X_0,\\
&&\,(f\lhd\iii_x)~\widetilde{\lhd}~\iii_x=f,\qquad(f~\widetilde{\lhd}~\iii_x)\lhd \iii_x=f,\qquad\forall f\in X_1.
\end{eqnarray*}
Moreover, since $\frkR=\Id$, for any $x,y,z\in X_0$, we have
$$(x\lhd y)\lhd z=(x\lhd z)\lhd(y\lhd z),$$
which implies that $(X_0,\lhd)$ is a rack.

For any morphism $g:x\to y$, there is a map $\bullet\lhd g:X_1\to X_1$, where for any map $f:a\to b$ in $X_1$, $(\bullet\lhd g)(f)=f\lhd g$.
we will see that
$\bullet~\widetilde{\lhd}~g:X_1\to X_1$ defined as follows is the inverse of the map $\bullet{\lhd}g$:
$$
(\bullet~\widetilde{\lhd}~g)(f)=f~\widetilde{\lhd}~g
=(\bullet~\widetilde{\lhd}~y)\circ f\circ(\bullet{\lhd}x):
a~\widetilde{\lhd}~x\to b~\widetilde{\lhd}~y.
$$
More precisely, $f~\widetilde{\lhd}~g$ can be described as:
\begin{equation*}
\begin{array}{l}
\hspace{1mm}\xymatrix@R=2pc@C=2.5pc{
 \txt{$a~\widetilde{\lhd}~x$}
 \ar[r]^-{\bullet{\lhd}x}
 \ar@/_{2.3pc}/[rrr]!U_(.4){\qquad\quad f~\widetilde{\lhd}~g}
 & \txt{$(a~\widetilde{\lhd}~x)\lhd x=a$}
 \ar[r]^-{f}
 & \txt{$b$}
 \ar[r]^-{\bullet~\widetilde{\lhd}~y}
 & \txt{$b~\widetilde{\lhd}~y$.}
 &
 }
 \end{array}
 \end{equation*}
For any $f:a\to b$ and $g:x\to y$ in $X_1$, we have the following commutative diagram, where $f\lhd g$ is written as $(\iii_b\lhd g)\circ(f\lhd\iii_x)$ and the left bottom square implies that $(\iii_b\lhd g)~\widetilde{\lhd}~g=\iii_b$.
Then the blue arrows show that $(f\lhd g)~\widetilde{\lhd}~g=f$:
\begin{equation*}%
\xymatrix@R=2.5pc@C=3.5pc{
 \txt{$a$}
 \ar[r]^-{\bullet\lhd x}
 \ar[d]_-{f}
 & \txt{$a\lhd x$}
 \ar[r]^-{\bullet\widetilde{\lhd}x}
 \ar@[blue][d]_-{f\lhd\iii_x}
 & \txt{$(a\lhd x)~\widetilde{\lhd}~x=a$}
 \ar@[blue]@<0.9ex>[l]^-{\bullet\lhd x}
 \ar@/^{4.5pc}/@[blue][dd]^-{f}
 \ar[d]_-{(f\lhd\iii_x)~\widetilde{\lhd}~\iii_x=f}\\
 \txt{$b$}
 \ar[r]^-{\bullet\lhd x}
 \ar[d]_-{\iii_b}
 &\txt{$b\lhd x$}
 \ar[r]^-{\bullet\widetilde{\lhd}x}
 \ar@[blue][d]_-{\iii_b\lhd g}
 & \txt{$(b\lhd x)~\widetilde{\lhd}~x=b$}
 \ar[d]_-{(\iii_b\lhd g)~\widetilde{\lhd}~g=\iii_b}\\
 \txt{$b$}
 \ar[r]^-{\bullet\lhd y}
 & \txt{$b\lhd y$}
 \ar@[blue][r]^-{\bullet\widetilde{\lhd}y}
 & \txt{$(b\lhd y)~\widetilde{\lhd}~y=b$}
 }
 \end{equation*}
One can also obtain $(f~\widetilde{\lhd}~g)\lhd g=f$ in a similar way.
Therefore, we obtain that $\bullet{\lhd}g:X_1\to X_1$ is always invertible for any $g\in X_1$. For any maps $f:x\to x'$, $g:y\to y'$ and $h:z\to z'$, since $\frkR=\Id$, we have
\begin{equation*}%
\xymatrix@R=2.5pc@C=2.5pc{
 \txt{$(x\lhd y)\lhd z$}
 \ar@{=}[r]
 \ar[d]_-{(f\lhd g)\lhd h}
 & \txt{$(x\lhd z)\lhd(y\lhd z)$}
 \ar[d]^-{(f\lhd h)\lhd(g\lhd h)}\\
 \txt{$(x'\lhd y')\lhd z'$}
 \ar@{=}[r]
 & \txt{$(x'\lhd z')\lhd(y'\lhd z')$}
 }
 \end{equation*}
which implies that $(f\lhd g)\lhd h=(f\lhd h)\lhd(g\lhd h)$.
Therefore, $(X_1,\lhd)$ is a rack too.
\end{proof}

\begin{rmk}
The definition of a strict $2$-rack in Definition \ref{def-2-rack} is consistent with the one given in \cite{CW},
where the authors defined a strict $2$-rack $X=(X_0,X_1,s,t,\iii,\circ)$ as an internal category in the category of racks,
i.e. $X_0$ and $X_1$ are racks and $s,t,\iii,\circ$ are rack morphisms.
\end{rmk}

It is well known that the set of group-like elements of a linear rack constitutes a rack. Now we introduce the notion of a group-like category of a linear $2$-rack and show that the group-like category of a linear $2$-rack is a semistrict $2$-rack.

\begin{defi}
Let $(V,\Delta,\varepsilon,\lhd,\frkR)$ be a linear $2$-rack, where $V=(V_0,V_1,s,t,\iii)$ is a $2$-vector space. If $x\in V_0$ satisfies
\begin{eqnarray}\label{gp-obj}
\Delta(x)=x\otimes x,
\end{eqnarray}
then we call $x$ a {\bf group-like object} of the linear $2$-rack $(V,\Delta,\varepsilon,\lhd,\frkR)$.
\end{defi}
We denote the set of all group-like objects by $V^\rmG_0$, that is,
$$V^{\rmG}_0=\{x\in V_0~|~\Delta(x)=x\otimes x\},$$
which is obviously a subset of the set $V_0$.
Define a subset $V_1^\rmG$ of the set $V_1$ by
$$V_1^\rmG=\{f\in V_1~|~s(f)\in V_0^\rmG~{\rm and}~t(f)\in V_0^\rmG\},$$
that is, $f\in V_1^\rmG$ if and only if
\begin{eqnarray}\label{gp-morp}
\Delta(s(f))=s(f)\otimes s(f),\quad\Delta(t(f))=t(f)\otimes t(f).
\end{eqnarray}
We call $f\in V_1^\rmG$ a {\bf group-like morphism}.

\begin{thm}
Let $(V,\Delta,\varepsilon,\lhd,\frkR)$ be a linear $2$-rack, where $V=(V_0,V_1,s,t,\iii)$ is a $2$-vector space.
Then $V^\rmG=(V^\rmG_0,V^\rmG_1,s,t,\iii,\circ)$ is a small subcategory of the category $V$, where $s,t,\iii,\circ$ are inherited from the category $V$ and we call $V^\rmG$ the {\bf group-like category} of the linear $2$-rack $(V,\Delta,\varepsilon,\lhd,\frkR)$.
Moreover, $(V^\rmG,\lhd^\rmG,\frkR^\rmG)$ is a semistrict $2$-rack, where
$\lhd^\rmG$ and $\frkR^\rmG$ are obtained respectively from the functor $\lhd$ and the natural isomorphism $\frkR$ by forgetting their linear structures.
\end{thm}

\begin{proof}
First, in order to show that $V^\rmG=(V^\rmG_0,V^\rmG_1,s,t,\iii,\circ)$ is a small subcategory of the category $V$, we need to show that, for any composable maps $f,g\in V^\rmG_1$ and any object $x\in V^\rmG_0$, $g\circ f\in V^\rmG_1$ and $\iii_x\in V^\rmG_1$ hold.
Indeed, by the compatible condition between the source map, the target map and the composition, we have
\begin{eqnarray*}
\Delta(s(g\circ f))&=&\Delta(s(f))
=s(f)\otimes s(f)
=s(g\circ f)\otimes s(g\circ f),\\
\Delta(t(g\circ f))&=&\Delta(t(g))
=t(g)\otimes t(g)
=t(g\circ f)\otimes t(g\circ f).
\end{eqnarray*}
Then we have $s(g\circ f),t(g\circ f)\in V^\rmG_0$, which implies $g\circ f\in V^\rmG_1$.
For $x\in V^\rmG_0$, since $s(\iii_x)=t(\iii_x)=x\in V^\rmG_0$, we have $\iii_x\in V^\rmG_1$.

Second, we show that the functor $\lhd^\rmG$ and the natural isomorphism $\frkR^\rmG$ are well-defined on the small category $V^\rmG$.
That is, for any $x,y,z\in V^\rmG_0$ and $f,g\in V^\rmG_1$, we need to show that $x\lhd^\rmG y\in V^\rmG_0$, $f\lhd^\rmG g\in V^\rmG_1$ and $\frkR^\rmG_{x,y,z}\in V^\rmG_1$. Indeed, by \eqref{rack-coproduct}, we have
\begin{eqnarray*}
&&\Delta(x\lhd^\rmG y)=\Delta(x\lhd y)
\overset{\eqref{rack-coproduct},\eqref{gp-obj}}
=(x\lhd y)\otimes(x\lhd y)
=(x\lhd^\rmG y)\otimes(x\lhd^\rmG y),
\end{eqnarray*}
which implies that $x\lhd^\rmG y\in V^\rmG_0$.
Since $\lhd$ preserves the source and target maps, we have
\begin{eqnarray*}
\Delta(s(f\lhd^\rmG g))
&=&\Delta(s(f\lhd g))
=\Delta(s(f)\lhd s(g))
\overset{\eqref{rack-coproduct},\eqref{gp-morp}}
=(s(f)\lhd s(g))\otimes(s(f)\lhd s(g))\\
&=&(s(f)\lhd^\rmG s(g))\otimes(s(f)\lhd^\rmG s(g)),\\
\Delta(t(f\lhd^\rmG g))
&=&\Delta(t(f\lhd g))
=\Delta(t(f)\lhd t(g))
\overset{\eqref{rack-coproduct},\eqref{gp-morp}}
=(t(f)\lhd t(g))\otimes(t(f)\lhd t(g))\\
&=&(t(f)\lhd^\rmG t(g))\otimes(t(f)\lhd^\rmG t(g)),
\end{eqnarray*}
which implies that $f\lhd^\rmG g\in V_1^\rmG$.
Moreover, since $x\lhd^\rmG y\in V^\rmG_0$ for any $x,y\in V^\rmG_0$, we have
$s(\frkR^\rmG_{x,y,z})=(x\lhd^\rmG y)\lhd^\rmG z\in V_0^\rmG$ and $t(\frkR^\rmG_{x,y,z})=(x\lhd^\rmG z)\lhd^\rmG(y\lhd^\rmG z)\in V_0^\rmG$, which implies that $\frkR^\rmG_{x,y,z}\in V_1^\rmG$.

Third, we show that, for any $x\in V_0^\rmG$, $\bullet\lhd^\rmG x:V^\rmG\to V^\rmG$ is an invertible functor.
Denote by $\widetilde{\lhd^\rmG}$ the linear functor $\widetilde{\lhd}$ forgetting its linear structure, where $\widetilde{\lhd}$ satisfies \eqref{rack-inv}.
Since the linear functor $\Delta$ preserves the identity map, for any $x\in V_0^\rmG$, we have $\Delta(\iii_x)=\iii_{\Delta(x)}=\iii_{x\otimes x}=\iii_x\otimes\iii_x$.
Then by \eqref{counit}, we have
\begin{eqnarray*}\label{counit-gp-elem}
&&\varepsilon(x)=1,\quad\varepsilon(\iii_x)=1,\quad\forall x\in V^\rmG_0.
\end{eqnarray*}
Thus for any $y\in V^\rmG_0$ and $f\in V_1^\rmG$, we have
\begin{eqnarray*}
(y\lhd^\rmG x)~\widetilde{\lhd^\rmG}~x
&=&(y\lhd x)~\widetilde{\lhd}~x
\overset{\eqref{rack-inv}}=(\Id\otimes\varepsilon)(y\otimes x)=y,\\
(y~\widetilde{\lhd^\rmG}~x)\lhd^\rmG x
&=&(y~\widetilde{\lhd}~x)\lhd x
\overset{\eqref{rack-inv}}=(\Id\otimes\varepsilon)(y\otimes x)
=y,\\
(f\lhd^\rmG\iii_x)~\widetilde{\lhd^\rmG}~\iii_x
&=&(f\lhd\iii_x)~\widetilde{\lhd}~\iii_x
\overset{\eqref{rack-inv}}=(\Id\otimes\varepsilon)(f\otimes \iii_x)
=f,\\
(f~\widetilde{\lhd^\rmG}~\iii_x){\lhd^\rmG}\iii_x
&=&(f\lhd\iii_x)~\widetilde{\lhd}~\iii_x
\overset{\eqref{rack-inv}}=(\Id\otimes\varepsilon)(f\otimes \iii_x)
=f,
\end{eqnarray*}
which implies that $\bullet\lhd^\rmG x:V^\rmG\to V^\rmG$ is an invertible functor for any $x\in V_0^\rmG$.

Finally, it is easy to see that the condition \eqref{lin-distri} for $\frkR_{x,y,z}:(x\lhd y)\lhd z\to (x\lhd z)\lhd(y\lhd z)$ is actually the condition \eqref{2-rack-identity}.
Therefore, $(V^\rmG,\lhd^\rmG,\frkR^\rmG)$ is a semistrict $2$-rack.
\end{proof}

At the end of this section, we give an example of strict $2$-racks from a strict $2$-group $(G,\otimes,\III)$, which is a strict monoidal category where all objects are invertible with respect to $\otimes$ and all morphisms are invertible under the composition.
Here $G$ is a small category with a set of objects $G_0$ and a set of morphisms $G_1$. Denote the inverse of an object $g\in G_0$ with respect to $\otimes$ by $g^{\dagger}$, that is, $g\otimes g^{\dagger}=\III=g^{\dagger}\otimes g$.
For any morphism $\alpha:g\to h$, there is a morphism ${\iii_{g^{\dagger}}~\otimes~ \alpha^{-1}~\otimes~\iii_{h^{\dagger}}}$ from $g^{\dagger}$ to $h^{\dagger}$ as follows:
$$g^{\dagger}=g^{\dagger}\otimes\III=g^{\dagger}\otimes h\otimes h^{\dagger}\xrightarrow{\iii_{g^{\dagger}}~\otimes~ \alpha^{-1}~\otimes~\iii_{h^{\dagger}}}g^{\dagger}\otimes g\otimes h^{\dagger}=\III\otimes h^{\dagger}=h^{\dagger}.$$

\begin{pro}
Let $(G,\otimes,\III)$ be a strict $2$-group, $X$ a small category with a set of objects $X_0$ and a set of morphisms $X_1$, and $F:G\times X\to X$ a functor such that
\begin{eqnarray*}
F(g\otimes h, x)&=&F(g, F(h, x)),\qquad\forall g,h\in G_0~and~x\in X_0,\\
F(\alpha\otimes\beta, \xi)&=&F(\alpha, F(\beta, \xi)),\qquad\forall \alpha,\beta\in G_1 ~and~\xi\in X_1.
\end{eqnarray*}
Define $\lhd:(G\times X)\times(G\times X)\to G\times X$ for any objects $(g,x),(h,y)\in G_0\times X_0$ and any morphisms
$(\alpha,\xi):(g_1,x_1)\to(g_2,x_2),~(\beta,\eta):(h_1,y_1)\to(h_2,y_2)$ as follows:
\emptycomment{
\begin{equation*}
  \xymatrix@R=0.1pc@C=2.5pc{
  \txt{$\lhd:(G\times X)\times(G\times X)$} \ar[rr] && \txt{$G\times X\qquad\qquad\qquad\qquad\qquad\qquad\qquad\quad\,\,$}  \\
  \txt{$\big((g_1,x_1),(h_1,y_1)\big)$} \ar@{|->}[rr]
  \ar[ddddd]^{\big((\alpha,\xi),(\beta,\eta)\big)} && \txt{$(g_1,x_1)\lhd(h_1,y_1)=\big(h_1\otimes g_1\otimes h_1^{-1}, F(h_1,x_1)\big)$}
  \ar@<-19ex>[ddddd]^{(\alpha,\xi)\lhd(\beta,\eta)=\big(\beta\otimes\alpha\otimes(\iii_{h_1^{-1}}\otimes\beta^{-1}\otimes\iii_{h_2^{-1}}), F(\beta,\xi)\big)}\\
  &&\\&&\\&&\\&&\\
  \txt{$\big((g_2,x_2),(h_2,y_2)\big)$} \ar@{|->}[rr] && \txt{$(g_2,x_2)\lhd(h_2,y_2)=\big(h_2\otimes g_2\otimes h_2^{-1}, F(h_2,x_2)\big)$}
  }
\end{equation*}
}
\begin{eqnarray*}
(g,x)\lhd(h,y)&=&\big(h\otimes g\otimes h^{\dagger}, F(h,x)\big),\\
(\alpha,\xi)\lhd(\beta,\eta)&=&
\big(\beta\otimes\alpha\otimes(\iii_{h_1^{\dagger}}\otimes\beta^{-1}\otimes\iii_{h_2^{\dagger}}), F(\beta,\xi)\big).
\end{eqnarray*}
Then $(X,\lhd)$ is strict $2$-rack.
\end{pro}
\begin{proof}
First, we show that $\lhd$ defined as above is a functor, that is, $\lhd$ is compatible with the source map, the target map, the composition and the identity map.
Indeed, since the functor $F$ preserves the source and target maps, for any maps
$(\alpha,\xi):(g_1,x_1)\to(g_2,x_2),~(\beta,\eta):(h_1,y_1)\to(h_2,y_2)$, it is obvious that
\begin{eqnarray*}
s\big((\alpha,\xi)\lhd(\beta,\eta))
&=&s(\beta\otimes\alpha\otimes(\iii_{h_1^{\dagger}}\otimes\beta^{-1}\otimes\iii_{h_2^{\dagger}}), F(\beta,\xi)\big)\\
&=&(h_1\otimes g_1\otimes{h_1^{\dagger}}, F(h_1,x_1))\\
&=&(g_1,x_1)\lhd(h_1,y_1)
=s(\alpha,\xi)\lhd s(\beta,\eta),\\
t\big((\alpha,\xi)\lhd(\beta,\eta))
&=&t(\beta\otimes\alpha\otimes(\iii_{h_1^{\dagger}}\otimes\beta^{-1}\otimes\iii_{h_2^{\dagger}}),F(\beta,\xi))\\
&=&(h_2\otimes g_2\otimes{h_2^{\dagger}},F(h_2,x_2))\\
&=&(g_2,x_2)\lhd(h_2,y_2)
=t(\alpha,\xi)\lhd t(\beta,\eta),
\end{eqnarray*}
which implies that $\lhd$ is compatible with the source and target maps.
Given two composable maps
\begin{equation*}
\begin{array}{l}
\hspace{1mm}\xymatrix@R=5pc@C=5pc{
 \txt{$\big((g_1,x_1),(h_1,y_1)\big)$}
 \ar[r]^{\big((\alpha,\xi),(\beta,\eta)\big)}
 \ar@/_{2.5pc}/[rr]!U_(.4){\txt\scriptsize{$\qquad\qquad\qquad\big((\alpha',\xi'),(\beta',\eta')\big)\circ\big((\alpha,\xi),(\beta,\eta)\big)=\big((\alpha'\alpha,\xi'\xi),(\beta'\beta,\eta'\eta)\big)$}}
 & \txt{$\big((g_2,x_2),(h_2,y_2)\big)$}
 \ar[r]^{\big((\alpha',\xi'),(\beta',\eta')\big)}
 & \txt{$\big((g_3,x_3),(h_3,y_3)\big)$}
 }
 \end{array}
 \end{equation*}
on the one hand, we have
\begin{eqnarray*}
(\alpha'\alpha,\xi'\xi)\lhd(\beta'\beta,\eta'\eta)
&=&\big(\beta'\beta\otimes\alpha'\alpha\otimes(\iii_{h_1^{\dagger}}\otimes\beta^{-1}\beta'^{-1}\otimes\iii_{h_3^{\dagger}}), F(\beta'\beta,\xi'\xi)\big),
\end{eqnarray*}
on the other hand, since the functor $F$ preserves the composition, we have
\begin{eqnarray*}
&&{\big((\alpha',\xi')\lhd(\beta',\eta')\big)}\circ{\big((\alpha,\xi)\lhd(\beta,\eta)\big)}\\
&=&{\big(\beta'\otimes\alpha'\otimes(\iii_{h_2^{\dagger}}\otimes\beta'^{-1}\otimes\iii_{h_3^{\dagger}}), F(\beta',\xi')\big)}\circ
{\big(\beta\otimes\alpha\otimes(\iii_{h_1^{\dagger}}\otimes\beta^{-1}\otimes\iii_{h_2^{\dagger}}), F(\beta,\xi)\big)}\\
&=&{\big(\beta'\beta\otimes\alpha'\alpha\otimes
\big((\iii_{h_2^{\dagger}}\otimes\beta'^{-1}\otimes\iii_{h_3^{\dagger}})\circ(\iii_{h_1^{\dagger}}\otimes\beta^{-1}\otimes\iii_{h_2^{\dagger}})\big), F(\beta',\xi')\circ F(\beta,\xi)\big)}\\
&=&{\big(\beta'\beta\otimes\alpha'\alpha\otimes
\big((\iii_{h_2^{\dagger}}\otimes\beta'^{-1}\otimes\iii_{h_3^{\dagger}})\circ(\iii_{h_1^{\dagger}}\otimes\beta^{-1}\otimes\iii_{h_2^{\dagger}})\big), F(\beta'\beta,\xi'\xi)\big)}.
\end{eqnarray*}
Since $G$ is a strict monoidal category, for any objects $g,h\in G_0$ and morphism $\alpha:g\to h$, we have
\begin{equation*}
  \xymatrix@R=1.5pc@C=-0.15pc{
  \txt{$g\otimes\III$}
  \ar[d]_-{\alpha\otimes\iii_\III}
  &\txt{$=$}
  &\txt{$g$}
  \ar[d]_-{\alpha}
  &\txt{$=$}
  &\txt{$\III\otimes g$}
  \ar[d]^-{\iii_\III\otimes\alpha}
  &&&&&&&&&&
  \txt{$g\otimes g^{\dagger}$}
  \ar[d]_-{\iii_g\otimes\iii_{g^{\dagger}}}
  &\txt{$=$}
  &\txt{$\III$}
  \ar[d]_-{\iii_\III}
  &\txt{$=$}
  &\txt{$g^{\dagger}\otimes g$}
  \ar[d]^-{\iii_{g^{\dagger}}\otimes\iii_g}\\
  \txt{$h\otimes\III$}
  &\txt{$=$}
  &\txt{$h$}
  &\txt{$=$}
  &\txt{$\III\otimes h$}
  &&&&&&&&&&
  \txt{$g\otimes g^{\dagger}$}
  &\txt{$=$}
  &\txt{$\III$}
  &\txt{$=$}
  &\txt{$g^{\dagger}\otimes g$}
  }
\end{equation*}
which implies that $\alpha\otimes\iii_\III=\alpha=\iii_\III\otimes\alpha$ and $\iii_g\otimes\iii_{g^{\dagger}}=\iii_\III=\iii_{g^{\dagger}}\otimes\iii_g$.
Then we obtain that
\begin{eqnarray*}
\iii_{h_1^{\dagger}}\otimes \beta^{-1}\beta'^{-1}\otimes\iii_{h_3^{\dagger}}
&=&(\iii_{h_1^{\dagger}}\otimes \beta^{-1}\otimes\iii_{h_3^{\dagger}})\circ(\iii_{h_1^{\dagger}}\otimes \beta'^{-1}\otimes\iii_{h_3^{\dagger}})\\
&=&(\iii_{h_1^{\dagger}}\otimes \beta^{-1}\otimes\iii_\III\otimes\iii_{h_3^{\dagger}})\circ(\iii_{h_1^{\dagger}}\otimes\iii_\III\otimes \beta'^{-1}\otimes\iii_{h_3^{\dagger}})\\
&=&(\iii_{h_1^{\dagger}}\otimes \beta^{-1}\otimes\iii_{h_2^{\dagger}}\otimes\iii_{h_2}\otimes\iii_{h_3^{\dagger}})
\circ(\iii_{h_1^{\dagger}}\otimes\iii_{h_2}\otimes\iii_{h_2^{\dagger}}\otimes\beta'^{-1}\otimes\iii_{h_3^{\dagger}})\\
&=&\iii_{h_1^{\dagger}}\otimes \beta^{-1}\otimes\iii_{h_2^{\dagger}}\otimes\beta'^{-1}\otimes\iii_{h_3^{\dagger}},\\
(\iii_{h_2^{\dagger}}\otimes\beta'^{-1}\otimes\iii_{h_3^{\dagger}})\circ(\iii_{h_1^{\dagger}}\otimes\beta^{-1}\otimes\iii_{h_2^{\dagger}})
&=&(\iii_\III\otimes\iii_{h_2^{\dagger}}\otimes\beta'^{-1}\otimes\iii_{h_3^{\dagger}})\circ(\iii_{h_1^{\dagger}}\otimes\beta^{-1}\otimes\iii_{h_2^{\dagger}}\otimes\iii_\III)\\
&=&(\iii_{h_1^{\dagger}}\otimes\iii_{h_1}\otimes\iii_{h_2^{\dagger}}\otimes\beta'^{-1}\otimes\iii_{h_3^{\dagger}})\circ(\iii_{h_1^{\dagger}}\otimes\beta^{-1}\otimes\iii_{h_2^{\dagger}}\otimes\iii_{h_3}\otimes\iii_{h_3^{\dagger}})\\
&=&\iii_{h_1^{\dagger}}\otimes \beta^{-1}\otimes\iii_{h_2^{\dagger}}\otimes\beta'^{-1}\otimes\iii_{h_3^{\dagger}},
\end{eqnarray*}
which implies that $\iii_{h_1^{\dagger}}\otimes \beta^{-1}\beta'^{-1}\otimes\iii_{h_3^{\dagger}}=(\iii_{h_2^{\dagger}}\otimes\beta'^{-1}\otimes\iii_{h_3^{\dagger}})\circ(\iii_{h_1^{\dagger}}\otimes\beta^{-1}\otimes\iii_{h_2^{\dagger}})$.
Thus we have $$(\alpha'\alpha,\xi'\xi)\lhd(\beta'\beta,\eta'\eta)={\big((\alpha',\xi')\lhd(\beta',\eta')\big)}\circ{\big((\alpha,\xi)\lhd(\beta,\eta)\big)},$$
which implies that $\lhd$ preserves the composition. For any $(g,x),(h,y)\in G_0\times X_0$, since the functor $F$ preserves the identity map, we have
\begin{eqnarray*}
(\iii_g,\iii_x)\lhd(\iii_h,\iii_y)
&=&\big(\iii_h\otimes\iii_g\otimes(\iii_{h^{\dagger}}\otimes\iii_h^{-1}\otimes\iii_{h^{\dagger}}), F(\iii_h,\iii_x)\big)\\
&=&\big(\iii_h\otimes\iii_g\otimes\iii_{h^{\dagger}}\otimes\iii_h\otimes\iii_{h^{\dagger}}, \iii_{F(h,x)}\big)\\
&=&\big(\iii_h\otimes\iii_g\otimes\iii_{h^{\dagger}}, \iii_{F(h,x)}\big)\\
&=&\big(\iii_{h\otimes g\otimes{h^{\dagger}}}, \iii_{F(h,x)}\big)\\
&=&\iii_{({h\otimes g\otimes{h^{\dagger}}}, {F(h,x)})}\\
&=&\iii_{(g,x)\lhd(h,y)},
\end{eqnarray*}
which implies that $\lhd$ preserves the identity map. Therefore, $\lhd$ is a functor.

Second, it is obvious that, for any $(h,y)\in G_0\times X_0$, $\bullet\lhd(h,y):G\times X\to G\times X$ defined as follows is functor:
\begin{eqnarray*}
(\bullet\lhd(h,y))(g,x)&=&(g,x)\lhd(h,y)=\big(h\otimes g\otimes h^{\dagger}, F(h,x)\big),\qquad\,\,\,\,\,\forall (g,x),(h,y)\in G_0\times X_0,\\
(\bullet\lhd(h,y))(\alpha,\xi)&=&(\alpha,\xi)\lhd(\iii_h,\iii_x)=\big(\iii_h\otimes\alpha\otimes\iii_{h^{\dagger}}, F(\iii_h,\xi)\big),\quad\forall(\alpha,\xi)\in G_1\times X_1.
\end{eqnarray*}
\emptycomment{
\begin{equation*}
  \xymatrix@R=0.1pc@C=2.5pc{
  \txt{$\bullet\lhd(h,y):G\times X$} \ar[rr] && \txt{$G\times X\qquad\qquad\qquad\qquad\qquad\qquad\qquad\quad\,\,$}  \\
  \txt{$\qquad\qquad\,\,(g_1,x_1)$} \ar@{|->}[rr]
  \ar@<7ex>[ddddd]^{(\alpha,\xi)} && \txt{$(g_1,x_1)\lhd(h,y)=\big(h\otimes g_1\otimes h^{-1}, F(h,x_1)\big)$}
  \ar@<-16ex>[ddddd]^{(\alpha,\xi)\lhd(\iii_h,\iii_x)=\big(\iii_h\otimes\alpha\otimes(\iii_{h^{-1}}\otimes\iii_h^{-1}\otimes\iii_{h^{-1}}), F(\iii_h,\xi)\big)
  =\big(\iii_h\otimes\alpha\otimes\iii_{h^{-1}}, F(\iii_h,\xi)\big)}\\
  &&\\&&\\&&\\&&\\
  \txt{$\qquad\qquad\,\,(g_2,x_2)$} \ar@{|->}[rr] && \txt{$(g_2,x_2)\lhd(h,y)=\big(h\otimes g_2\otimes h^{-1}, F(h,x_2)\big)$}
  }
\end{equation*}
}
Moreover, the functor $\bullet\lhd(h,y)$ is an invertible functor with $\bullet\lhd(h^{\dagger},x)$ as its inverse.

Finally, for any $(g,x),(h,y),(k,z)\in G_0\times X_0$, we have
\begin{eqnarray*}
\big((g,x)\lhd(h,y)\big)\lhd(k,z)
&=&(h\otimes g\otimes h^{\dagger},F(h,x))\lhd(k,z)\\
&=&\big(k\otimes h\otimes g\otimes h^{\dagger}\otimes k^{\dagger},F(k,F(h,x))\big)\\
&=&\big(k\otimes h\otimes g\otimes h^{\dagger}\otimes k^{\dagger},F(k\otimes h,x))\big),\\
\big((g,x)\lhd(k,z)\big)\lhd\big((h,y)\lhd(k,z)\big)
&=&(k\otimes g\otimes k^{\dagger},F(k,x))\lhd(k\otimes h\otimes k^{\dagger},F(k,y))\\
&=&\big(k\otimes h\otimes k^{\dagger}\otimes k\otimes g\otimes k^{\dagger}\otimes k\otimes h^{\dagger}\otimes k^{\dagger},F(k\otimes h\otimes k^{\dagger},F(k,x))\big)\\
&=&\big(k\otimes h\otimes g\otimes h^{\dagger}\otimes k^{\dagger},F(k\otimes h\otimes k^{\dagger}\otimes k,x))\big)\\
&=&\big(k\otimes h\otimes g\otimes h^{\dagger}\otimes k^{\dagger},F(k\otimes h,x))\big),
\end{eqnarray*}
which implies that $\big((g,x)\lhd(h,y)\big)\lhd(k,z)=\big((g,x)\lhd(k,z)\big)\lhd\big((h,y)\lhd(k,z)\big)$.
Therefore, $(G\times X,\lhd)$ is a strict $2$-rack.
\end{proof}

\emptycomment{
\section{$2$-racks and Zamolodchikov Tetrahedron maps}
\begin{pro}
Let $X$ be a small category, $\lhd:X\times X\to X$ a functor and $\frkR_{x,y,z}:(x\lhd y)\lhd z\to (x\lhd z)\lhd(y\lhd z)$ a natural isomorphism. Define a functor $B:X\times X\to X\times X$ as $B(x,y)=(y,x\lhd y)$, where $x$ and $y$ are both either objects or morphisms in $X$. Define a natural transformation $Y:(B\times\Id)(\Id\times B)(B\times\Id)\Rightarrow(\Id\times B)(B\times\Id)(\Id\times B)$ as
$Y_{x,y,z}=(\Id_z,\Id_{y\lhd z},\frkR_{x,y,z})$.
Then $(B,Y)$ is a solution of the Zamolodchikov tetrahedron equation if and only if $(X,\lhd,\frkR)$ is a $2$-rack.
\end{pro}
\begin{proof}
Let $\widetilde{\lhd}:X\times X\to X$ be a functor and define a functor $\widetilde{B}:X\times X\to X\times X$ by $\widetilde{B}(x,y)=(y~\widetilde{\lhd}~x,x)$,
where $x$ and $y$ are both either objects or morphisms in $X$.
Then $B$ and $\widetilde{B}$ are equivalent functors if and only if $\bullet\lhd x$ and $\bullet~\widetilde{\lhd}~x$ are equivalent functors for any $x\in X_0$.
By direct calculations, $(B,Y)$ satisfies the Zamolodchikov tetrahedron equation if and only if $\frkR$ satisfies the distributor identity.
Therefore, $(B,Y)$ is a solution of the Zamolodchikov tetrahedron equation if and only if $(X,\lhd,\frkR)$ is a $2$-rack.
\end{proof}
}

\section*{Appendix}
The left and right side of the Zamolodchikov Tetrahedron equation can be showed as the following diagrams:
\begin{itemize}
\item the left side of the Zamolodchikov Tetrahedron equation:
  {\footnotesize
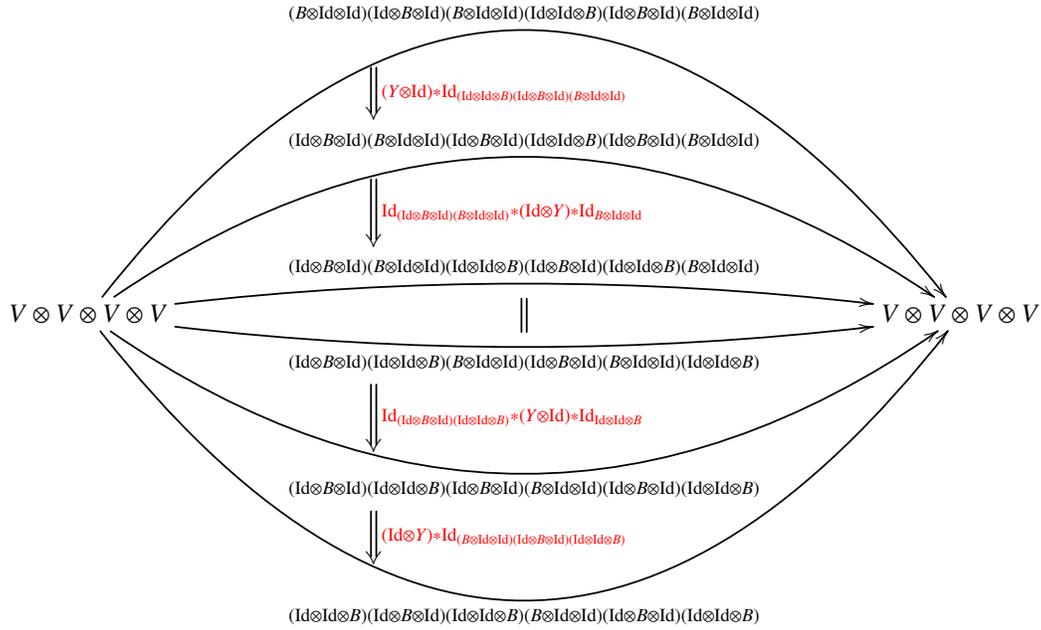
\[
  \xymatrix{
  V\otimes V\otimes V\otimes V
    \ar@/^9pc/[rrrrrrrrr]^{(B\otimes\Id\otimes\Id)(\Id\otimes B\otimes\Id)
    (B\otimes\Id\otimes\Id)(\Id\otimes\Id\otimes B)
    (\Id\otimes B\otimes\Id)(B\otimes\Id\otimes\Id)}_{}="0"
    \ar@/^5pc/[rrrrrrrrr]^{(\Id\otimes B\otimes\Id)(B\otimes\Id\otimes\Id)
    (\Id\otimes B\otimes\Id)(\Id\otimes\Id\otimes B)
     (\Id\otimes B\otimes\Id)(B\otimes\Id\otimes\Id)}_{}="1"
    \ar@/^1pc/[rrrrrrrrr]^{(\Id\otimes B\otimes\Id)(B\otimes\Id\otimes\Id)
    (\Id\otimes\Id\otimes B)(\Id\otimes B\otimes\Id)
    (\Id\otimes\Id\otimes B)(B\otimes\Id\otimes\Id)}_{}="2"
    \ar@/_1pc/[rrrrrrrrr]_{(\Id\otimes B\otimes\Id)(\Id\otimes\Id\otimes B)
    (B\otimes\Id\otimes\Id)(\Id\otimes B\otimes\Id)
    (B\otimes\Id\otimes\Id)(\Id\otimes\Id\otimes B)}_{}="3"
    \ar@/_5pc/[rrrrrrrrr]_{(\Id\otimes B\otimes\Id)(\Id\otimes\Id\otimes B)
    (\Id\otimes B\otimes\Id)(B\otimes\Id\otimes\Id)
    (\Id\otimes B\otimes\Id)(\Id\otimes\Id\otimes B)}_{}="4"
    \ar@/_9pc/[rrrrrrrrr]_{(\Id\otimes\Id\otimes B)(\Id\otimes B\otimes\Id)
    (\Id\otimes\Id\otimes B)(B\otimes\Id\otimes\Id)
    (\Id\otimes B\otimes\Id)(\Id\otimes\Id\otimes B)}_{}="5"
    \ar@{=>}"0"+(-20,-5);"1"+(-20,5) ^{\color{red}{(Y\otimes\Id)\ast\Id_{(\Id\otimes\Id\otimes B)
     (\Id\otimes B\otimes\Id)(B\otimes\Id\otimes\Id)}}}
    \ar@{=>}"1"+(-20,-3);"2"+(-20,5) ^{\color{red}{\Id_{(\Id\otimes B\otimes\Id)(B\otimes\Id\otimes\Id)}
    \ast(\Id\otimes Y)\ast\Id_{B\otimes\Id\otimes\Id}}}
    \ar@{=}"2"+(0,-2);"3"+(0,2)
    \ar@{=>}"3"+(-20,-5);"4"+(-20,3) ^{\color{red}{\Id_{(\Id\otimes B\otimes\Id)(\Id\otimes\Id\otimes B)}
    \ast(Y\otimes\Id)\ast\Id_{\Id\otimes\Id\otimes B}}}
    \ar@{=>}"4"+(-20,-5);"5"+(-20,5) ^{\color{red}{(\Id\otimes Y)\ast\Id_{(B\otimes\Id\otimes\Id)
    (\Id\otimes B\otimes\Id)(\Id\otimes\Id\otimes B)}}}
   &&&&&&&&&  V\otimes V\otimes V\otimes V
  }
  \]}
  \item the right side of the Zamolodchikov Tetrahedron equation:
  {\footnotesize
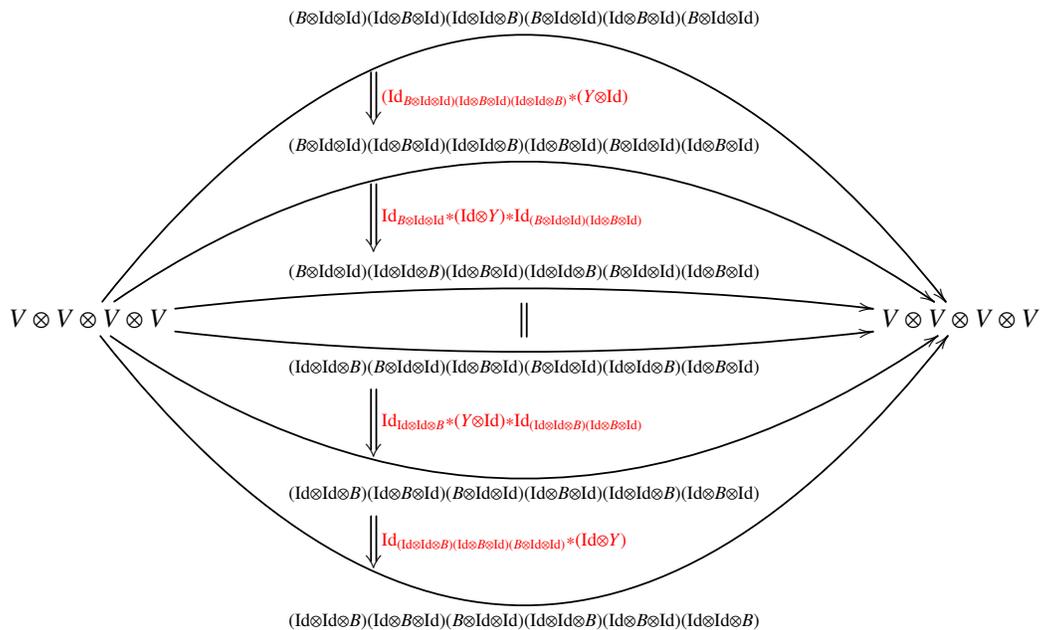
\[
  \xymatrix{
  V\otimes V\otimes V\otimes V
    \ar@/^9pc/[rrrrrrrrr]^{(B\otimes\Id\otimes\Id)(\Id\otimes B\otimes\Id)
    (\Id\otimes\Id\otimes B)(B\otimes\Id\otimes\Id)
    (\Id\otimes B\otimes\Id)(B\otimes\Id\otimes\Id)}_{}="0"
    \ar@/^5pc/[rrrrrrrrr]^{(B\otimes\Id\otimes\Id)(\Id\otimes B\otimes\Id)
    (\Id\otimes\Id\otimes B)(\Id\otimes B\otimes\Id)
     (B\otimes\Id\otimes\Id)(\Id\otimes B\otimes\Id)}_{}="1"
    \ar@/^1pc/[rrrrrrrrr]^{(B\otimes\Id\otimes\Id)(\Id\otimes\Id\otimes B)
    (\Id\otimes B\otimes\Id)(\Id\otimes\Id\otimes B)
    (B\otimes\Id\otimes\Id)(\Id\otimes B\otimes\Id)}_{}="2"
    \ar@/_1pc/[rrrrrrrrr]_{(\Id\otimes\Id\otimes B)(B\otimes\Id\otimes\Id)
    (\Id\otimes B\otimes\Id)(B\otimes\Id\otimes\Id)
    (\Id\otimes\Id\otimes B)(\Id\otimes B\otimes\Id)}_{}="3"
    \ar@/_5pc/[rrrrrrrrr]_{(\Id\otimes\Id\otimes B)(\Id\otimes B\otimes\Id)
    (B\otimes\Id\otimes\Id)(\Id\otimes B\otimes\Id)
    (\Id\otimes\Id\otimes B)(\Id\otimes B\otimes\Id)}_{}="4"
    \ar@/_9pc/[rrrrrrrrr]_{(\Id\otimes\Id\otimes B)(\Id\otimes B\otimes\Id)
    (B\otimes\Id\otimes\Id)(\Id\otimes\Id\otimes B)
    (\Id\otimes B\otimes\Id)(\Id\otimes\Id\otimes B)}_{}="5"
    \ar@{=>}"0"+(-20,-5);"1"+(-20,5) ^{\color{red}{(\Id_{B\otimes\Id\otimes\Id)(\Id\otimes B\otimes\Id)(\Id\otimes\Id\otimes B)}\ast(Y\otimes\Id)}}
    \ar@{=>}"1"+(-20,-3);"2"+(-20,5) ^{\color{red}{\Id_{B\otimes\Id\otimes\Id}\ast(\Id\otimes Y)\ast\Id_{(B\otimes\Id\otimes\Id)(\Id\otimes B\otimes\Id)}}}
    \ar@{=}"2"+(0,-2);"3"+(0,2)
    \ar@{=>}"3"+(-20,-5);"4"+(-20,3) ^{\color{red}{\Id_{\Id\otimes\Id\otimes B}
    \ast(Y\otimes\Id)\ast\Id_{(\Id\otimes\Id\otimes B)(\Id\otimes B\otimes\Id)}}}
    \ar@{=>}"4"+(-20,-5);"5"+(-20,5) ^{\color{red}{\Id_{(\Id\otimes\Id\otimes B)(\Id\otimes B\otimes\Id)(B\otimes\Id\otimes\Id)}\ast(\Id\otimes Y)}}
  &&&&&&&&&  V\otimes V\otimes V\otimes V
  }
  \]}
\end{itemize}

In fact, each of the above functors can be associated with a braiding diagram. For example, the braid diagram corresponding to
$$(B\otimes\Id\otimes\Id)(\Id\otimes B\otimes\Id)(B\otimes\Id\otimes\Id)(\Id\otimes\Id\otimes B)(\Id\otimes B\otimes\Id)(B\otimes\Id\otimes\Id)$$
is the following graph from the top to the bottom:
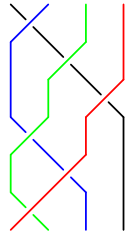
\[
\xy
   (0,0)*{}="00";
   (5,0)*{}="10";
   (10,0)*{}="20";
   (15,0)*{}="30";
   (0,5)*{}="01";
   (5,5)*{}="11";
   (10,5)*{}="21";
   (15,5)*{}="31";
   (0,10)*{}="02";
   (5,10)*{}="12";
   (10,10)*{}="22";
   (15,10)*{}="32";
   (0,15)*{}="03";
   (5,15)*{}="13";
   (10,15)*{}="23";
   (15,15)*{}="33";
   (0,-5)*{}="0-1";
   (5,-5)*{}="1-1";
   (10,-5)*{}="2-1";
   (15,-5)*{}="3-1";
   (0,-10)*{}="0-2";
   (5,-10)*{}="1-2";
   (10,-10)*{}="2-2";
   (15,-10)*{}="3-2";
   (0,-15)*{}="0-3";
   (5,-15)*{}="1-3";
   (10,-15)*{}="2-3";
   (15,-15)*{}="3-3";
   (2,13)*{}="213";
   (3,12)*{}="312";
   (7,8)*{}="78";
   (8,7)*{}="87";
   (12,3)*{}="123";
   (13,2)*{}="132";
   (2,-2)*{}="c2-2";
   (3,-3)*{}="c3-3";
   (7,-7)*{}="7-7";
   (8,-8)*{}="8-8";
   (2,-12)*{}="2-12";
   (3,-13)*{}="3-13";
   "0-3";"2-1" **[red]@{-};
   "2-1";"20" **[red]@{-};
   "20";"31" **[red]@{-};
   "31";"33" **[red]@{-};
   "1-3";"3-13" **[green]@{-};
   "2-12";"0-2" **[green]@{-};
   "0-2";"0-1" **[green]@{-};
   "0-1";"10" **[green]@{-};
   "10";"11" **[green]@{-};
   "11";"22" **[green]@{-};
   "22";"23" **[green]@{-};
   "2-3";"2-2" **[blue]@{-};
   "2-2";"8-8" **[blue]@{-};
   "7-7";"c3-3" **[blue]@{-};
   "c2-2";"00" **[blue]@{-};
   "00";"02" **[blue]@{-};
   "02";"13" **[blue]@{-};
   "3-3";"30" **@{-};
   "30";"132" **@{-};
   "123";"87" **@{-};
   "78";"312" **@{-};
   "213";"03" **@{-};
\endxy
\]
where each crossing represents the action of $B$ on two strands.
Then a natural transformation represents a transformation process between two braiding diagrams.
For example, the natural transformation
$$(Y\otimes\Id)\ast\Id_{(\Id\otimes\Id\otimes B)(\Id\otimes B\otimes\Id)(B\otimes\Id\otimes\Id)}$$
corresponds to the following transformation process:
{\footnotesize
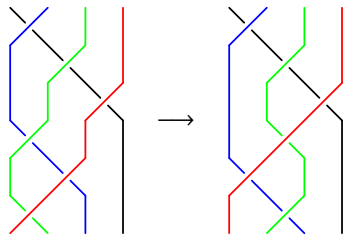
\[
 \xy
   (0,0)*{}="00";
   (5,0)*{}="10";
   (10,0)*{}="20";
   (15,0)*{}="30";
   (0,5)*{}="01";
   (5,5)*{}="11";
   (10,5)*{}="21";
   (15,5)*{}="31";
   (0,10)*{}="02";
   (5,10)*{}="12";
   (10,10)*{}="22";
   (15,10)*{}="32";
   (0,15)*{}="03";
   (5,15)*{}="13";
   (10,15)*{}="23";
   (15,15)*{}="33";
   (0,-5)*{}="0-1";
   (5,-5)*{}="1-1";
   (10,-5)*{}="2-1";
   (15,-5)*{}="3-1";
   (0,-10)*{}="0-2";
   (5,-10)*{}="1-2";
   (10,-10)*{}="2-2";
   (15,-10)*{}="3-2";
   (0,-15)*{}="0-3";
   (5,-15)*{}="1-3";
   (10,-15)*{}="2-3";
   (15,-15)*{}="3-3";
   (2,13)*{}="213";
   (3,12)*{}="312";
   (7,8)*{}="78";
   (8,7)*{}="87";
   (12,3)*{}="123";
   (13,2)*{}="132";
   (2,-2)*{}="c2-2";
   (3,-3)*{}="c3-3";
   (7,-7)*{}="7-7";
   (8,-8)*{}="8-8";
   (2,-12)*{}="2-12";
   (3,-13)*{}="3-13";
   "0-3";"2-1" **[red]@{-};
   "2-1";"20" **[red]@{-};
   "20";"31" **[red]@{-};
   "31";"33" **[red]@{-};
   "1-3";"3-13" **[green]@{-};
   "2-12";"0-2" **[green]@{-};
   "0-2";"0-1" **[green]@{-};
   "0-1";"10" **[green]@{-};
   "10";"11" **[green]@{-};
   "11";"22" **[green]@{-};
   "22";"23" **[green]@{-};
   "2-3";"2-2" **[blue]@{-};
   "2-2";"8-8" **[blue]@{-};
   "7-7";"c3-3" **[blue]@{-};
   "c2-2";"00" **[blue]@{-};
   "00";"02" **[blue]@{-};
   "02";"13" **[blue]@{-};
   "3-3";"30" **@{-};
   "30";"132" **@{-};
   "123";"87" **@{-};
   "78";"312" **@{-};
   "213";"03" **@{-};
\endxy
\quad \longrightarrow \quad
  \xy
   (0,0)*{}="00";
   (5,0)*{}="10";
   (10,0)*{}="20";
   (15,0)*{}="30";
   (0,5)*{}="01";
   (5,5)*{}="11";
   (10,5)*{}="21";
   (15,5)*{}="31";
   (0,10)*{}="02";
   (5,10)*{}="12";
   (10,10)*{}="22";
   (15,10)*{}="32";
   (0,15)*{}="03";
   (5,15)*{}="13";
   (10,15)*{}="23";
   (15,15)*{}="33";
   (0,-5)*{}="0-1";
   (5,-5)*{}="1-1";
   (10,-5)*{}="2-1";
   (15,-5)*{}="3-1";
   (0,-10)*{}="0-2";
   (5,-10)*{}="1-2";
   (10,-10)*{}="2-2";
   (15,-10)*{}="3-2";
   (0,-15)*{}="0-3";
   (5,-15)*{}="1-3";
   (10,-15)*{}="2-3";
   (15,-15)*{}="3-3";
   (2,13)*{}="213";
   (3,12)*{}="312";
   (7,8)*{}="78";
   (8,7)*{}="87";
   (12,3)*{}="123";
   (13,2)*{}="132";
   (7,-2)*{}="7-2";
   (8,-3)*{}="8-3";
   (2,-7)*{}="2-7";
   (3,-8)*{}="3-8";
   (7,-12)*{}="7-12";
   (8,-13)*{}="8-13";
   "0-3";"0-2" **[red]@{-};
   "0-2";"31" **[red]@{-};
   "31";"33" **[red]@{-};
   "1-3";"2-2" **[green]@{-};
   "2-2";"2-1" **[green]@{-};
   "2-1";"8-3" **[green]@{-};
   "7-2";"10" **[green]@{-};
   "10";"11" **[green]@{-};
   "11";"22" **[green]@{-};
   "22";"23" **[green]@{-};
   "2-3";"8-13" **[blue]@{-};
   "7-12";"3-8" **[blue]@{-};
   "2-7";"0-1" **[blue]@{-};
   "0-1";"02" **[blue]@{-};
   "02";"13" **[blue]@{-};
   "3-3";"30" **@{-};
   "30";"132" **@{-};
   "123";"87" **@{-};
   "78";"312" **@{-};
   "213";"03" **@{-};
\endxy
\]}
Therefore, the left and right sides of the Zamolodchikov Tetrahedron equation give rise to two distinct transformation processes from the following left braided diagram to the right one:
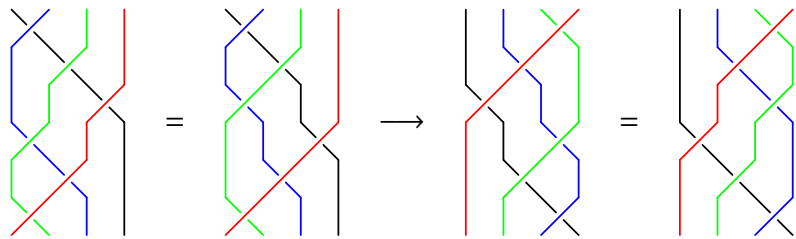
\[
\xy
   (0,0)*{}="00";
   (5,0)*{}="10";
   (10,0)*{}="20";
   (15,0)*{}="30";
   (0,5)*{}="01";
   (5,5)*{}="11";
   (10,5)*{}="21";
   (15,5)*{}="31";
   (0,10)*{}="02";
   (5,10)*{}="12";
   (10,10)*{}="22";
   (15,10)*{}="32";
   (0,15)*{}="03";
   (5,15)*{}="13";
   (10,15)*{}="23";
   (15,15)*{}="33";
   (0,-5)*{}="0-1";
   (5,-5)*{}="1-1";
   (10,-5)*{}="2-1";
   (15,-5)*{}="3-1";
   (0,-10)*{}="0-2";
   (5,-10)*{}="1-2";
   (10,-10)*{}="2-2";
   (15,-10)*{}="3-2";
   (0,-15)*{}="0-3";
   (5,-15)*{}="1-3";
   (10,-15)*{}="2-3";
   (15,-15)*{}="3-3";
   (2,13)*{}="213";
   (3,12)*{}="312";
   (7,8)*{}="78";
   (8,7)*{}="87";
   (12,3)*{}="123";
   (13,2)*{}="132";
   (2,-2)*{}="c2-2";
   (3,-3)*{}="c3-3";
   (7,-7)*{}="7-7";
   (8,-8)*{}="8-8";
   (2,-12)*{}="2-12";
   (3,-13)*{}="3-13";
   "0-3";"2-1" **[red]@{-};
   "2-1";"20" **[red]@{-};
   "20";"31" **[red]@{-};
   "31";"33" **[red]@{-};
   "1-3";"3-13" **[green]@{-};
   "2-12";"0-2" **[green]@{-};
   "0-2";"0-1" **[green]@{-};
   "0-1";"10" **[green]@{-};
   "10";"11" **[green]@{-};
   "11";"22" **[green]@{-};
   "22";"23" **[green]@{-};
   "2-3";"2-2" **[blue]@{-};
   "2-2";"8-8" **[blue]@{-};
   "7-7";"c3-3" **[blue]@{-};
   "c2-2";"00" **[blue]@{-};
   "00";"02" **[blue]@{-};
   "02";"13" **[blue]@{-};
   "3-3";"30" **@{-};
   "30";"132" **@{-};
   "123";"87" **@{-};
   "78";"312" **@{-};
   "213";"03" **@{-};
\endxy
\quad=\quad
\xy
   (0,0)*{}="00";
   (5,0)*{}="10";
   (10,0)*{}="20";
   (15,0)*{}="30";
   (0,5)*{}="01";
   (5,5)*{}="11";
   (10,5)*{}="21";
   (15,5)*{}="31";
   (0,10)*{}="02";
   (5,10)*{}="12";
   (10,10)*{}="22";
   (15,10)*{}="32";
   (0,15)*{}="03";
   (5,15)*{}="13";
   (10,15)*{}="23";
   (15,15)*{}="33";
   (0,-5)*{}="0-1";
   (5,-5)*{}="1-1";
   (10,-5)*{}="2-1";
   (15,-5)*{}="3-1";
   (0,-10)*{}="0-2";
   (5,-10)*{}="1-2";
   (10,-10)*{}="2-2";
   (15,-10)*{}="3-2";
   (0,-15)*{}="0-3";
   (5,-15)*{}="1-3";
   (10,-15)*{}="2-3";
   (15,-15)*{}="3-3";
   (2,13)*{}="213";
   (3,12)*{}="312";
   (7,8)*{}="78";
   (8,7)*{}="87";
   (2,3)*{}="c23";
   (3,2)*{}="c32";
   (12,-2)*{}="12-2";
   (13,-3)*{}="13-3";
   (7,-7)*{}="7-7";
   (8,-8)*{}="8-8";
   (2,-12)*{}="2-12";
   (3,-13)*{}="3-13";
   "0-3";"30" **[red]@{-};
   "30";"33" **[red]@{-};
   "1-3";"3-13" **[green]@{-};
   "2-12";"0-2" **[green]@{-};
   "0-2";"00" **[green]@{-};
   "00";"22" **[green]@{-};
   "22";"23" **[green]@{-};
   "2-3";"2-2" **[blue]@{-};
   "2-2";"8-8" **[blue]@{-};
   "7-7";"1-1" **[blue]@{-};
   "1-1";"10" **[blue]@{-};
   "10";"c32" **[blue]@{-};
   "c23";"01" **[blue]@{-};
   "01";"02" **[blue]@{-};
   "02";"13" **[blue]@{-};
   "3-3";"3-1" **@{-};
   "3-1";"13-3" **@{-};
   "12-2";"20" **@{-};
   "20";"21" **@{-};
   "21";"87" **@{-};
   "78";"312" **@{-};
   "213";"03" **@{-};
\endxy
\quad\longrightarrow\quad
\xy
   (0,0)*{}="00";
   (5,0)*{}="10";
   (10,0)*{}="20";
   (15,0)*{}="30";
   (0,5)*{}="01";
   (5,5)*{}="11";
   (10,5)*{}="21";
   (15,5)*{}="31";
   (0,10)*{}="02";
   (5,10)*{}="12";
   (10,10)*{}="22";
   (15,10)*{}="32";
   (0,15)*{}="03";
   (5,15)*{}="13";
   (10,15)*{}="23";
   (15,15)*{}="33";
   (0,-5)*{}="0-1";
   (5,-5)*{}="1-1";
   (10,-5)*{}="2-1";
   (15,-5)*{}="3-1";
   (0,-10)*{}="0-2";
   (5,-10)*{}="1-2";
   (10,-10)*{}="2-2";
   (15,-10)*{}="3-2";
   (0,-15)*{}="0-3";
   (5,-15)*{}="1-3";
   (10,-15)*{}="2-3";
   (15,-15)*{}="3-3";
   (12,13)*{}="1213";
   (13,12)*{}="1312";
   (7,8)*{}="78";
   (8,7)*{}="87";
   (2,3)*{}="c23";
   (3,2)*{}="c32";
   (12,-2)*{}="12-2";
   (13,-3)*{}="13-3";
   (7,-7)*{}="7-7";
   (8,-8)*{}="8-8";
   (12,-12)*{}="12-12";
   (13,-13)*{}="13-13";
   "0-3";"00" **[red]@{-};
   "00";"33" **[red]@{-};
   "1-3";"1-2" **[green]@{-};
   "1-2";"30" **[green]@{-};
   "30";"32" **[green]@{-};
   "32";"1312" **[green]@{-};
   "1213";"23" **[green]@{-};
   "2-3";"3-2" **[blue]@{-};
   "3-2";"3-1" **[blue]@{-};
   "3-1";"13-3" **[blue]@{-};
   "12-2";"20" **[blue]@{-};
   "20";"21" **[blue]@{-};
   "21";"87" **[blue]@{-};
   "78";"12" **[blue]@{-};
   "12";"13" **[blue]@{-};
   "3-3";"13-13" **@{-};
   "12-12";"8-8" **@{-};
   "7-7";"1-1" **@{-};
   "1-1";"10" **@{-};
   "10";"c32" **@{-};
   "c23";"01" **@{-};
   "01";"03" **@{-};
\endxy
\quad=\quad
\xy
   (0,0)*{}="00";
   (5,0)*{}="10";
   (10,0)*{}="20";
   (15,0)*{}="30";
   (0,5)*{}="01";
   (5,5)*{}="11";
   (10,5)*{}="21";
   (15,5)*{}="31";
   (0,10)*{}="02";
   (5,10)*{}="12";
   (10,10)*{}="22";
   (15,10)*{}="32";
   (0,15)*{}="03";
   (5,15)*{}="13";
   (10,15)*{}="23";
   (15,15)*{}="33";
   (0,-5)*{}="0-1";
   (5,-5)*{}="1-1";
   (10,-5)*{}="2-1";
   (15,-5)*{}="3-1";
   (0,-10)*{}="0-2";
   (5,-10)*{}="1-2";
   (10,-10)*{}="2-2";
   (15,-10)*{}="3-2";
   (0,-15)*{}="0-3";
   (5,-15)*{}="1-3";
   (10,-15)*{}="2-3";
   (15,-15)*{}="3-3";
   (12,13)*{}="1213";
   (13,12)*{}="1312";
   (7,8)*{}="78";
   (8,7)*{}="87";
   (12,3)*{}="123";
   (13,2)*{}="132";
   (2,-2)*{}="c2-2";
   (3,-3)*{}="c3-3";
   (7,-7)*{}="7-7";
   (8,-8)*{}="8-8";
   (12,-12)*{}="12-12";
   (13,-13)*{}="13-13";
   "0-3";"0-1" **[red]@{-};
   "0-1";"10" **[red]@{-};
   "10";"11" **[red]@{-};
   "11";"33" **[red]@{-};
   "1-3";"1-2" **[green]@{-};
   "1-2";"2-1" **[green]@{-};
   "2-1";"20" **[green]@{-};
   "20";"31" **[green]@{-};
   "31";"32" **[green]@{-};
   "32";"1312" **[green]@{-};
   "1213";"23" **[green]@{-};
   "2-3";"3-2" **[blue]@{-};
   "3-2";"30" **[blue]@{-};
   "30";"132" **[blue]@{-};
   "123";"87" **[blue]@{-};
   "78";"12" **[blue]@{-};
   "12";"13" **[blue]@{-};
   "3-3";"13-13" **@{-};
   "12-12";"8-8" **@{-};
   "7-7";"c3-3" **@{-};
   "c2-2";"00" **@{-};
   "00";"03" **@{-};
\endxy
\]
More precisely,
\begin{itemize}
\item the left side of the Zamolodchikov Tetrahedron Equation corresponds to the following process:
{\footnotesize
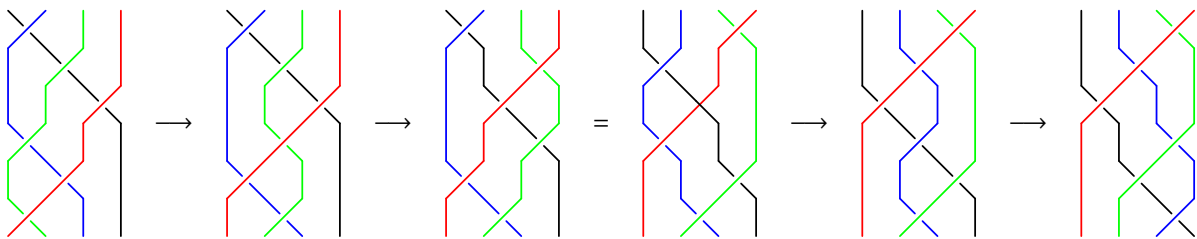
\[
 \xy
   (0,0)*{}="00";
   (5,0)*{}="10";
   (10,0)*{}="20";
   (15,0)*{}="30";
   (0,5)*{}="01";
   (5,5)*{}="11";
   (10,5)*{}="21";
   (15,5)*{}="31";
   (0,10)*{}="02";
   (5,10)*{}="12";
   (10,10)*{}="22";
   (15,10)*{}="32";
   (0,15)*{}="03";
   (5,15)*{}="13";
   (10,15)*{}="23";
   (15,15)*{}="33";
   (0,-5)*{}="0-1";
   (5,-5)*{}="1-1";
   (10,-5)*{}="2-1";
   (15,-5)*{}="3-1";
   (0,-10)*{}="0-2";
   (5,-10)*{}="1-2";
   (10,-10)*{}="2-2";
   (15,-10)*{}="3-2";
   (0,-15)*{}="0-3";
   (5,-15)*{}="1-3";
   (10,-15)*{}="2-3";
   (15,-15)*{}="3-3";
   (2,13)*{}="213";
   (3,12)*{}="312";
   (7,8)*{}="78";
   (8,7)*{}="87";
   (12,3)*{}="123";
   (13,2)*{}="132";
   (2,-2)*{}="c2-2";
   (3,-3)*{}="c3-3";
   (7,-7)*{}="7-7";
   (8,-8)*{}="8-8";
   (2,-12)*{}="2-12";
   (3,-13)*{}="3-13";
   "0-3";"2-1" **[red]@{-};
   "2-1";"20" **[red]@{-};
   "20";"31" **[red]@{-};
   "31";"33" **[red]@{-};
   "1-3";"3-13" **[green]@{-};
   "2-12";"0-2" **[green]@{-};
   "0-2";"0-1" **[green]@{-};
   "0-1";"10" **[green]@{-};
   "10";"11" **[green]@{-};
   "11";"22" **[green]@{-};
   "22";"23" **[green]@{-};
   "2-3";"2-2" **[blue]@{-};
   "2-2";"8-8" **[blue]@{-};
   "7-7";"c3-3" **[blue]@{-};
   "c2-2";"00" **[blue]@{-};
   "00";"02" **[blue]@{-};
   "02";"13" **[blue]@{-};
   "3-3";"30" **@{-};
   "30";"132" **@{-};
   "123";"87" **@{-};
   "78";"312" **@{-};
   "213";"03" **@{-};
\endxy
\quad \longrightarrow \quad
  \xy
   (0,0)*{}="00";
   (5,0)*{}="10";
   (10,0)*{}="20";
   (15,0)*{}="30";
   (0,5)*{}="01";
   (5,5)*{}="11";
   (10,5)*{}="21";
   (15,5)*{}="31";
   (0,10)*{}="02";
   (5,10)*{}="12";
   (10,10)*{}="22";
   (15,10)*{}="32";
   (0,15)*{}="03";
   (5,15)*{}="13";
   (10,15)*{}="23";
   (15,15)*{}="33";
   (0,-5)*{}="0-1";
   (5,-5)*{}="1-1";
   (10,-5)*{}="2-1";
   (15,-5)*{}="3-1";
   (0,-10)*{}="0-2";
   (5,-10)*{}="1-2";
   (10,-10)*{}="2-2";
   (15,-10)*{}="3-2";
   (0,-15)*{}="0-3";
   (5,-15)*{}="1-3";
   (10,-15)*{}="2-3";
   (15,-15)*{}="3-3";
   (2,13)*{}="213";
   (3,12)*{}="312";
   (7,8)*{}="78";
   (8,7)*{}="87";
   (12,3)*{}="123";
   (13,2)*{}="132";
   (7,-2)*{}="7-2";
   (8,-3)*{}="8-3";
   (2,-7)*{}="2-7";
   (3,-8)*{}="3-8";
   (7,-12)*{}="7-12";
   (8,-13)*{}="8-13";
   "0-3";"0-2" **[red]@{-};
   "0-2";"31" **[red]@{-};
   "31";"33" **[red]@{-};
   "1-3";"2-2" **[green]@{-};
   "2-2";"2-1" **[green]@{-};
   "2-1";"8-3" **[green]@{-};
   "7-2";"10" **[green]@{-};
   "10";"11" **[green]@{-};
   "11";"22" **[green]@{-};
   "22";"23" **[green]@{-};
   "2-3";"8-13" **[blue]@{-};
   "7-12";"3-8" **[blue]@{-};
   "2-7";"0-1" **[blue]@{-};
   "0-1";"02" **[blue]@{-};
   "02";"13" **[blue]@{-};
   "3-3";"30" **@{-};
   "30";"132" **@{-};
   "123";"87" **@{-};
   "78";"312" **@{-};
   "213";"03" **@{-};
\endxy
\quad \longrightarrow \quad
\xy
   (0,0)*{}="00";
   (5,0)*{}="10";
   (10,0)*{}="20";
   (15,0)*{}="30";
   (0,5)*{}="01";
   (5,5)*{}="11";
   (10,5)*{}="21";
   (15,5)*{}="31";
   (0,10)*{}="02";
   (5,10)*{}="12";
   (10,10)*{}="22";
   (15,10)*{}="32";
   (0,15)*{}="03";
   (5,15)*{}="13";
   (10,15)*{}="23";
   (15,15)*{}="33";
   (0,-5)*{}="0-1";
   (5,-5)*{}="1-1";
   (10,-5)*{}="2-1";
   (15,-5)*{}="3-1";
   (0,-10)*{}="0-2";
   (5,-10)*{}="1-2";
   (10,-10)*{}="2-2";
   (15,-10)*{}="3-2";
   (0,-15)*{}="0-3";
   (5,-15)*{}="1-3";
   (10,-15)*{}="2-3";
   (15,-15)*{}="3-3";
   (2,13)*{}="213";
   (3,12)*{}="312";
   (12,8)*{}="128";
   (13,7)*{}="137";
   (7,3)*{}="73";
   (8,2)*{}="82";
   (12,-2)*{}="12-2";
   (13,-3)*{}="13-3";
   (2,-7)*{}="2-7";
   (3,-8)*{}="3-8";
   (7,-12)*{}="7-12";
   (8,-13)*{}="8-13";
   "0-3";"0-2" **[red]@{-};
   "0-2";"1-1" **[red]@{-};
   "1-1";"10" **[red]@{-};
   "10";"32" **[red]@{-};
   "32";"33" **[red]@{-};
   "1-3";"2-2" **[green]@{-};
   "2-2";"2-1" **[green]@{-};
   "2-1";"30" **[green]@{-};
   "30";"31" **[green]@{-};
   "31";"137" **[green]@{-};
   "128";"22" **[green]@{-};
   "22";"23" **[green]@{-};
   "2-3";"8-13" **[blue]@{-};
   "7-12";"3-8" **[blue]@{-};
   "2-7";"0-1" **[blue]@{-};
   "0-1";"02" **[blue]@{-};
   "02";"13" **[blue]@{-};
   "3-3";"3-1" **@{-};
   "3-1";"13-3" **@{-};
   "12-2";"82" **@{-};
   "73";"11" **@{-};
   "11";"12" **@{-};
   "12";"312" **@{-};
   "213";"03" **@{-};
\endxy
\quad = \quad
\xy
   (0,0)*{}="00";
   (5,0)*{}="10";
   (10,0)*{}="20";
   (15,0)*{}="30";
   (0,5)*{}="01";
   (5,5)*{}="11";
   (10,5)*{}="21";
   (15,5)*{}="31";
   (0,10)*{}="02";
   (5,10)*{}="12";
   (10,10)*{}="22";
   (15,10)*{}="32";
   (0,15)*{}="03";
   (5,15)*{}="13";
   (10,15)*{}="23";
   (15,15)*{}="33";
   (0,-5)*{}="0-1";
   (5,-5)*{}="1-1";
   (10,-5)*{}="2-1";
   (15,-5)*{}="3-1";
   (0,-10)*{}="0-2";
   (5,-10)*{}="1-2";
   (10,-10)*{}="2-2";
   (15,-10)*{}="3-2";
   (0,-15)*{}="0-3";
   (5,-15)*{}="1-3";
   (10,-15)*{}="2-3";
   (15,-15)*{}="3-3";
   (12,13)*{}="1213";
   (13,12)*{}="1312";
   (2,8)*{}="28";
   (3,7)*{}="37";
   (7,3)*{}="73";
   (8,2)*{}="82";
   (2,-2)*{}="2-2";
   (3,-3)*{}="3--3";
   (12,-7)*{}="12-7";
   (13,-8)*{}="13-8";
   (7,-12)*{}="7-12";
   (8,-13)*{}="8-13";
   "0-3";"0-1" **[red]@{-};
   "0-1";"21" **[red]@{-};
   "21";"22" **[red]@{-};
   "22";"33" **[red]@{-};
   "1-3";"3-1" **[green]@{-};
   "3-1";"32" **[green]@{-};
   "32";"1312" **[green]@{-};
   "1213";"23" **[green]@{-};
   "2-3";"8-13" **[blue]@{-};
   "7-12";"1-2" **[blue]@{-};
   "1-2";"1-1" **[blue]@{-};
   "1-1";"3--3" **[blue]@{-};
   "2-2";"00" **[blue]@{-};
   "00";"01" **[blue]@{-};
   "01";"12" **[blue]@{-};
   "12";"13" **[blue]@{-};
   "3-3";"3-2" **@{-};
   "3-2";"13-8" **@{-};
   "12-7";"2-1" **@{-};
   "2-1";"20" **@{-};
   "20";"73" **@{-};
   "73";"37" **@{-};
   "28";"02" **@{-};
   "02";"03" **@{-};
\endxy
\quad \longrightarrow \quad
\xy
   (0,0)*{}="00";
   (5,0)*{}="10";
   (10,0)*{}="20";
   (15,0)*{}="30";
   (0,5)*{}="01";
   (5,5)*{}="11";
   (10,5)*{}="21";
   (15,5)*{}="31";
   (0,10)*{}="02";
   (5,10)*{}="12";
   (10,10)*{}="22";
   (15,10)*{}="32";
   (0,15)*{}="03";
   (5,15)*{}="13";
   (10,15)*{}="23";
   (15,15)*{}="33";
   (0,-5)*{}="0-1";
   (5,-5)*{}="1-1";
   (10,-5)*{}="2-1";
   (15,-5)*{}="3-1";
   (0,-10)*{}="0-2";
   (5,-10)*{}="1-2";
   (10,-10)*{}="2-2";
   (15,-10)*{}="3-2";
   (0,-15)*{}="0-3";
   (5,-15)*{}="1-3";
   (10,-15)*{}="2-3";
   (15,-15)*{}="3-3";
   (12,13)*{}="1213";
   (13,12)*{}="1312";
   (7,8)*{}="78";
   (8,7)*{}="87";
   (2,3)*{}="c23";
   (3,2)*{}="c32";
   (7,-2)*{}="7-2";
   (8,-3)*{}="8-3";
   (12,-7)*{}="12-7";
   (13,-8)*{}="13-8";
   (7,-12)*{}="7-12";
   (8,-13)*{}="8-13";
   "0-3";"00" **[red]@{-};
   "00";"33" **[red]@{-};
   "1-3";"3-1" **[green]@{-};
   "3-1";"32" **[green]@{-};
   "32";"1312" **[green]@{-};
   "1213";"23" **[green]@{-};
   "2-3";"8-13" **[blue]@{-};
   "7-12";"1-2" **[blue]@{-};
   "1-2";"1-1" **[blue]@{-};
   "1-1";"20" **[blue]@{-};
   "20";"21" **[blue]@{-};
   "21";"87" **[blue]@{-};
   "78";"12" **[blue]@{-};
   "12";"13" **[blue]@{-};
   "3-3";"3-2" **@{-};
   "3-2";"13-8" **@{-};
   "12-7";"8-3" **@{-};
   "7-2";"c32" **@{-};
   "c23";"01" **@{-};
   "01";"03" **@{-};
\endxy
\quad \longrightarrow \quad
\xy
   (0,0)*{}="00";
   (5,0)*{}="10";
   (10,0)*{}="20";
   (15,0)*{}="30";
   (0,5)*{}="01";
   (5,5)*{}="11";
   (10,5)*{}="21";
   (15,5)*{}="31";
   (0,10)*{}="02";
   (5,10)*{}="12";
   (10,10)*{}="22";
   (15,10)*{}="32";
   (0,15)*{}="03";
   (5,15)*{}="13";
   (10,15)*{}="23";
   (15,15)*{}="33";
   (0,-5)*{}="0-1";
   (5,-5)*{}="1-1";
   (10,-5)*{}="2-1";
   (15,-5)*{}="3-1";
   (0,-10)*{}="0-2";
   (5,-10)*{}="1-2";
   (10,-10)*{}="2-2";
   (15,-10)*{}="3-2";
   (0,-15)*{}="0-3";
   (5,-15)*{}="1-3";
   (10,-15)*{}="2-3";
   (15,-15)*{}="3-3";
   (12,13)*{}="1213";
   (13,12)*{}="1312";
   (7,8)*{}="78";
   (8,7)*{}="87";
   (2,3)*{}="c23";
   (3,2)*{}="c32";
   (12,-2)*{}="12-2";
   (13,-3)*{}="13-3";
   (7,-7)*{}="7-7";
   (8,-8)*{}="8-8";
   (12,-12)*{}="12-12";
   (13,-13)*{}="13-13";
   "0-3";"00" **[red]@{-};
   "00";"33" **[red]@{-};
   "1-3";"1-2" **[green]@{-};
   "1-2";"30" **[green]@{-};
   "30";"32" **[green]@{-};
   "32";"1312" **[green]@{-};
   "1213";"23" **[green]@{-};
   "2-3";"3-2" **[blue]@{-};
   "3-2";"3-1" **[blue]@{-};
   "3-1";"13-3" **[blue]@{-};
   "12-2";"20" **[blue]@{-};
   "20";"21" **[blue]@{-};
   "21";"87" **[blue]@{-};
   "78";"12" **[blue]@{-};
   "12";"13" **[blue]@{-};
   "3-3";"13-13" **@{-};
   "12-12";"8-8" **@{-};
   "7-7";"1-1" **@{-};
   "1-1";"10" **@{-};
   "10";"c32" **@{-};
   "c23";"01" **@{-};
   "01";"03" **@{-};
\endxy
\]}
  \item the right side of the Zamolodchikov Tetrahedron Equation corresponds to the following process:
{\footnotesize
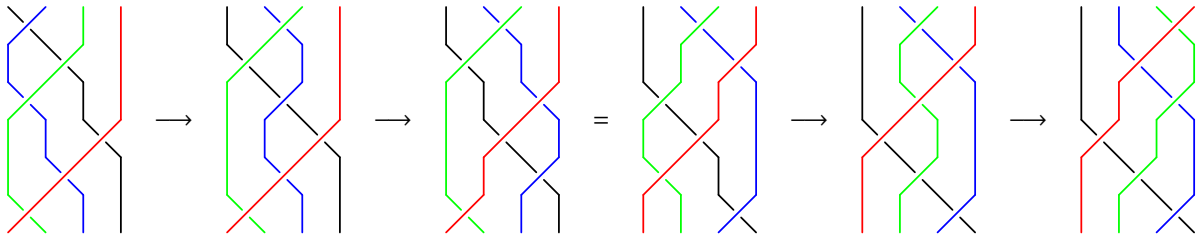
\[
 \xy
   (0,0)*{}="00";
   (5,0)*{}="10";
   (10,0)*{}="20";
   (15,0)*{}="30";
   (0,5)*{}="01";
   (5,5)*{}="11";
   (10,5)*{}="21";
   (15,5)*{}="31";
   (0,10)*{}="02";
   (5,10)*{}="12";
   (10,10)*{}="22";
   (15,10)*{}="32";
   (0,15)*{}="03";
   (5,15)*{}="13";
   (10,15)*{}="23";
   (15,15)*{}="33";
   (0,-5)*{}="0-1";
   (5,-5)*{}="1-1";
   (10,-5)*{}="2-1";
   (15,-5)*{}="3-1";
   (0,-10)*{}="0-2";
   (5,-10)*{}="1-2";
   (10,-10)*{}="2-2";
   (15,-10)*{}="3-2";
   (0,-15)*{}="0-3";
   (5,-15)*{}="1-3";
   (10,-15)*{}="2-3";
   (15,-15)*{}="3-3";
   (2,13)*{}="213";
   (3,12)*{}="312";
   (7,8)*{}="78";
   (8,7)*{}="87";
   (2,3)*{}="c23";
   (3,2)*{}="c32";
   (12,-2)*{}="12-2";
   (13,-3)*{}="13-3";
   (7,-7)*{}="7-7";
   (8,-8)*{}="8-8";
   (2,-12)*{}="2-12";
   (3,-13)*{}="3-13";
   "0-3";"30" **[red]@{-};
   "30";"33" **[red]@{-};
   "1-3";"3-13" **[green]@{-};
   "2-12";"0-2" **[green]@{-};
   "0-2";"00" **[green]@{-};
   "00";"22" **[green]@{-};
   "22";"23" **[green]@{-};
   "2-3";"2-2" **[blue]@{-};
   "2-2";"8-8" **[blue]@{-};
   "7-7";"1-1" **[blue]@{-};
   "1-1";"10" **[blue]@{-};
   "10";"c32" **[blue]@{-};
   "c23";"01" **[blue]@{-};
   "01";"02" **[blue]@{-};
   "02";"13" **[blue]@{-};
   "3-3";"3-1" **@{-};
   "3-1";"13-3" **@{-};
   "12-2";"20" **@{-};
   "20";"21" **@{-};
   "21";"87" **@{-};
   "78";"312" **@{-};
   "213";"03" **@{-};
\endxy
\quad\longrightarrow\quad
 \xy
   (0,0)*{}="00";
   (5,0)*{}="10";
   (10,0)*{}="20";
   (15,0)*{}="30";
   (0,5)*{}="01";
   (5,5)*{}="11";
   (10,5)*{}="21";
   (15,5)*{}="31";
   (0,10)*{}="02";
   (5,10)*{}="12";
   (10,10)*{}="22";
   (15,10)*{}="32";
   (0,15)*{}="03";
   (5,15)*{}="13";
   (10,15)*{}="23";
   (15,15)*{}="33";
   (0,-5)*{}="0-1";
   (5,-5)*{}="1-1";
   (10,-5)*{}="2-1";
   (15,-5)*{}="3-1";
   (0,-10)*{}="0-2";
   (5,-10)*{}="1-2";
   (10,-10)*{}="2-2";
   (15,-10)*{}="3-2";
   (0,-15)*{}="0-3";
   (5,-15)*{}="1-3";
   (10,-15)*{}="2-3";
   (15,-15)*{}="3-3";
   (7,13)*{}="713";
   (8,12)*{}="812";
   (2,8)*{}="28";
   (3,7)*{}="37";
   (7,3)*{}="73";
   (8,2)*{}="82";
   (12,-2)*{}="12-2";
   (13,-3)*{}="13-3";
   (7,-7)*{}="7-7";
   (8,-8)*{}="8-8";
   (2,-12)*{}="2-12";
   (3,-13)*{}="3-13";
   "0-3";"30" **[red]@{-};
   "30";"33" **[red]@{-};
   "1-3";"3-13" **[green]@{-};
   "2-12";"0-2" **[green]@{-};
   "0-2";"01" **[green]@{-};
   "01";"23" **[green]@{-};
   "2-3";"2-2" **[blue]@{-};
   "2-2";"8-8" **[blue]@{-};
   "7-7";"1-1" **[blue]@{-};
   "1-1";"10" **[blue]@{-};
   "10";"21" **[blue]@{-};
   "21";"22" **[blue]@{-};
   "22";"812" **[blue]@{-};
   "713";"13" **[blue]@{-};
   "3-3";"3-1" **@{-};
   "3-1";"13-3" **@{-};
   "12-2";"82" **@{-};
   "73";"37" **@{-};
   "28";"02" **@{-};
   "02";"03" **@{-};
\endxy
\quad\longrightarrow\quad
 \xy
   (0,0)*{}="00";
   (5,0)*{}="10";
   (10,0)*{}="20";
   (15,0)*{}="30";
   (0,5)*{}="01";
   (5,5)*{}="11";
   (10,5)*{}="21";
   (15,5)*{}="31";
   (0,10)*{}="02";
   (5,10)*{}="12";
   (10,10)*{}="22";
   (15,10)*{}="32";
   (0,15)*{}="03";
   (5,15)*{}="13";
   (10,15)*{}="23";
   (15,15)*{}="33";
   (0,-5)*{}="0-1";
   (5,-5)*{}="1-1";
   (10,-5)*{}="2-1";
   (15,-5)*{}="3-1";
   (0,-10)*{}="0-2";
   (5,-10)*{}="1-2";
   (10,-10)*{}="2-2";
   (15,-10)*{}="3-2";
   (0,-15)*{}="0-3";
   (5,-15)*{}="1-3";
   (10,-15)*{}="2-3";
   (15,-15)*{}="3-3";
   (7,13)*{}="713";
   (8,12)*{}="812";
   (2,8)*{}="28";
   (3,7)*{}="37";
   (12,3)*{}="123";
   (13,2)*{}="132";
   (7,-2)*{}="7-2";
   (8,-3)*{}="8-3";
   (13,-8)*{}="13-8";
   (12,-7)*{}="12-7";
   (2,-12)*{}="2-12";
   (3,-13)*{}="3-13";
   "0-3";"1-2" **[red]@{-};
   "1-2";"1-1" **[red]@{-};
   "1-1";"31" **[red]@{-};
   "31";"33" **[red]@{-};
   "1-3";"3-13" **[green]@{-};
   "2-12";"0-2" **[green]@{-};
   "0-2";"01" **[green]@{-};
   "01";"23" **[green]@{-};
   "2-3";"2-2" **[blue]@{-};
   "2-2";"3-1" **[blue]@{-};
   "3-1";"30" **[blue]@{-};
   "30";"132" **[blue]@{-};
   "123";"21" **[blue]@{-};
   "21";"22" **[blue]@{-};
   "22";"812" **[blue]@{-};
   "713";"13" **[blue]@{-};
   "3-3";"3-2" **@{-};
   "3-2";"13-8" **@{-};
   "12-7";"8-3" **@{-};
   "7-2";"10" **@{-};
   "10";"11" **@{-};
   "11";"37" **@{-};
   "28";"02" **@{-};
   "02";"03" **@{-};
\endxy
\quad = \quad
 \xy
   (0,0)*{}="00";
   (5,0)*{}="10";
   (10,0)*{}="20";
   (15,0)*{}="30";
   (0,5)*{}="01";
   (5,5)*{}="11";
   (10,5)*{}="21";
   (15,5)*{}="31";
   (0,10)*{}="02";
   (5,10)*{}="12";
   (10,10)*{}="22";
   (15,10)*{}="32";
   (0,15)*{}="03";
   (5,15)*{}="13";
   (10,15)*{}="23";
   (15,15)*{}="33";
   (0,-5)*{}="0-1";
   (5,-5)*{}="1-1";
   (10,-5)*{}="2-1";
   (15,-5)*{}="3-1";
   (0,-10)*{}="0-2";
   (5,-10)*{}="1-2";
   (10,-10)*{}="2-2";
   (15,-10)*{}="3-2";
   (0,-15)*{}="0-3";
   (5,-15)*{}="1-3";
   (10,-15)*{}="2-3";
   (15,-15)*{}="3-3";
   (7,13)*{}="713";
   (8,12)*{}="812";
   (12,8)*{}="128";
   (13,7)*{}="137";
   (2,3)*{}="c23";
   (3,2)*{}="c32";
   (7,-2)*{}="7-2";
   (8,-3)*{}="8-3";
   (2,-7)*{}="2-7";
   (3,-8)*{}="3-8";
   (12,-12)*{}="12-12";
   (13,-13)*{}="13-13";
   "0-3";"0-2" **[red]@{-};
   "0-2";"20" **[red]@{-};
   "20";"21" **[red]@{-};
   "21";"32" **[red]@{-};
   "32";"33" **[red]@{-};
   "1-3";"1-2" **[green]@{-};
   "1-2";"3-8" **[green]@{-};
   "2-7";"0-1" **[green]@{-};
   "0-1";"00" **[green]@{-};
   "00";"11" **[green]@{-};
   "11";"12" **[green]@{-};
   "12";"23" **[green]@{-};
   "2-3";"3-2" **[blue]@{-};
   "3-2";"31" **[blue]@{-};
   "31";"137" **[blue]@{-};
   "128";"812" **[blue]@{-};
   "713";"13" **[blue]@{-};
   "3-3";"13-13" **@{-};
   "12-12";"2-2" **@{-};
   "2-2";"2-1" **@{-};
   "2-1";"8-3" **@{-};
   "7-2";"c32" **@{-};
   "c23";"01" **@{-};
   "01";"03" **@{-};
\endxy
\quad\longrightarrow\quad
 \xy
   (0,0)*{}="00";
   (5,0)*{}="10";
   (10,0)*{}="20";
   (15,0)*{}="30";
   (0,5)*{}="01";
   (5,5)*{}="11";
   (10,5)*{}="21";
   (15,5)*{}="31";
   (0,10)*{}="02";
   (5,10)*{}="12";
   (10,10)*{}="22";
   (15,10)*{}="32";
   (0,15)*{}="03";
   (5,15)*{}="13";
   (10,15)*{}="23";
   (15,15)*{}="33";
   (0,-5)*{}="0-1";
   (5,-5)*{}="1-1";
   (10,-5)*{}="2-1";
   (15,-5)*{}="3-1";
   (0,-10)*{}="0-2";
   (5,-10)*{}="1-2";
   (10,-10)*{}="2-2";
   (15,-10)*{}="3-2";
   (0,-15)*{}="0-3";
   (5,-15)*{}="1-3";
   (10,-15)*{}="2-3";
   (15,-15)*{}="3-3";
   (7,13)*{}="713";
   (8,12)*{}="812";
   (12,8)*{}="128";
   (13,7)*{}="137";
   (7,3)*{}="73";
   (8,2)*{}="82";
   (2,-2)*{}="c2-2";
   (3,-3)*{}="c3-3";
   (7,-7)*{}="7-7";
   (8,-8)*{}="8-8";
   (12,-12)*{}="12-12";
   (13,-13)*{}="13-13";
   "0-3";"0-1" **[red]@{-};
   "0-1";"32" **[red]@{-};
   "32";"33" **[red]@{-};
   "1-3";"1-2" **[green]@{-};
   "1-2";"2-1" **[green]@{-};
   "2-1";"20" **[green]@{-};
   "20";"82" **[green]@{-};
   "73";"11" **[green]@{-};
   "11";"12" **[green]@{-};
   "12";"23" **[green]@{-};
   "2-3";"3-2" **[blue]@{-};
   "3-2";"31" **[blue]@{-};
   "31";"137" **[blue]@{-};
   "128";"812" **[blue]@{-};
   "713";"13" **[blue]@{-};
   "3-3";"13-13" **@{-};
   "12-12";"8-8" **@{-};
   "7-7";"c3-3" **@{-};
   "c2-2";"00" **@{-};
   "00";"03" **@{-};
\endxy
\quad\longrightarrow\quad
 \xy
   (0,0)*{}="00";
   (5,0)*{}="10";
   (10,0)*{}="20";
   (15,0)*{}="30";
   (0,5)*{}="01";
   (5,5)*{}="11";
   (10,5)*{}="21";
   (15,5)*{}="31";
   (0,10)*{}="02";
   (5,10)*{}="12";
   (10,10)*{}="22";
   (15,10)*{}="32";
   (0,15)*{}="03";
   (5,15)*{}="13";
   (10,15)*{}="23";
   (15,15)*{}="33";
   (0,-5)*{}="0-1";
   (5,-5)*{}="1-1";
   (10,-5)*{}="2-1";
   (15,-5)*{}="3-1";
   (0,-10)*{}="0-2";
   (5,-10)*{}="1-2";
   (10,-10)*{}="2-2";
   (15,-10)*{}="3-2";
   (0,-15)*{}="0-3";
   (5,-15)*{}="1-3";
   (10,-15)*{}="2-3";
   (15,-15)*{}="3-3";
   (12,13)*{}="1213";
   (13,12)*{}="1312";
   (7,8)*{}="78";
   (8,7)*{}="87";
   (12,3)*{}="123";
   (13,2)*{}="132";
   (2,-2)*{}="c2-2";
   (3,-3)*{}="c3-3";
   (7,-7)*{}="7-7";
   (8,-8)*{}="8-8";
   (12,-12)*{}="12-12";
   (13,-13)*{}="13-13";
   "0-3";"0-1" **[red]@{-};
   "0-1";"10" **[red]@{-};
   "10";"11" **[red]@{-};
   "11";"33" **[red]@{-};
   "1-3";"1-2" **[green]@{-};
   "1-2";"2-1" **[green]@{-};
   "2-1";"20" **[green]@{-};
   "20";"31" **[green]@{-};
   "31";"32" **[green]@{-};
   "32";"1312" **[green]@{-};
   "1213";"23" **[green]@{-};
   "2-3";"3-2" **[blue]@{-};
   "3-2";"30" **[blue]@{-};
   "30";"132" **[blue]@{-};
   "123";"87" **[blue]@{-};
   "78";"12" **[blue]@{-};
   "12";"13" **[blue]@{-};
   "3-3";"13-13" **@{-};
   "12-12";"8-8" **@{-};
   "7-7";"c3-3" **@{-};
   "c2-2";"00" **@{-};
   "00";"03" **@{-};
\endxy
\]}
\end{itemize}
The Zamolodchikov Tetrahedron equation  means that the above two processes are the same.

\vspace{3mm}

\noindent
{\bf Acknowledgements. } This research is supported by NSFC (12471060, W2412041). We give warmest thanks to Rong Tang for helpful discussions.

\noindent
{\bf Declaration of interests. } The authors have no conflicts of interest to disclose.

\noindent
{\bf Data availability. } Data sharing is not applicable to this article as no new data were created or analyzed in this study.

 \end{document}